\documentclass[12pt, draftclsnofoot, onecolumn]{IEEEtran}
\usepackage{lipsum}
\usepackage{amsmath,epsfig,comment}
\usepackage{amsthm}
\usepackage{tcolorbox}
\usepackage{amssymb}
\usepackage{multirow}
\usepackage{mathrsfs}
\usepackage{amsmath}
\usepackage{pifont}
\usepackage{cases} 
\usepackage{amssymb,amsmath,cite}
\usepackage{epsfig}
\usepackage{color}
\usepackage{bm}
\usepackage{graphicx,subfigure}
\usepackage{algorithm}
\usepackage{algorithmic}
\usepackage{multirow}
\usepackage{soul}
\usepackage{multicol}
\usepackage{cuted}
\usepackage{extarrows}
\usepackage{hyperref}
\hypersetup{
	colorlinks=true,
	linkcolor=black,
	filecolor=magenta,
	urlcolor=cyan,
	citecolor=blue
}

\newcommand\rsx[1]{\left.{#1}\vphantom{\Big|}\right|}

\def\red{\color{black}}

\def\s{\mathbf{s}}

\def\I{\mathcal{I}}
\def\RR{\mathcal{R}}

\def\bU{\mathbf{U}}
\def\bA{\mathbf{A}}
\def\bT{\mathbf{T}}
\def\bX{\mathbf{X}}
\def\bV{\mathbf{V}}
\def\bQ{\mathbf{Q}}

\def\n{\mathbf{n}}
\def\sg{\mathbf{g}}
\def\bD{\mathbf{D}}
\def\bR{\mathbf{R}}
\def\sz{\mathbf{z}}
\def\bs{\mathbf{s}}

\def\I{\mathcal{I}}
\def\R{\mathbb{R}}
\def\C{\mathbb{C}}
\def\H{\mathbf{H}}

\def\bU{\mathbf{U}}
\def\bP{\mathbf{P}}

\def\y{\mathbf{y}}

\def\x{\mathbf{x}}

\def\bB{\mathbf{B}}

\def\bM{\mathbf{M}}

\def\v{\mathbf{v}}

\def\sz{\mathsf{z}}
\def\sZ{\mathsf{Z}}
\def\sg{\mathsf{g}}

\def\qce{q_{\text{\tiny CE}}}
\newtheorem{theorem}{Theorem}

\newtheorem{lemma}{Lemma}
\newtheorem{remark}{Remark}
\newtheorem{definition}{Definition}

\newtheorem{corollary}{Corollary}
\newtheorem{assumption}{Assumption}

\newtheoremstyle{noparens}%
  {}{}%
  {\itshape}{}%
  {\bfseries}{.}%
  { }%
  {\thmname{#1}\thmnumber{ #2}\mdseries\thmnote{ #3}}

\theoremstyle{noparens}

\date{today} 
\begin{document}
%
\title{Asymptotic SEP Analysis and Optimization of Linear-Quantized Precoding in Massive MIMO Systems}
%
%
%

 \author{\IEEEauthorblockN{Zheyu Wu, Junjie Ma,  Ya-Feng Liu, and A. Lee Swindlehurst}
  	\thanks{
		Z. Wu, J. Ma, {and Y.-F. Liu} are with the State Key Laboratory of Scientific and Engineering Computing, Institute of Computational Mathematics and Scientific/Engineering Computing, Academy of Mathematics and Systems Science, Chinese Academy of Sciences, Beijing 100190, China (e-mail: \{wuzy, majunjie, yafliu\}@lsec.cc.ac.cn). A. L. Swindlehurst is with the Center for Pervasive Communications and Computing, University of California, Irvine, CA 92697, USA (e-mail: swindle@uci.edu).}		
  }
\maketitle

\begin{abstract}
A promising approach to deal with the high hardware cost and energy consumption of massive MIMO {\red{transmitters}} is to use low-resolution digital-to-analog converters (DACs) at each antenna element. This leads to a {\red transmission scheme where} the transmitted signals are restricted to {\color{black}a finite set of voltage levels.} 
This paper is concerned with the analysis and optimization of a low-cost quantized precoding strategy, {\red referred to as}  linear-quantized precoding, for a downlink massive MIMO system under Rayleigh fading. In linear-quantized precoding, the signals are first processed by a linear precoding matrix and subsequently quantized component-wise {\red by the DAC}. In this paper, we analyze both the signal-to-interference-plus-noise ratio (SINR) and the symbol error probability (SEP) performances of such linear-quantized precoding schemes in an asymptotic framework where the number of transmit antennas and the number of users {\color{black} grow large} with a fixed ratio. Our results provide a rigorous justification for the heuristic {\color{black} arguments based on the Bussgang decomposition that are} commonly used in prior works. Based on the asymptotic analysis, we further derive the optimal precoder within a class of linear-quantized precoders that includes several popular precoders as special cases. Our numerical results demonstrate the excellent accuracy of the asymptotic analysis for finite systems and the optimality of the derived precoder.
\end{abstract}

\begin{IEEEkeywords}
Massive MIMO, {\red quantized precoding},  linear precoding, asymptotic analysis, Haar random matrix, random matrix theory.
\end{IEEEkeywords}

\section{Introduction}
\IEEEPARstart{M}{assive} multiple-input multiple-output (MIMO)  is a key technology for 5G wireless communication systems. By equipping the base station (BS) with {\color{black}many} antennas, massive MIMO can significantly improve the channel capacity, energy efficiency, and spectral efficiency of wireless communication systems \cite{massivemimo2,massivemimo1,massivemimo3}. Despite the great potential of massive MIMO systems,
 high power consumption and hardware cost are serious practical challenges for their commercial deployment.

 One of the main power-hungry components in {\red a} massive MIMO {\red system} is the digital-to-analog-converter (DAC) \cite{book:DAC}, {\red the number of which} scales {\red linearly} with the number of antennas at the BS. {\color{black}  To reduce circuit complexity and power consumption, low-resolution DACs have been considered for massive MIMO systems \cite{SQUID,WFQ,WSR,frequency,scalar1,scalar2,scalar3}. } Unlike conventional precoding schemes where {\color{black} at the symbol sampling points} the transmitted signals can be freely chosen from a continuous set, only a small finite set of signals can be transmitted to convey information when low-resolution DACs are employed.  The analysis and design of quantized precoding with the use of low-resolution DACs has become an active research topics in recent years \cite{SQUID,WFQ,WSR,frequency,scalar1,scalar2,scalar3}. 

{\red Power ampifiers (PAs) are} another main source of power consumption in massive MIMO systems. To achieve {\red the} highest power efficiency, the PAs need to operate close to saturation, but {\red for continuous-valued signals} this incurs nonlinear distortion and causes difficulties for signals with high peak-to-average power ratio (PAPR) \cite{PAbook}. A popular way to handle such difficulty is to restrict the transmitted signal from each antenna to have the same amplitude \cite{CE2,CE,CE3}, which minimizes the PAPR and enables the employment of the most efficient and cheapest PAs.   Combining such {\red a} constant envelope (CE) constraint with the use of low-resolution DACs motivates a new quantized precoding scheme,  quantized constant envelope (QCE) precoding, where at the {\color{black}symbol sampling points the transmitted signals} are restricted to have a fixed amplitude and their phases are limited to finite values. 
The QCE precoding scheme   has attracted {\red significant research interest} \cite{diversity,QCEMSE1,QCEMSE2,QCEMSE3,QCECI,GEMM,QCE_conference} as it combines the advantage of using low-resolution DACs and energy-efficient PAs. 
In particular, as an extreme case of both QCE and traditional quantized precoding, one-bit precoding  (where one-bit DACs are employed) has been widely and extensively studied \cite{MMSE2,reconsider,MFrate,ZF,ZF2,ZF3,C3PO2, CImodel, PBB,sep2,journal}.  {\red The power efficiency gain can be several dB, and in many cases is sufficient to overcome the loss in fidelity due to the coarse quantization.}

We note here that traditional quantized precoding (without the CE constraint) and QCE precoding are both realized by the use of low-resolution DACs and have a common feature that the transmitted signals are only allowed to be selected from a finite  set.  In the following {\red discussion}, we will refer to them (and possibly other precoding schemes with the finite transmission set feature) collectively as quantized precoding. 
 Existing quantized precoding schemes can be broadly categorized into two classes: \textit{linear-quantized} precoding and \textit{nonlinear} precoding. A linear-quantized precoding scheme\footnote{Note that the overall operation of a linear-quantized scheme is not linear due to the presence of the quantization step, but for convenience we will use this term throughout this paper.} simply quantizes the output of a  linear precoder  \cite{SQUID,WFQ,WSR,frequency,MMSE2,reconsider,ZF,ZF2,ZF3,MFrate}. In contrast, nonlinear precoders do not have this simple structure and are typically obtained by solving appropriate optimization problems \cite{QCEMSE1,QCEMSE2,QCEMSE3,QCECI,GEMM,QCE_conference,C3PO2,CImodel,PBB,sep2,journal,scalar1,scalar2,scalar3,CGMT}.

In this paper, we focus on the analysis and optimization of linear-quantized precoding schemes, which are arguably more practical than the computationally expensive nonlinear schemes. Unlike existing works that focus on either traditional quantized precoding or QCE precoding,  this paper deals with these two types of quantized precoding in a unified framework. 
In what follows, we first give a brief review of related works and then present the main contributions of this paper.

 \subsection{Related Work}\label{section1}

\subsubsection{Linear-quantized precoding}
A direct approach to obtain linear-quantized precoders is to quantize the output of classical linear precoders such as the matched filter (MF) and zero-forcing (ZF) precoders \cite{SQUID}. However, this approach does not take into account the effect of quantization and thus yields precoders that, although simple, are suboptimal in the context of quantized  precoding. Noting this, the authors in \cite{WFQ} characterized the mean square error (MSE) between the desired symbol and the received signal with the presence of low-resolution DACs and proposed {\red the} quantized transmit Wiener filter (TxWFQ) {\red precoder} that minimizes the MSE. 
Under the same setup as \cite{WFQ}, the authors of \cite{WSR} proposed a gradient-based approach to maximize the weighted sum rate of the system. For the one-bit case, the authors  in \cite{MMSE2} proposed a minimum mean square error (MMSE) based precoder. Later, a  higher-rank linear precoder was designed in \cite{reconsider} for a downlink one-bit massive MIMO system, showing superior performance to traditional linear-quantized precoders {\red of} channel rank. We remark here that all existing works on the design of linear-quantized precoding focus on traditional quantized precoding  or the special one-bit case.   To the best of our knowledge, no existing work considers the design of linear-quantized {\red precoding} in the QCE context. 

There are also some works focusing on the performance analysis of linear-quantized schemes, e.g., \cite{SQUID,frequency} for traditional quantized precoding, \cite{MFrate,ZF, ZF2, ZF3} for one-bit precoding, and \cite{CEQ,diversity} for QCE precoding. Specifically,  
 the authors in \cite{SQUID} and \cite{frequency}  derived  lower-bounds on the downlink achievable rates of linear-quantized precoding for a flat-fading and a frequency-selective channel, respectively.   For a one-bit massive MIMO system, \cite{MFrate}  derived  a lower bound on the  achievable rate for MF precoding with estimated channel state information (CSI).  The performance of {\red the} one-bit ZF precoder was investigated in \cite{ZF}, in which a closed-form expression of the symbol error probability (SEP) was derived {\red in} the asymptotic setting where the numbers of transmit antennas and users both tend to infinity with a fixed ratio. The same problem was considered in \cite{ZF2} and \cite{ZF3}, where the input-output correlation relationship was expanded up to third-order instead of first-order as in \cite{ZF}, and the derived SEP expression shows better accuracy than that in \cite{ZF} when the number of users is small  relative to the number of transmit antennas. 

{\color{black} A widely used technique for analyzing the performances of linear-quantized precoding schemes is the \textit{Bussgang decomposition} \cite{Bussgang},  which decomposes a non-linear function of a Gaussian signal as the sum of a linear signal term and an uncorrelated distortion term. We remark that although the Bussgang decomposition is per se rigorous, it is often used in conjunction with various {heuristics} to analyze the performance of linear-quantized precoding. For instance, the distortion term is often treated as a random variable that is  independent of all other random variables in the system. 
Although there is strong numerical evidence that the heuristic treatments can yield accurate predictions (e.g., SEP performance) for large systems \cite{ZF,ZF2,ZF3}, a rigorous analysis of such heuristics in the context of linear-quantized precoding is still lacking. Please refer to Section II-C for a detailed discussion of the Bussgang decomposition technique. }

 
 Beyond the analyses of traditional quantized precoding and one-bit precoding, there are also some preliminary attempts to analyze the performance of QCE precoding. Specifically, \cite{CEQ} studied the statistical properties of the CE quantizer (which models the overall operation of low-resolution DACs and {\red the} CE constraint and generates signals satisfying the QCE constraint) and derived closed-form expressions of the cross-correlation factors between the input and output signals of the CE quantizer. Very recently, the authors in \cite{diversity} considered QCE precoding for a multiple-input single-output (MISO) system and derived the diversity order of {\red the }MF precoder, which characterized how fast the system SEP tends to zero as the signal-to-noise ratio (SNR) grows \cite{diversity2}.

\subsubsection{Nonlinear precoding}
 Besides linear-quantized precoding schemes, various nonlinear precoders based on different criteria have been proposed in recent years, see, e.g., \cite{scalar1,scalar2,scalar3} for traditional quantized precoding and \cite{QCEMSE1, QCEMSE2, QCEMSE3, QCECI, GEMM,QCE_conference} for QCE precoding. There are also  many algorithms designed specifically for one-bit precoding, see, e.g., \cite{C3PO2, CImodel, PBB,sep2,journal}. Nonlinear precoding schemes (especially symbol-level nonlinear precoders) usually have better symbol error rate (SER) {\red performance} than their linear counterparts but {\red their computational complexity is} much higher.
 
  Although there have been {\red substantial progress} in the design of nonlinear precoding schemes,  {\red the performance analysis of nonlinear precoding remains an open problem.} This is because nonlinear precoders are typically solutions to complicated optimization problems {\red without} closed-form expressions. In addition, the discrete nature of the transmitted signals in quantized precoding {\red leads to} discrete constraints in the corresponding optimization problem, which further complicates the analysis. 
 Analyzing and designing nonlinear precoding schemes are beyond the scope of this paper and can be considered as a future work.
 

\subsection{Our Contributions}
In this paper, we analyze the performance of a broad class of linear-quantized precoding schemes for a downlink massive MIMO system. Our results rigorously justify and substantially generalize existing results for MF and ZF based schemes derived using heuristic Bussgang decomposition arguments. 
The main contributions are summarized as follows.
\begin{enumerate}
\item \emph{Statistically equivalent model}:
By exploiting a recursive characterization of the Haar random matrix \cite{HD}, we derive a model that is statistically equivalent to the original system model. The statistically equivalent model is close to a ``signal plus independent Gaussian noise'' form and is more amenable to analysis. 
This step is non-asymptotic and the technique we use {\red may} be applicable to other problems as well.

\item \emph{Asymptotic analysis}:
We further consider the large system limit as in \cite{ZF} and show that in the asymptotic regime, the statistically equivalent model is exactly in a ``signal plus independent Gaussian noise'' form.  This provides a rigorous justification for the heuristic analyzes based on {\red the} Bussgang decomposition. 
We also prove that the signal-to-interference-plus-noise ratio (SINR) and SEP of the original model converge to those of the asymptotic model. 
 Simulations show that the asymptotic results  are  accurate 
  for realistic systems with finite dimensions.
\item \emph{Optimal linear-quantized precoder}: Based on the asymptotic analysis, we derive the optimal linear-quantized precoder that optimizes both the asymptotic SINR and the asymptotic SEP {\red performance}. We show that the optimal linear-quantized precoder is {\red a} regularized ZF (RZF) precoder, whose regularization parameter is determined by the quantization type/level as well as  the system parameters. 
To the best of our knowledge, the optimal RZF precoder derived in this paper is the first linear-quantized precoder {\red applicable to general forms of quantization.}
\end{enumerate}

\subsection{Organization and Notations}
The remaining parts of the paper are organized as follows.
Section \ref{section2} describes the system model and the problem formulation. Some preliminaries for analysis are introduced in Section \ref{section3}.  Section \ref{section4} derives the statistically equivalent model and gives the asymptotic analysis. The optimal linear-quantized precoder is then given in Section \ref{section5}. 
  Simulation results are shown in Section \ref{section6} and the paper is concluded in Section \ref{section7}.

\textit{{\red Notation}}: Throughout the paper, we use {\red the typefaces} $x$, $\x$, $\mathbf{X}$, and $\mathcal{X}$ to
denote scalar, vector, matrix, and set, respectively.  For a vector $\x\in\C^n$, $\x[i_1:i_2]$ denotes a sub-vector of $\x$ consisting of its $i_1$-th to $i_2$-th elements, where $1\leq i_1\leq i_2\leq n$; in particular, $\x[i]$ is the $i$-th entry of $\x$, and  $x_i$ is also used if it does not cause any ambiguity. For a matrix $\mathbf{X}$, $\mathbf{X}[i_1,i_2]$ is the $(i_1,i_2)$-th entry of $\mathbf{X}$.  The operators $\arg(\cdot)$, $\mathcal{R}(\cdot)$, $\mathcal{I}(\cdot)$, $(\cdot)^\dagger$, $(\cdot)^\mathsf{T}$, $(\cdot)^\mathsf{H}$, and  $(\cdot)^{-1}$ return the angle, the real part, the imaginary part, the conjugate, the transpose, the conjugate transpose, and the inverse of their corresponding arguments, respectively. {\color{black}We use $\|\cdot\|$ to denote} the $\ell_2$ norm of the corresponding vector or {\red the spectral norm of the corresponding matrix.  {\color{black}For $m,n\in\mathbb{N}$, we denote the $m\times m$ identity matrix by $\mathbf{I}_m$ and the $m\times n$ matrix of all zero entries by $\mathbf{0}_{m\times n}$. We use  $\text{diag}(x_1,x_2,\dots,x_n)$ to refer to} a diagonal matrix with $\{x_i\}_{i=1}^n$ {\red as} its diagonal entries. {\color{black}We use} $\mathcal{U}(n)$  to denote the set of $n\times n$ unitary matrices over $\C$.  {\color{black} The operators} $\mathbb{E}[\cdot]$, $\text{var}(\cdot)$, and $\mathbb{P}(\cdot)$ return the expectation, the variance, and the probability of their corresponding argument, respectively. {\color{black}For two random variables $X$ and $Y$}, $X\overset{d}{=}Y$ means that they have the same distribution.  {\color{black}We denote almost sure convergence by $\xrightarrow{a.s.}$. We use $\mathcal{C}\mathcal{N}(\mathbf{0},\sigma^2\mathbf{I})$ to denote} the zero-mean circularly symmetric complex Gaussian distribution with covariance matrix $\sigma^2\mathbf{I}$, and $\text{Unif}(\mathcal{S})$ to denote uniform distribution on set $\mathcal{S}$. We reserve the sans serif font (e.g., $\sg$) for vectors with i.i.d. standard Gaussian random variables. Finally, $j$ is the imaginary unit satisfying $j^2=-1$.
\section{System Model and Problem Formulation}\label{section2}

\subsection{Linear-Quantized Precoding}\label{section2A}
Consider the downlink of a multiuser massive MIMO system in which an $N$-antenna BS {\red simultaneously} serves $K$ single-antenna users, {\red where }$K<N$.
The received signals at the users can be modeled as
\begin{equation*}\label{sysmodel}
\y=\H\mathbf{x}+\n,
\end{equation*}
where $\y\in\mathbb{C}^K$ is the received signal vector of the users;  $\mathbf{x}\in\mathbb{C}^N$ is the transmitted signal vector from the BS; $\H\in\C^{K\times N}$ models the channel matrix between the BS and the users, and $\n\in\C^K$ is the additive noise. We assume that the analog-to-digital converters (ADCs) equipped at the user side {\color{black} are ideal and }have infinite resolution and {\red that perfect CSI is available at the BS}. \textcolor{black}{We model the DAC as a quantization function and ignore various practical effects such as glitches, element mismatch, slewing, thermal noise, clipping, etc \cite{DACref1,DACref2,DACref3}}.

\begin{figure}[t]
\includegraphics[scale=0.33]{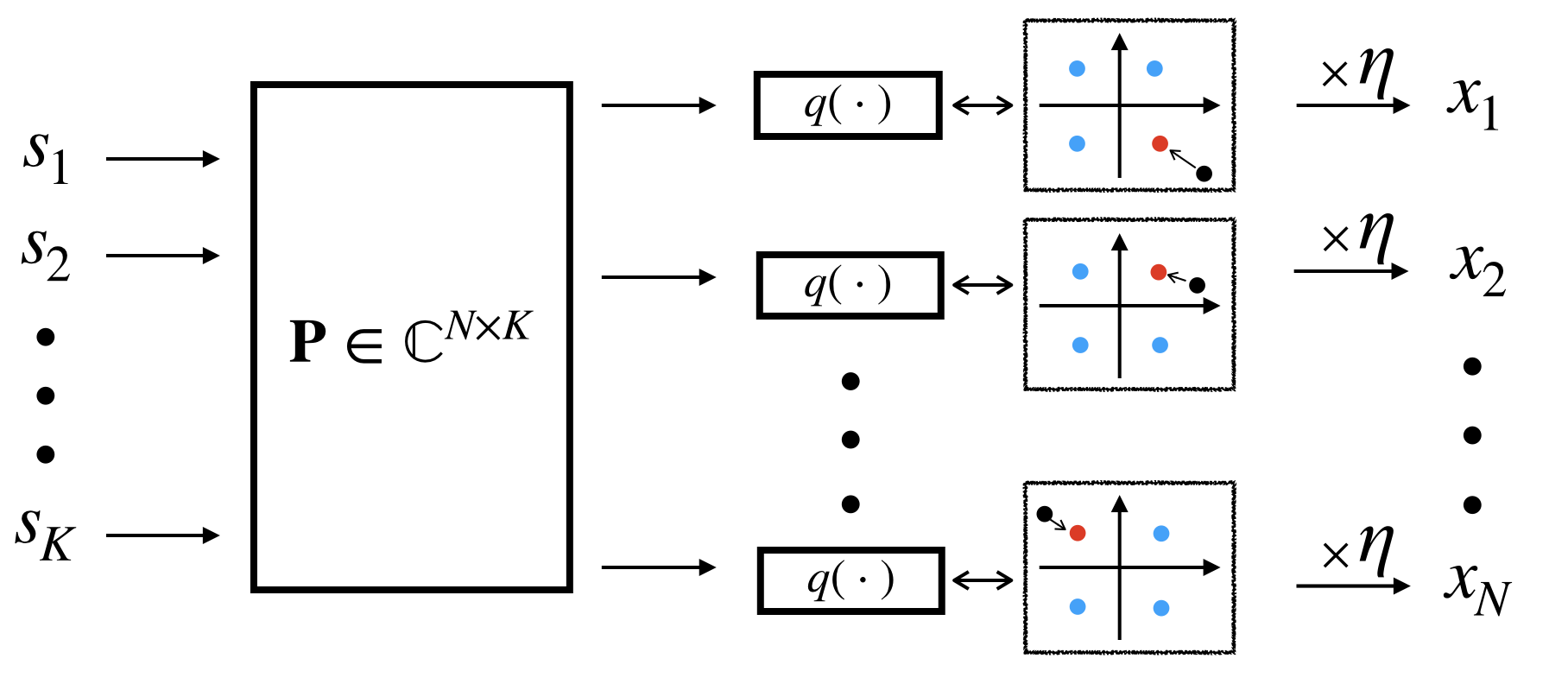}
\centering
\caption{An illustration of the linear-quantized precoding scheme.}
\label{x_figure}
\end{figure}

In this paper, we consider the linear-quantized precoding scheme, where the signal vector to be transmitted at the BS has the following form
\begin{equation}\label{xform}
\x=\eta\, q(\mathbf{Ps}).
\end{equation}
In the above expression, $\mathbf{s}\in\mathbb{C}^K$ is the desired data vector; $\bP\in\C^{N\times K}$ is a precoding matrix;  $q(\cdot): \C\rightarrow \mathcal{X}_L$ is a quantization function that acts component-wise on its input vector, where $\mathcal{X}_L$ is a finite set with $L$ elements and $L$ is referred to as the quantization level; $\eta$ is a scaling factor to ensure that the following average transmit power constraint is satisfied:
 \begin{equation}\label{powerconstraint}
  \frac{1}{N}\mathbb{E}[\|\x\|_2^2]\leq P_T,
  \end{equation} where $P_T>0$ is the maximum average transmit power. {\red S}ee Fig. \ref{x_figure} for an illustration of the linear-quantized precoding scheme. 
  
We now introduce the quantization function corresponding to traditional quantized precoding and QCE precoding, which are most relevant for applications. 
    \begin{itemize}
  \item\emph{Traditional quantized precoding}:   In this case, the real and imaginary parts of the input signal are quantized independently with a pair of low-resolution DACs and $\mathcal{X}_L$ can be expressed as 
    \begin{equation}\label{eq:xl}
  \begin{aligned}
  \mathcal{X}_L&=\left\{x\mid\RR(x),\I(x)\in\left\{\frac{\Delta}{2}\left(2\ell-1-\sqrt{L}\right),~\ell=1,\dots, \sqrt{L}\right\}\right\},
  \end{aligned}
  \end{equation}
  where 
  $\Delta$ is the quantization interval.  {\red The corresponding} quantization function maps its input to the nearest point in \eqref{eq:xl}. {\red In the following,  we call it independent quantizer since the quantizer acts independently on the real and imaginary parts of its input, and  denote it by $q_{\text{\tiny I}}(\cdot)$}. For an $L$-level independent quantizer, the resolution of {\red the} DACs is  $\frac{1}{2}\log_2 L$ bits, where $L{\color{black}\geq 4}$ and is a power of $2$.
  \item \emph{QCE precoding}: In this case, the CE constraint is combined with the use of low-resolution DACs and $\mathcal{X}_L$ has the following expression:
    \begin{equation}\label{eq:xl2}
  \mathcal{X}_L=\left\{e^{j\frac{(2\ell-1)\pi}{L}}\mid~\ell=1,2,\dots,L\right\}. 
  \end{equation}
  The {\red corresponding} quantization function maps its input to the nearest point in \eqref{eq:xl2}. {\red In the following, we call it CE quantizer and denote it by}  $q_{\text{\tiny CE}}(\cdot)$. The resolution of {\red the} DACs is  $\log_2\frac{L}{2}$ bits for an $L$-level CE quantizer, where $L{\color{black}\geq 4}$ and is a power of $2$.  \end{itemize}
  Note that when $L=4$ and $\Delta=2$, the independent quantizer and the CE quantizer are the same, both reducing to the one-bit precoding case.


Let $\H=\bU\bD\bV^\mathsf{H}$ be the singular value decomposition (SVD) of $\H$, where $\bU\in\mathcal{U}(K), \bV\in\mathcal{U}(N)$, and $\bD=\left(\begin{matrix}\text{diag}{\red (}d_1,d_2,\dots,d_K{\red)}&\mathbf{0}_{{\color{black}K\times (N-K)}}\end{matrix}\right)\in\R^{K\times N}$ {\color{black}with $d_1,d_2,\dots, d_K$ representing the non-zero singular values\footnote{\red We assume throughout the paper that $\H$ is of full row rank. This holds with probability one if Assumption \ref{Ass:distribution} further ahead is satisfied, i.e., if the entries of $\H$ are i.i.d. following $\mathcal{CN}(0,\frac{1}{N})$.} of $\H$.}
In this paper, we focus on precoding matrices with the following structure:  
\begin{equation}\label{Eqn:P}
\bP=\bV f(\bD)^\mathsf{T}\bU^\mathsf{H},
\end{equation}
where $f(\cdot)$ acts independently on  the nonzero singular values of $\H$, i.e., 
\begin{equation*}
f(\bD)=\left(\begin{matrix}\text{diag}{\red(}f(d_1),f(d_2),\dots,f(d_K){\red )}& \mathbf{0}_{{\color{black}K\times (N-K)}}\end{matrix}\right).
\end{equation*}
See Fig. \ref{P} for an illustration of the structure of $\bP$.  
\begin{figure}[t]
\includegraphics[scale=0.35]{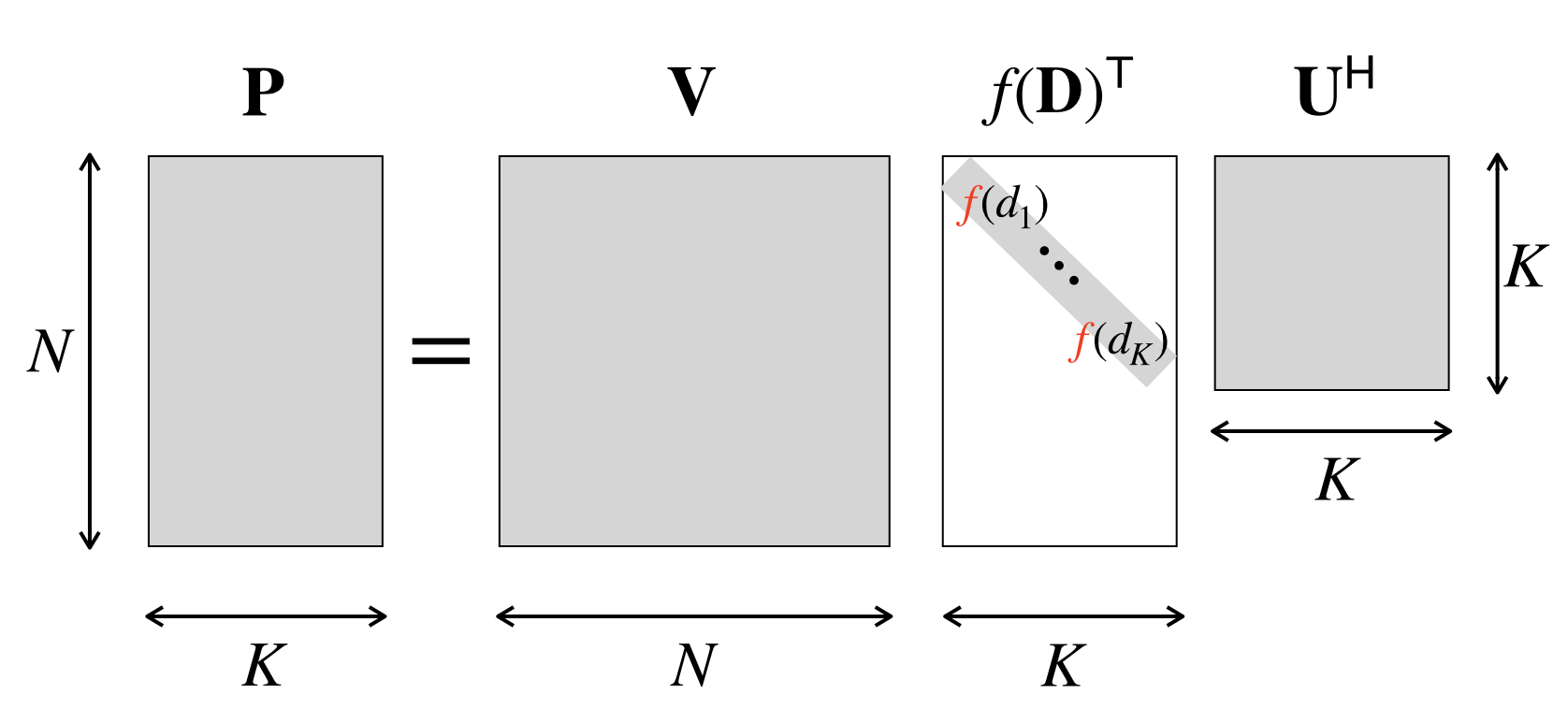}
\centering
\caption{The structure of the precoding matrix in consideration.}
\label{P}
\end{figure}
The motivations to consider the special class of the precoding matrix in \eqref{Eqn:P} are twofold. First, as will be shown in Section \ref{section3}, the structure of $\mathbf{P}$ in \eqref{Eqn:P} enables us to apply existing results in random matrix theory for performance analysis. Second, the structure of $\mathbf{P}$ is fairly general and includes the following popular precoders as special cases:
\vspace{5pt}
\begin{itemize}
\item \textit{MF}: $\bP=\H^\mathsf{H}$, which corresponds to $f(x)=x$ in \eqref{Eqn:P};
\item \textit{ZF}: $\bP=\H^\mathsf{H}(\H\H^\mathsf{H})^{-1}$, which corresponds to $f(x)=x^{-1}$ in \eqref{Eqn:P};
\item  \textit{RZF}: $\bP=\H^\mathsf{H}(\H\H^\mathsf{H}+\rho\mathbf{I}_{\color{black}K})^{-1}$, which corresponds to $f(x)=\frac{x}{x^2+\rho}$ in \eqref{Eqn:P}.
\end{itemize}

\vspace{5pt}

With the above linear-quantized precoding scheme, the received signals at the user side read
\begin{equation}\label{sysmodel1}
\y=\eta\H q(\bP \s)+\n=\eta \bU\bD\bV^\mathsf{H} q(\bV f(\bD)^\mathsf{T}\bU^\mathsf{H}\s)+\n.
\end{equation}
 As in \cite{QCEMSE2, QCEMSE3, QCECI, GEMM}, we assume that each user is able to rescale the received signal $y_k$ by a factor $\beta_k\in\C$, i.e., $r_k=\beta_k y_k$. 
(This corresponds to removing the effective channel gain.) After the rescaling step, the users employ symbol-wise nearest-neighbor decoding, i.e., each user $k$ maps $r_k$ to the nearest constellation point.

\vspace{3pt}

As will be shown {\red below}, the nonlinear function $f(\cdot)$ in $\mathbf{P}$ has a major impact on the performance of the overall system. In this paper, we will first analyze the performance of the linear-quantized scheme in the asymptotic regime where $N$ and $K$ tend to infinity simultaneously, and then optimize $f(\cdot)$ based on the asymptotic analysis.


\subsection{Assumptions}




In this subsection, we specify our assumptions on the system model in \eqref{sysmodel1}. We first make a few standard assumptions on $\mathbf{H}$, $\mathbf{n}$, and $\mathbf{s}$.

\begin{assumption}\label{Ass:distribution}
The entries of $\mathbf{H}$ and $\mathbf{n}$ are independently drawn from $\mathcal{C}\mathcal{N}\left(0,\frac{1}{N}\right)$ and $\mathcal{C}\mathcal{N}\left(0,\sigma^2\right)$, respectively. The entries of $\mathbf{s}$ are independently and uniformly drawn from a finite set $\mathcal{S}_M$ with nonzero elements (i.e.,  $0\notin \mathcal{S}_M$), and $\mathbb{E}[|s_1|^2]=\sigma_s^2$. Furthermore, $\mathbf{H},\mathbf{s},$ and $\mathbf{n}$ are mutually independent.
\end{assumption}

{\color{black} The i.i.d. Gaussian assumption on the channel $\H$ is widely adopted in the massive MIMO literature for ease of analysis, see, e.g., \cite{ZF,ZF2,ZF3}. This assumption is reasonable in a rich scattering environment where the number of scattered components is large and independent. Such a scenario arises when the antennas are widely spaced or when the physical environment exhibits scattering in all directions \cite{book:Tse}.  }
Note that we have assumed ${H}_{ij}\sim\mathcal{CN}(0,\frac{1}{N})$ instead of $H_{ij}\sim\mathcal{CN}(0,1)$. This normalization is introduced as in \cite{asymp1,asymp2,assumption1} to {\red ensure that the received power of the users does not grow with $N$}. {\color{black}We would like to emphasize that the i.i.d. Gaussian assumption on the channel $\H$ is not essential to our analysis. Our results can be extended to a broader class of channel models, as discussed in Remark \ref{rem:otherd} below Theorem \ref{asy}. }
 The assumption on $\s$ is quite general and is satisfied by common constellation schemes including phase shift keying (PSK) and quadrature amplitude modulation (QAM).

For technical reasons, we impose the following assumption on the nonlinear function $f(\cdot)$ in {\red $\mathbf{P}$ in \eqref{Eqn:P}}.

\begin{assumption}\label{Ass:f}
{\color{black}The function $f(\cdot)$ is positive, continuous almost everywhere (a.e.), and bounded on any compact set of $(0,\infty)$.}
\end{assumption}
{\color{black}Notice that the $f$ functions corresponding to the MF, ZF, and RZF precoders discussed in the previous subsection all satisfy Assumption \ref{Ass:f}.}
We emphasize that the positivity assumption on $f$ is not essential and can be relaxed to $\mathbb{P}\left(f(\bD)=\mathbf{0}\right)=0$. Further, as will be shown in Section \ref{section5}, there {\red always} exists an optimal precoder satisfying $f>0$, implying that the positivity  assumption does not impose any restriction in terms of the best achievable performance. 

Finally, we assume that the quantization function $q$ in \eqref{xform} satisfies some regularity conditions, as stated in Assumption \ref{Ass:q} below. 

\begin{assumption}\label{Ass:q}
{\color{black} The quantization function} $q: \C\rightarrow \mathcal{X}_L$ is continuous  a.e. and bounded.
\end{assumption}

It is straightforward to verify that the independent quantizer $q_{\text{\tiny I}}(\cdot)$ and the CE quantizer $\qce(\cdot)$ both satisfy  Assumption \ref{Ass:q} (as they are piecewise constant). We emphasize that some of our results can be simplified for QCE precoding, i.e., when $q(\cdot)=\qce(\cdot)$. In the following, we will first present our results in the most general form and then discuss the case of QCE precoding {\red separately}.
\subsection{Heuristic Analysis Via Bussgang Decomposition}\label{bussgangsection}

The nonlinear quantization function $q(\cdot)$ causes some difficulties for performance analysis. A popular technique to deal with it is {\red the} Busggang decomposition \cite{MFrate, SQUID, frequency, ZF, ZF2, ZF3}, 
which decomposes a nonlinear function of Gaussian random variables into a linear signal term and an uncorrelated nonlinear distortion term. We now outline a \textit{heuristic} analysis of the problem using the Bussgang decomposition technique.


Consider the nonlinear quantization process $q(\bP\s)$,  where $\mathbf{P}=\mathbf{V}f(\mathbf{D})^\mathsf{T}\mathbf{U}^{\mathsf{H}}$. Under Assumption \ref{Ass:distribution},  $\mathbf{V}$ and $\mathbf{U}$ are Haar distributed (see Definition 1 and Lemma \ref{Lem:SVD_Haar} further ahead) and it can be shown that $\mathbf{Ps}$ is approximately distributed as
\[
\mathbf{Ps}\overset{d}{\approx}\mathcal{CN}(\mathbf{0},\bar{\alpha}^2\mathbf{I}_{\color{black}N}),
\]
where
\[
\begin{split}
\bar{\alpha}^2 &= \frac{1}{N}\mathbb{E}[\|\mathbf{Ps}\|^2].
\end{split}
\]
Based on the Busggang decomposition technique \cite{Bussgang}, we can write $q(\mathbf{Ps})$ as
\[
q(\mathbf{Ps}) = {\red \overline{C}_1 (\mathbf{Ps})}+\mathbf{q}_\perp,
\]
where $ \overline{C}_1=\mathbb{E}[\sZ^\dagger q(\bar{\alpha}\sZ)]/\bar{\alpha}$, $\sZ\sim\mathcal{CN}(0,1)$,  and $\mathbf{q}_\perp$ is the residual nonlinear distortion which is approximately orthogonal to $\mathbf{Ps}$. Substituting this decomposition into \eqref{sysmodel1} gives
\begin{equation}\label{Eqn:Baussgang}
\begin{split}
\y &=\eta\,\H q(\bP \s)+\n\\
&= \eta\,\overline{C}_1 \H \mathbf{Ps}+\eta\,\H\mathbf{q}_\perp+\mathbf{n}\\
&= \eta\frac{ \overline{C}_1}{K}\text{tr}(\mathbf{HP})\mathbf{s} +\eta\,\overline{C}_1\left(\mathbf{HP}-\frac{1}{K}\text{tr}(\mathbf{HP})\mathbf{I}\right)\mathbf{s}+\eta\H\mathbf{q}_\perp+\mathbf{n}.\\
\end{split}
\end{equation}
In the above decomposition, the first term is a signal term and the last three terms are effective noise terms. This demonstrates the advantage of the Bussgang decomposition: it can transform a nonlinear system into a linear one so that the useful signal and the effective noise can be distinguished. However, the problem is that the distribution of $\mathbf{q}_\perp$ and its correlation with $(\H,\s)$ are hard to characterize, which makes it still highly non-trivial to analyze the performance (e.g., SEP performance) of the system with \eqref{Eqn:Baussgang}. Heuristically, one may treat $\mathbf{q}_{\perp}$ as if it is independent of both $\mathbf{H}$ and $\mathbf{s}$, so that $\mathbf{Hq}_\perp$ can be approximated {\red as} independent Gaussian noise. It turns out that this treatment, though heuristic, leads to very accurate predictions \cite{ZF}. Developing a new analytical framework that can rigorously justify the above heuristics is a main motivation behind this work.

%
\vspace{-0.5cm}
\section{{\color{black}Preliminaries}}\label{section3}

Our analysis is based on  Householder dice \cite{HD}, a technique for recursively generating Haar random matrices. Before presenting our main results, {\red in this section} we first give some preliminaries on the Haar random matrix and the Householder dice technique.

We begin with the definition of the Haar measure and the Haar random matrix.

\begin{definition}[Haar measure \cite{Haarbook}]
{\red The Haar measure on $\mathcal{U}(N)$ is defined as the unique probability measure $\mu$ on $\mathcal{U}(N)$ that satisfies the following translation invariant property: for any measurable subset $\mathcal{A}\subset\mathcal{U}(N)$ and any fixed $\bM\in\mathcal{U}(N)$, 
$$\mu(\bM\mathcal{A})=\mu(\mathcal{A}\bM)=\mu(\mathcal{A}),$$
where $\bM\mathcal{A}$ denotes the set obtained by taking all the elements of $\mathcal{A}$ and multiplying them by $\bM$.}


\end{definition}
\hspace{-0.35cm}In the following, we denote by $\text{Haar}(N)$ the ensemble of random unitary matrices drawn from the Haar measure on $\mathcal{U}(N).$

Lemma \ref{Lem:SVD_Haar} below is a well-known fact in random matrix theory \cite{tulino2004random} and suggests the crucial role of the Haar random matrix {\red plays} in our analysis. 

\begin{lemma}\label{Lem:SVD_Haar}
Let $\mathbf{H}=\mathbf{UDV}^\mathsf{H}$ be the SVD of $\mathbf{H}$. Under Assumption \ref{Ass:distribution}, $\mathbf{U},\mathbf{D},\mathbf{V}$ are mutually independent and $\mathbf{U},\mathbf{V}$ are Haar distributed random matrices.
\end{lemma}



{\color{black}We now introduce the Householder dice (HD) technique proposed in \cite{HD}, which deals with an iterative process involving Haar matrices as follows:
\begin{equation}\label{HDiter}
\x_{t+1}=f_t(\bQ\x_t),~0\leq t\leq T-1,
\end{equation}
where $\bQ\sim\text{Haar}(N)$ and $\x_0\in\C^N$ are independent.  HD was originally proposed as an efficient numerical method for simulating iterations like \eqref{HDiter} of large dimension.  In our paper, we employ it as a tool for performance analysis, which is a novel application of this technique.  
Specifically, using HD, one can show that the sequence $\{\x_0,\x_1,\x_2,\dots,\x_T\}$ generated by \eqref{HDiter} is statistically equivalent to another sequence that is fully determined by the initial vector $\x_0$ and a sequence of independent standard Gaussian random vectors. Compared to  the original sequence that exhibits complicated correlation through the Haar matrix $\bQ$, the new sequence is more amenable to analysis, particularly in the high-dimensional case, which enables us to study the statistical properties of the original sequence with greater ease. 
 
To get some insight on how HD facilitates analysis, we consider the following simple example that contains only two iterations:
\begin{equation}\label{example}
\left\{
\begin{aligned}
\x_1=f_0(\bQ\x_0),\\
\x_2=f_1(\bQ\x_1),
\end{aligned}
\right.
\end{equation}
where $\bQ\sim\text{Haar}(N)$ and $\x_0\in\C^N$ are independent. Using HD, one can show that $(\x_0,\x_1,\x_2)$ is statistically equivalent to $(\x_0,\tilde{\x}_1,\tilde{\x}_2)$ given below:
\begin{equation}\label{example2}
\left\{
\begin{aligned}
\tilde{\x}_1&=f_0(a_0^1\sg_1),\\
\tilde{\x}_2&=f_1(a_1^1\sg_1+a_1^2\sg_2),
\end{aligned}
\right.
\end{equation}
where $\sg_1\sim\mathcal{CN}\left(\mathbf{0},\mathbf{I}_N\right)$ and $\sg_2\sim\mathcal{CN}\left(\mathbf{0},\mathbf{I}_N\right)$ are independent and both are independent of $\x_0$, and the random variables $\{a_0^1, a_1^1, a_1^2\}$ are defined by 
\begin{equation*}
\begin{aligned}
a_0^1&=\frac{\|\x_0\|}{\|\sg_1\|},\\
a_1^1&=\frac{\x_0^\mathsf{H}\tilde{\x}_1}{\|\sg_1\|}-\frac{\sg_1^\mathsf{H}\sg_2}{\|\sg_1\|\|\x_0\|}\frac{\sqrt{\|\tilde{\x}_1\|^2\|\x_0\|^2-\left|{\x_0^\mathsf{H}\tilde{\x}_1}\right|^2}}{\sqrt{\|\sg_1\|^2\|\sg_2\|^2-|\sg_1^\mathsf{H}\sg_2|^2}},\\
a_1^2&=\frac{\|\sg_1\|\sqrt{\|\tilde{\x}_1\|^2\|\x_0\|^2-\left|{\x_0^\mathsf{H}\tilde{\x}_1}\right|^2}}{\|\x_0\|\sqrt{\|\sg_1\|^2\|\sg_2\|^2-|\sg_1^\mathsf{H}\sg_2|^2}}.
\end{aligned}
\end{equation*}
The statistical equivalence between \eqref{example} and \eqref{example2}  can be proved using a  technique similar to that in \cite[Section 3.3]{HD} and thus we omit the details here. Clearly, $\tilde{\x}_1$ and $\tilde{\x}_2$ are fully specified by the initial vector $\x_0$ and the two Gaussian vectors $\sg_1$ and $\sg_2$. The scaling factors $\{a_0^1,a_1^1, a_1^2\}$, though correlated with $\{\x_0, \sg_1, \sg_2\}$ in a complicated way, converge in many cases to deterministic values as the matrix dimension $N$ tends to infinity. For instance, when $f_0(\cdot)$ is separable and satisfies some mild regularity conditions and the entries of $\x_0$ are i.i.d., the convergence of $\{a_0^1,a_1^1,a_1^2\}$ can be easily proved via the law of large numbers. 

The above example illustrates the strength of the HD technique: it transforms the original sequence, which is specified by the $N\times N$ Haar matrix $\bQ$, into another sequence that is determined by only a few Gaussian vectors (e.g., two Gaussian vectors of dimension $N$ for the above example)  with an explicit form. The new sequence is usually more convenient for analysis. 
The interested reader is referred to  \cite{HD} for a detailed description of the HD technique. We will provide a comprehensive description of the HD technique for handling our specific problem in Appendix \ref{Appendix:HD} and Appendix \ref{derivy}. 
}

\section{Statistically Equivalent Model And Asymptotic Analysis}\label{section4}
In this section, we first use the HD technique to derive a statistically equivalent model for \eqref{sysmodel1}, which is close to a ``signal plus independent Gaussian noise'' form. {\color{black}This step is non-asymptotic and the equivalence holds for any finite dimension when $N,K\geq 3$.} The ``signal plus Gaussian noise'' insight is made precise by further considering the large system limit where $N$ and $K$ tend to infinity at a fixed ratio. We will derive sharp asymptotic expressions for the SINR and SEP performances of the linear-quantized precoding scheme. 


\vspace{-0.2cm}
\subsection{Statistically Equivalent Model}
Recall that our system model is 
$$\y=\eta\H q(\bP\s)+\n=\eta\bU\bD\bV^\mathsf{H}q(\bV f(\bD)^\mathsf{T} \bU^\mathsf{H}\s)+\n, $$
where $\bU\sim \text{Haar}(K), \bV\sim\text{Haar} (N),$ and $\{\bU,\bV,\bD,\s,\n\}$ are mutually independent. The received signal $\y$ can be seen as {\red being} obtained by performing the following iterations:
\begin{equation}\label{Eqn:model_s1s3}
\left\{
\begin{aligned}
\s_1&=f(\bD)^\mathsf{T}\bU^\mathsf{H}\s,\\
\s_2&=q(\bV \s_1),\\
\s_3&=\bD\bV^\mathsf{H}\s_2,\\
\y&=\eta\bU\s_3+\n,
\end{aligned}\right.
\end{equation}
 {\color{black}The above iterative process has a form similar to \eqref{HDiter}. At each iteration, it involves one multiplication of a Haar random matrix and a random vector, while the other operations can be modeled as $f_t(\cdot)$ in \eqref{HDiter}, since $\{\bD,\n\}$ are independent of $\{\bU,\bV,\bs\}$. Specifically,  $f_0(\x)=f(\bD)^\mathsf{T}\x, f_1(\x)=q(\x), f_2(\x)=\bD\x, f_3(\x)=\eta \x+\n$, where $\x$ is a vector of an appropriate dimension. The only minor difference with \eqref{HDiter} is that  two different Haar matrices and their conjugate transposes are included in the above iterations. However, this difference is not important and the HD technique can still be applied. }With the help of the HD technique, we can obtain the following statistically equivalent model, which is more convenient for analysis.

\begin{theorem}[Statistically Equivalent Model]\label{Equimodel} {\color{black}When $N\geq3, K\geq3$},
the distribution of $(\y,\s)$ in the original model \eqref{sysmodel1} is the same as {\color{black}that of} $(\hat{\y},\s)$ specified by the following model:
 \begin{equation}\label{Equiy}
\begin{aligned}
 \hat{\y}=\eta T_\s\,\s+\eta T_\sg\, \sg_2+\n,
\end{aligned}
\end{equation}
where
\begin{equation}\label{Equiy_2}
\begin{aligned}
T_\s=&\frac{\sg_1^\mathsf{H}\{C_1\bD\hat{\bs}_1+C_2\bD \bB(\hat{\bs}_1)\sz_2[2:N]\}}{\|\sg_1\|\|\s\|}-T_\sg\,\frac{(\bR(\bs)^{-1}\sg_2)[1]}{\|\bs\|},\\
T_\sg=&\frac{\|\bB(\sg_1)^\mathsf{H}\{C_1\bD \hat{\bs}_1+C_2\bD \bB(\hat{\bs}_1)\sz_2[2:N]\}\|}{\|(\bR(\bs)^{-1}\sg_2)[2:K]\|},\\
C_1\hspace{-0.05cm}=&\frac{\sz_1^\mathsf{H}q\left(\frac{\|\hat{\s}_1\|}{\|\sz_1\|}\,\sz_1\right)}{\|\hat{\bs}_1\|\|\sz_1\|},~ C_2=\frac{\left\|\bB(\sz_1)^\mathsf{H}q\left(\frac{\|\hat{\s}_1\|}{\|\sz_1\|}\,\sz_1\right)\right\|}{\|\sz_2[2:N]\|},\\
 \hat{\s}_1=&\frac{\|\s\|}{\|\sg_1\|} f(\bD)^{\mathsf{T}}\sg_1.
 \end{aligned}
 \end{equation} 
 In the above expressions, $\sg_1{\color{black}\sim\mathcal{CN}\left(\mathbf{0},\mathbf{I}_K \right)}, \sg_2\sim\mathcal{CN}\left(\mathbf{0},\mathbf{I}_K\right), \sz_1\sim\mathcal{CN}\left(\mathbf{0},\mathbf{I}_N\right), \sz_2\sim\mathcal{CN}\left(\mathbf{0},\mathbf{I}_N\right)$ are mutually independent standard Gaussian random vectors, which are further independent of the signal vector $\s$, the singular value matrix $\bD$, and the noise vector $\n$; $f(\cdot)$ is a processing function involved in the precoder \eqref{Eqn:P}; $\bR(\cdot)$ denotes the Householder transform of the input vector and $\bB(\cdot)$ represents the submatrix of $\bR(\cdot)$ with the first column removed (see   \eqref{reflector} and \eqref{B} in Appendix \ref{Appendix:HD}).
\end{theorem}
\begin{proof}
See Appendix \ref{derivy}. 
\end{proof}

Before proceeding, let us take a look at the statistically equivalent model in \eqref{Equiy}:
 the first term is the (scaled) signal vector, the second term is an equivalent noise that captures both the multi-user interference and the distortion caused by quantization, and the last term is the channel noise. 
Note, however, it is still difficult to {\red exactly} analyze the performance of the system (e.g., SEP performance)  based on the {\red statistically} equivalent model in \eqref{Equiy}. This is because $T_\s$ and $T_\sg$ therein are correlated with $\bs$ and $\sg_2$ in a complicated way. Fortunately, as $ N, K\rightarrow \infty$ and $N/K\rightarrow\gamma \in(1,\infty)$, both $T_\bs$ and $T_\sg$ converge to deterministic quantities, enabling us to derive sharp asymptotic formulas for both the SINR and SEP {\red performance}.
  
 \vspace{-0.2cm}
 \subsection{Asymptotic Analysis}\label{Sec:Asymptotic}
 \vspace{-0.1cm}
 In this subsection, we consider the large system limit where both $N$ and $K$ tend to infinity while keeping a finite ratio $\frac{N}{K}\rightarrow\gamma\in(1,\infty).$ This is a common assumption in the performance analysis of massive MIMO systems, and such asymptotic {\red analyses} can usually provide tight approximations for realistic systems with finite $N, K$ (see, e.g., \cite{asymp1,asymp2,assumption1, assumption2}). 
 In what follows, all vectors and matrices should be understood as sequences of vectors and matrices of growing dimensions. For simplicity, their dependence on $N$ and $K$ is not explicitly shown. 

Our main asymptotic result is summarized in the following theorem. Its proof is given in Appendix \ref{asymptotic}.

{\begin{theorem}[Asymptotic Model]\label{asy}
Define the following asymptotic model:
\begin{equation}\label{Asympmodel}
\bar{\y}:=\eta\overline{T}_\bs \, \s +\eta\overline{T}_\sg \, \sg_2+ \n,
\end{equation}
where
\begin{equation}\label{Asympmodel2}
\begin{split}
\overline{T}_\bs&=\overline{C}_1\,\mathbb{E}[d\, f(d)],\\
\overline{T}_\sg&=\sqrt{\sigma_s^2\,|\overline{C}_1|^2 \text{\normalfont{var}}[{d}\, f({d})]+\overline{C}_2^2},\\
\overline{C}_1&= \frac{\mathbb{E}[\sZ^\dagger q(\bar{\alpha} \sZ)]}{\bar{\alpha}},\\
\overline{C}_2&=\sqrt{\mathbb{E}[|q(\bar{\alpha} \sZ)|^2]-|\mathbb{E}[\sZ^\dagger q(\bar{\alpha}\sZ)]|^2},\\
\bar{\alpha}&=\sqrt{\frac{\sigma_s^2\,\mathbb{E}[f^2(d)]}{\gamma}},\\
\end{split} 
\end{equation}
$\sZ\sim\mathcal{C}\mathcal{N}(0,1),$ $d=\sqrt{\lambda}$, and $\lambda$ follows the Marchenko-Pastur distribution, whose probability density function is given by 
\begin{equation}\label{MPdistribution}
p_{\lambda}(x)=\frac{\sqrt{(x-a)_+(b-x)_+}}{2\pi cx}
\end{equation} with 
$a=(1-\sqrt{c})^2, b=(1+\sqrt{c})^2, c=\frac{1}{\gamma}; (x)_+=\max\{x,0\}.$ Then under Assumptions 1-3, the following holds as $N, K\to\infty$, and $\frac{N}{K}\to\gamma\in(1,\infty)$:
\[
(\hat{y}_k,s_k)\xrightarrow{a.s.}(\bar{y}_k,s_k),\quad \forall\, k\in[K],
\]
where $(\hat{y}_k,s_k)$ and $(\bar{y}_k,s_k)$  are given in \eqref{Equiy} and \eqref{Asympmodel}, respectively.
\end{theorem}
\vspace{5pt}}

{\color{black}Several remarks on Theorems 1 and 2 are in order}. 
\begin{remark}
{\color{black}It is worth noting that the conditioning technique developed in \cite{VAMP}\cite{VAMP:model} may be used to derive the asymptotic model in Theorem \ref{asy}, though the model analyzed in the current paper is more complicated than that in \cite{VAMP:model}; compare \eqref{sysmodel1} and \cite[Eq.~(15)]{VAMP:model}. Note that both the HD technique and the conditioning method of \cite{VAMP}\cite{VAMP:model} heavily rely on the rotational invariance of the Haar random matrix. For our problem, the HD technique is  more direct and transparent.}
\end{remark}
\begin{remark}[Connection with Bussgang decomposition]
The asymptotic model \eqref{Asympmodel} gives a precise characterization of the Bussgang decomposition {\red and} results in \eqref{Eqn:Baussgang}:
\begin{equation}\label{Eqn:Baussgang_temp}
\begin{split}
\y &= \eta\frac{ \overline{C}_1}{K}\text{\normalfont{tr}}(\mathbf{HP})\mathbf{s} +\eta\,\overline{C}_1\left(\mathbf{HP}-\frac{1}{K}\text{\normalfont{tr}}(\mathbf{HP})\mathbf{I}\right)\mathbf{s}+\eta\H\mathbf{q}_\perp+\mathbf{n}.\\
\end{split}
\end{equation}
Loosely speaking, we have the following correspondence between \eqref{Asympmodel} and \eqref{Eqn:Baussgang_temp}: 
\[
\begin{split}
\eta\overline{T}_\bs  \,\s  &\  \leftrightarrow \  \eta\frac{ \overline{C}_1}{K}\text{\normalfont{tr}}(\mathbf{HP})\mathbf{s}\\
\eta \overline{T}_\sg \, \sg_2 & \ \leftrightarrow \eta\,\overline{C}_1\left(\mathbf{HP}-\frac{1}{K}\text{\normalfont{tr}}(\mathbf{HP})\mathbf{I}\right)\mathbf{s}+\eta\H\mathbf{q}_\perp \\
\end{split}
\]
The above correspondence is in a distributional sense, i.e., {\red the corresponding terms have the same (asymptotic) distribution}, 
and is implied by our proof of Theorem \ref{asy}.
\end{remark}

\begin{remark}[Assumption on channel $\H$] \label{rem:otherd}
{\color{black}Theorems 1 and 2 are stated under the assumption that $\H$ is i.i.d. Gaussian, but can be extended to the following more general models.
\begin{itemize}
\item (Unitarily invariant model) For the unitarily invariant model, the SVD of $\H$ satisfies $\bU\sim\text{Haar}(K), \bV\sim\text{Haar}(N)$, and $\{\bU,\bD,\bV\}$ are independent  \cite{tulino2004random},  and hence Theorem 1 holds. In addition, Theorem 2 holds as long as the empirical spectral distribution (see Definition \ref{def:esd} in Appendix \ref{asymptotic}) of $\H\H^\mathsf{H}$ further has a continuous limiting distribution with a bounded support, with $\lambda$ in Theorem 2 following the limiting e.s.d. of $\H\H^\mathsf{H}$.
\item (Large scale fading) In the case where the users have different large scale fading variances, the channel can be modeled as $\H=\boldsymbol{\Sigma}_0\tilde{\H}$, where $\tilde{\H}$ satisfies Assumption 1 and $\boldsymbol{\Sigma}_0=\text{\normalfont{diag}}(\sigma_1, \sigma_2, \dots, \sigma_K)$ is a diagonal matrix with $\sigma_k$ representing the standard deviation of the large scale fading of user $k$, $k=1,2,\dots,K$.  Let $\tilde{\H}=\bU\bD\bV^\mathsf{H}$ be the SVD of $\tilde{\H}$.  Theorems 1 and 2 can be directly generalized to this more general channel model with precoding matrices of the following form:
\begin{equation*}
\bP=\bV f(\bD)^\mathsf{T}\bU^\mathsf{H}g(\boldsymbol{\Sigma}_0),
\end{equation*}
where $\boldsymbol{\Sigma}_0$ and $g$ satisfy some mild regularity conditions. The above class of precoding matrices still includes  MF precoding and ZF precoding as special cases. 
\end{itemize}}
\end{remark}



Theorem \ref{asy} shows that in the asymptotic regime, the system model is in a simple ``signal plus independent Gaussian noise'' form as \eqref{Asympmodel}.  
In the following, we will characterize the individual performance of the $K$ users with the help of Theorem \ref{asy}.
Two commonly used performance measures are the SINR: 
\begin{equation}\label{def:sinr}
{\text{\normalfont{SINR}}}_k:=\frac{|\rho_k|^2\,\mathbb{E}[|s_k|^2]}{\mathbb{E}[|y_k|^2]-|\rho_k|^2\,\mathbb{E}[|s_k|^2]},
\end{equation} where $ \rho_k=\mathbb{E}[s_k^\dagger y_k]/\mathbb{E}[|s_k|^2]$, and the SEP:
 \begin{equation}\label{def:sep}{\text{\normalfont{SEP}}}_k(\beta):=\mathbb{P}\left(\text{dec}(\beta y_k)\neq s_k\right),\end{equation} where $\text{dec}(\cdot)$ is the decision function that maps its argument to the nearest constellation point in $\mathcal{S}_M$. 
 
We note that the SINR and SEP performance of the $K$ users {\red in} the asymptotic model \eqref{Asympmodel} are the same and can be   characterized by the following {\red asymptotic scalar} model:
\begin{equation}\label{scalarmodel}
\bar{y}:=\eta\overline{T}_\s\, s+\eta\overline{T}_\sg\, \sg +n,
\end{equation}
where $s\sim\text{Unif} (\mathcal{S}_M),~\sg\sim\mathcal{C}\mathcal{N}(0,1), ~n\sim\mathcal{CN}(0,\sigma^2)$ are independent. 
Following the definitions in \eqref{def:sinr} and \eqref{def:sep}, the SINR and SEP of the above scalar asymptotic  model  are given by
\begin{equation}\label{snr}
\begin{split}
\overline{\text{\normalfont{SINR}}}&=\frac{\sigma_s^2\,\eta^2\,\overline{T}_\bs^2}{\eta^2\,\overline{T}_\sg^2+\sigma^2}
=\frac{\mathbb{E}^2[d\, f(d)]}{\text{var}[d\, f(d)]+\phi(\bar{\alpha},\eta)\, \frac{\mathbb{E}[f^2(d)]}{\gamma}},
\end{split}
\end{equation}
where $\bar{\alpha}=\sqrt{\sigma_s^2\,\mathbb{E}[f^2(d)]/\gamma}$ and 
\begin{equation}\label{Eqn:phi_def}
\phi(\bar{\alpha},\eta):=\frac{\mathbb{E}[|q(\bar{\alpha} \sZ)|^2]-|\mathbb{E}[\sZ^\dagger q(\bar{\alpha} \sZ)]|^2+\sigma^2/\eta^2}{|\mathbb{E}[\sZ^\dagger q(\bar{\alpha} \sZ)]|^2},
\end{equation} and 
\begin{equation*}
\overline{\text{\normalfont{SEP}}}(\beta)=\mathbb{P}\left(\text{\normalfont{dec}}(\beta\bar{y})\neq s\right).
\end{equation*}
The second step of \eqref{snr} is obtained based on the definitions of $\overline{T}_\bs$ and $\overline{T}_\sg$ in \eqref{Asympmodel2}. 

Theorem \ref{corollary:sinrsep} below shows that both the SINR and the SEP performance of the original model converge to those of the scalar asymptotic model in \eqref{scalarmodel}. Its proof can be found in Appendix \ref{proof:corollary1}.

  \begin{theorem}\label{corollary:sinrsep}
 Denote $\widehat{\text{\normalfont{SINR}}}_k$ and $\widehat{\text{\normalfont{SEP}}}_k(\beta)$ 
 as the 
 \text{\normalfont{SINR}} and {\normalfont{SEP}} of user $k$ of the model in \eqref{Equiy}, respectively.  Under the asymptotic setting in Theorem \ref{asy}, the following hold for any $k\in[K]$ and $\beta\in\mathbb{C}$:
\begin{enumerate}
 \item[(i)] $ \lim_{N,K\to\infty}\widehat{\text{\normalfont{SINR}}}_k=\overline{\text{\normalfont{SINR}}}$;
\item[(ii)] $\lim_{N,K\to\infty}\widehat{\text{\normalfont{SEP}}}_k(\beta)=\overline{\text{\normalfont{SEP}}}(\beta)$.
\end{enumerate}
 \end{theorem}
 \vspace{5pt}

With the above theorem, we can give predictions of the SINR and SEP performance of the original model based on the   {\red asymptotic scalar} model in \eqref{scalarmodel}. Note that $\overline{\text{SEP}}(\beta)$ is defined for a scalar additive white Gaussian noise (AWGN) channel and therefore can be easily computed for specific constellation types. 
Clearly, a natural scaling factor should be
\begin{equation}\label{factor}
\beta=\frac{\overline{T}_\s^\dagger}{\eta|\overline{T}_\s|^2}.
\end{equation}
We denote the corresponding asymptotic SEP as $\overline{\text{SEP}}$.   When $M$-PSK modulation is adopted, $\overline{\text{SEP}}$ can be tightly approximated as  \cite[Section 4.3-2]{digitalcommunication}:
\begin{equation}\label{psk}
\overline{\text{SEP}}\approx2Q\left(\sqrt{2}\sin\frac{\pi}{M}\sqrt{\overline{\text{SINR}}}\right),
\end{equation}
where $Q(x)=\frac{1}{\sqrt{2\pi}}\int_x^{\infty}e^{-\frac{t^2}{2}}dt$. For $M$-QAM modulation, the SEP  has an explicit expression \cite[Section 4.3-3]{digitalcommunication}: 
\begin{equation}\label{qam}
\overline{\text{SEP}}=4\left(1-\frac{1}{\sqrt{M}}\right)E_0-4\left(1-\frac{1}{\sqrt{M}}\right)^2E_0^2
\end{equation}
where 
\begin{equation}\label{e0}
 E_0=Q\left(\sqrt{\frac{3}{M-1}\,\overline{\text{SINR}}}\right).
\end{equation}
In the rest of this paper, we assume that $\beta$ in \eqref{factor} is used unless otherwise stated. {\color{black}It is worth noting that the asymptotic SEP formulas in \eqref{psk} and \eqref{qam}, though derived under the asymptotic assumption that the system dimension grows to infinity, are also accurate for practical systems with moderate dimensions; see the simulation results in Section \ref{section6}.}

\subsection{Example: QCE Precoding}
So far, our results hold for general $q(\cdot)$ satisfying Assumption \ref{Ass:q}. In this subsection, we specify our results to the case of CE quantizer $\qce(\cdot)$. We start with a few properties of $\qce(\cdot)$.

\begin{lemma}[Properties of $\qce(\cdot)$]\label{lemma:qce}\hspace{5cm}
\begin{enumerate}
\item[(i)] $|\qce(x)|=1,~\forall\, x\in \C$;\vspace{3pt}
\item[(ii)] $\qce(\alpha x)=\qce(x),~\forall\, x\in \C,\alpha>0$; \vspace{3pt}
\item[(iii)] $\mathbb{E}[{\sZ}^\dagger \qce({\sZ})]=\frac{L\sin\frac{\pi}{L}}{2\sqrt{\pi}}$, where $\sZ\sim\mathcal{C}\mathcal{N}(0,1)$ and $L$ is the number of quantization levels \cite{CEQ}. 
\end{enumerate}
\end{lemma}

Under the above properties of $\qce(\cdot)$, the asymptotic SINR in \eqref{snr} simplifies to
 \begin{equation}\label{snrqce}
\begin{aligned}
\overline{\text{SINR}}_{\text{QCE}}=\frac{\mathbb{E}^2[{d}\, f({d})]}{\text{var}[{d}\, f({d})]+\frac{C_{L,\sigma,\eta}}{\gamma}\,\mathbb{E}[f^2({d})]},
\end{aligned}
\end{equation} 
where 
\begin{equation}\label{flsigma}
 C_{L,\sigma,\eta}=\frac{1-\frac{L^2\sin^2\frac{\pi}{L}}{4\pi}+\sigma^2/\eta^2}{\frac{L^2\sin^2\frac{\pi}{L}}{4\pi}}.
\end{equation}
The asymptotic SEP is a simple function of the SINR, as discussed {\red above}.

By further specifying the nonlinear function $f(\cdot)$ {\red as} $f(x)=x$ and $f(x)=x^{-1}$ respectively (see discussions in Section \ref{section2A}) and using  \cite[Eq.~(2.104)]{tulino2004random}, we can {\red obtain} the following asymptotic SINR formulas for the quantized MF and ZF precoders:
\begin{equation*}
 \begin{aligned}
 \overline{\text{SINR}}_{\text{QCE}}^{\text{MF}}&= \frac{\gamma}{C_{L,\sigma,\eta}+1},\\
 \overline{\text{SINR}}_{\text{QCE}}^{\text{ZF}} &=\frac{\gamma-1}{C_{L,\sigma,\eta}},
 \end{aligned}
 \end{equation*}
where $C_{L,\sigma,\eta}$ is defined in \eqref{flsigma}. The above SINR formula for quantized ZF precoding has been obtained for the one-bit case (i.e., $L=4$) using the Bussgang decomposition technique in \cite{ZF}.

\section{Optimal Linear-quantized Precoding}\label{section5}

For the precoding scheme in \eqref{Eqn:P}, the function $f(\cdot)$ can be designed to optimize the system SEP performance. In this section, we derive the optimal $f(\cdot)$ based on the asymptotic characterization developed in Theorem \ref{asy}.

\subsection{Optimal Linear-Quantized Precoding}

Our goal is to find the optimal $f(\cdot)$ in terms of the asymptotic SEP performance. First, maximizing $\overline{\text{SINR}}$ is equivalent to minimizing $\overline{\text{SEP}}$, as shown in Appendix \ref{App:SINR_SEP}. In what follows, we shall focus on the following SINR maximization problem (see \eqref{snr}):
\begin{equation}\label{prob:maxsinr}
\begin{aligned}
\zeta^*:=\sup_{f,\eta>0,\bar{\alpha}>0}~&\frac{\mathbb{E}^2[d\, f(d)]}{\text{var}[d\, f(d)]+\phi(\bar{\alpha},\eta)\, \frac{\mathbb{E}[f^2(d)]}{\gamma} }\\
\text{s.t. }~~~~
&\eta^2\,\mathbb{E}[|q(\bar{\alpha} \sZ)|^2]\leq P_T,\\
&\mathbb{E}[f^2(d)]=\frac{\gamma}{\sigma_s^2}\, \bar{\alpha}^2,\\
\end{aligned}
\end{equation} 
where ${\sZ}\sim\mathcal{C}\mathcal{N}(0,1)$, $P_T>0$ is the maximum average transmit power in \eqref{powerconstraint}, $d$ is defined in Theorem \ref{asy}, and $\phi(\bar{\alpha},\eta)$ is defined in \eqref{Eqn:phi_def}.  If $\bar \alpha$ in the second constraint is eliminated and substituted into the first constraint, then the obtained constraint represents the asymptotic counterpart of the actual average power constraint in \eqref{powerconstraint}. Therefore, the function $f(\cdot)$ and the transmit power $\eta$  are jointly optimized in problem \eqref{prob:maxsinr}.




The SINR maximization problem \eqref{prob:maxsinr} may seem challenging to solve as it involves optimization over a function $f(\cdot)$. Moreover, the {\red variables} $(\bar{\alpha},\eta, f)$ are coupled by the constraints, which further complicates the problem.
Fortunately, the optimal solution of \eqref{prob:maxsinr} has a simple structure, as shown by the following theorem.

\begin{theorem}\label{optimalprecoder}
{\red Suppose that the following infimum is attained by some $\bar{\alpha}^\ast>0$:}
\begin{subequations}\label{Eqn:theorem_all}
\begin{equation}\label{alpha}
\phi^*:=\left(1+\frac{\sigma^2}{P_T}\right)\inf_{\bar{\alpha}> 0}\frac{\mathbb{E}[|q(\bar{\alpha} \sZ)|^2]}{\left|\mathbb{E}[\sZ^\dagger q(\bar{\alpha} \sZ)]\right|^2}-1.
\end{equation}
Then, $(\bar{\alpha}^\ast,\eta^\ast,f^\ast)$ is an optimal solution to the problem in \eqref{prob:maxsinr} where
\begin{equation}\label{eta}
\eta^\ast:=\sqrt{\frac{P_T}{\mathbb{E}[|q(\bar{\alpha}^*\sZ)|^2]}},
\end{equation}
\begin{equation}\label{fstar}
f^\ast(d):=\pm\frac{d}{\tau^*(d^2+\frac{\phi^\ast}{\gamma})},
\end{equation}
 \begin{equation}\label{tau}
\tau^\ast:=\sqrt{\frac{\sigma_s^2}{\gamma (\bar{\alpha}^\ast)^2} \, \mathbb{E}\left[\left(\frac{d}{d^2+\frac{\phi^\ast}{\gamma}}\right)^2\right]}.
 \end{equation}
Furthermore, the maximum SINR in \eqref{prob:maxsinr} is equal to
\begin{equation}\label{Eqn:zeta_opt}
\zeta^*=\frac{1}{1-\mathbb{E}\left[\frac{d^2}{d^2+\frac{\phi^{*}}{\gamma}}\right]}-1.
\end{equation}
\end{subequations}
\end{theorem}
\begin{proof}
See Appendix \ref{proof:Th4}.
\end{proof}

 \vspace{3pt}
 \begin{remark}\label{remark1}
 Some comments on Theorem \ref{optimalprecoder} are in order.
\begin{enumerate}
\item In Theorem \ref{optimalprecoder}, we do not impose the positivity constraint on $f$, and an optimal $f$ satisfying $f>0$ always exists; see \eqref{fstar}. 
 This supports our claim below Assumption \ref{Ass:f} that the positivity assumption on $f$ does not impose any restriction in terms of the best achievable performance. 
\item We assume that the {\red infimum in \eqref{alpha} is attainable}. Our numerical results suggest that this holds for commonly used quantization {\red functions} $q(\cdot)$. In {\red cases} where the infimum in \eqref{alpha} is not attainable, we may modify the theorem using a simple truncation argument. Specifically, we add an additional constraint $\frac{1}{M}\le\bar{\alpha}\le M$ to \eqref{alpha}  and let $\bar{\alpha}^\ast_M$ be any optimal solution to the constrained problem. We define $\phi_M^\ast$, $\eta^\ast_M$ and $f^\ast_M$ as in \eqref{Eqn:theorem_all} with $\bar{\alpha}^\ast$ replaced by $\bar{\alpha}^\ast_M$. Then, it can be shown that the SINR achieved by $(\bar{\alpha}_M^\ast,\eta^\ast_M,f^\ast_M)$ tends to $\zeta^\ast$ as $M\to\infty$.
 \end{enumerate}
 \end{remark}
 
  \vspace{5pt}

{\red We now take a closer look at \eqref{alpha}. Loosely speaking, it} may be interpreted as finding the optimal input power for the quantization function $q(\cdot)$.  
More precisely, using the Baussgang decomposition technique, we can decompose $q(\bar{\alpha} \sZ)$ as
\[
q(\bar{\alpha} \sZ)\overset{d}{=}\mathbb{E}[\sZ^\dagger q(\bar{\alpha} \sZ)]\, \sZ+d,
\]
where $\sZ\sim\mathcal{CN}(0,1)$ and $d$ models the quantization error which is uncorrelated with $\sZ$. Then, problem \eqref{alpha} can be viewed as maximizing the signal to quantization error plus noise ratio (SQNR):
\[
\begin{split}
\text{SQNR}&:=\frac{ |\mathbb{E}[\sZ^\dagger q(\bar{\alpha} \sZ)]|^2\, \mathbb{E}[|\sZ|^2]  }{\mathbb{E}[|d|^2]}\\
&=\frac{ |\mathbb{E}[\sZ^\dagger q(\bar{\alpha} \sZ)]|^2 }{\mathbb{E}[|q(\bar{\alpha} \sZ)|^2]-|\mathbb{E}[\sZ^\dagger q(\bar{\alpha} \sZ)]|^2}\\
&=\frac{ 1 }{\frac{\mathbb{E}[|q(\bar{\alpha} \sZ)|^2]}{|\mathbb{E}[\sZ^\dagger q(\bar{\alpha} \sZ)]|^2}-1}.
\end{split}
\]
For a general $q(\cdot)$, problem \eqref{alpha} can be solved by a one-dimensional search. For the special case of {\red the} CE quantizer, the objective function of \eqref{alpha} is actually constant (see Lemma \ref{lemma:qce}) and the problem is trivial; see details in the next subsection. 

\begin{remark}[Connection with the WFQ precoder \cite{WFQ}]
Theorem \ref{optimalprecoder} shows that for precoding matrices of the form $\bP=\bV f(\bD)^\mathsf{T}\bU^\mathsf{H}$, the {\color{black}asymptotically} optimal one is precisely the RZF precoder:  
\begin{equation}\label{eqn:optp}
\bP^*=\frac{1}{\tau^*}\H^\mathsf{H}\left(\H\H^\mathsf{H}+\frac{\phi^*}{\gamma}\,\mathbf{I}_{\color{black} K}\right)^{-1},
\end{equation}
where ${\tau^*}$ and $\phi^*$ can be obtained by solving the one-dimensional optimization problem in \eqref{alpha}. Interestingly, \eqref{eqn:optp} is closely related to the WFQ precoder{\red\cite{WFQ} designed for traditional quantized precoding}:
$$\bP_{\text{WFQ}}=\frac{1}{g_0}\left(\H^\mathsf{H}\H-\rho_q \text{\normalfont{nondiag}}(\H^\mathsf{H}\H)+\frac{\sigma^2}{\gamma P_T} \mathbf{I}_{\color{black}N}\right)^{-1}\H^\mathsf{H},$$
where $g_0$ is a scaling factor that ensures {\red a} certain power constraint to be satisfied. 
As $N,K\to\infty$, the diagonal entries of $\H^\mathsf{H}\H$ converge to
\begin{equation}\label{hhii}
(\H^\mathsf{H}\H)[{i,i}]=\sum_{j=1}^K|\H[j, i]|^2\xrightarrow{a.s.} \frac{1}{\gamma},\quad i=1,2,\dots, N.
\end{equation}
Therefore, $\bP_{\text{WFQ}}$ can be approximated as
\begin{equation}\label{wfqrzf}
\begin{aligned}
\bP_{\text{WFQ}}=&\frac{1}{g_0}\left(\H^\mathsf{H}\H-\rho_q \text{\normalfont{nondiag}}(\H^\mathsf{H}\H)+\frac{\sigma^2}{\gamma P_T}\mathbf{I}_{\color{black}N}\right)^{-1}\H^\mathsf{H}\\
{=}\,&\frac{1}{g_0}\left((1-\rho_q) \,\H^{\mathsf{H}}\H+\rho_q\text{\normalfont{diag}}(\H^\mathsf{H}\H)+\frac{\sigma^2}{\gamma P_T}\mathbf{I}_{\color{black}N}\right)^{-1}\hspace{-0.2cm}\H^\mathsf{H}\\
\overset{(a)}{\approx}&\frac{1}{g_0}\left((1-\rho_q) \,\H^{\mathsf{H}}\H+\frac{\rho_q+\frac{\sigma^2}{P_T}}{\gamma}\,\mathbf{I}_{\color{black}N}\right)^{-1}\H^\mathsf{H}\\
\overset{(b)}{=}\,&\frac{1}{g_0}\H^{\mathsf{H}}\left((1-\rho_q)\, \H\H^{\mathsf{H}}+\frac{\rho_q+\frac{\sigma^2}{P_T}}{\gamma}\,\mathbf{I}_{\color{black}K}\right)^{-1}\\
\overset{(c)}{=}\,&\frac{1}{g_0(1-\rho_q)}\, \H^{\mathsf{H}}\left(\H\H^{\mathsf{H}}+\frac{\rho_q+\frac{\sigma^2}{P_T}}{\gamma(1-\rho_q)}\,\mathbf{I}_{\color{black}K}\right)^{-1},\\
\end{aligned}
\end{equation}
{\color{black}where $(a)$ is due to \eqref{hhii}, (b) uses the fact that $\left(\bX^\mathsf{H}\bX+\rho \mathbf{I}_N\right)^{-1}\bX^\mathsf{H}=\bX^\mathsf{H}\left(\bX\bX^{\mathsf{H}}+\rho \mathbf{I}_K\right)^{-1}$ for any $\bX\in\C^{K\times N}$ and $\rho>0$, and (c) is obtained by extracting the  coefficient $1-\rho_q$ from the inverse matrix.} Comparing {\color{black} the last line of } \eqref{wfqrzf} with \eqref{eqn:optp}, we see that both $\bP_{\text{WFQ}}$ and  $\bP^\ast$ are RZF precoders, and are identical if the regularization parameters and the scaling factors are the same:
\begin{equation*}\label{Eqn:rho_q}
\begin{split}
\phi^\ast&=\frac{\rho_q+\frac{\sigma^2}{P_T}}{1-\rho_q},\quad\tau^*=g_0(1-\rho_q).
\end{split}
\end{equation*}

Note that 
$\bP_{\text{WFQ}}$ and $\bP^\ast$ are derived based on different criteria and motivations, and not directly comparable. For the WFQ precoder, the independent quantizer $q_{\text{\tiny I}}(\cdot)$ is optimized together with the linear precoder  (see also \cite{pointdesign1,pointdesign2} for optimization of the independent quantizer $q_{\text{\tiny I}}(\cdot)$), and the quantization intervals for $q_{\text{\tiny I}}(\cdot)$ have to be chosen carefully to satisfy several conditions (see \cite[Eqs.~(7)-(9)]{WFQ}). On the other hand, $q(\cdot)$ is a generic fixed function in this work, and the input power of $q(\cdot)$ (dictated by $\bar{\alpha}$) is properly optimized.
\end{remark}

\subsection{Example: QCE Precoding}

In this subsection, we consider the special case of QCE precoding. As a direct corollary of Theorem \ref{optimalprecoder}, we have the following result.

\begin{corollary}\label{optimalprecoder_qce}
When $q(\cdot)$ is specified as $\qce(\cdot)$ in problem \eqref{prob:maxsinr}, the optimal $\eta$ is $\eta^*=\sqrt{P_T}$, and the optimal $f$ has a closed-form expression:
\begin{equation}\label{optqcef}
f^*(d)=\frac{d}{d^2+\frac{C_{L,\sigma,\eta^*}}{\gamma}},
\end{equation} where $C_{L,\sigma,\eta^*}$ is given in \eqref{flsigma}.
In this case, the {\color{black}asymptotically} optimal precoding matrix is given by 
\begin{equation}\label{optqce}
\bP^*_{\text{QCE}}=
\H^\mathsf{H}\left(\H\H^\mathsf{H}+\frac{C_{L,\sigma,\eta^*}}{\gamma}\,\mathbf{I}_{\color{black}K}\right)^{-1},
\end{equation}
and its asymptotic SINR is 
\begin{equation*}\label{sqirrzf}
\begin{aligned}\overline{\text{\normalfont{SINR}}}^*_{\text{QCE}}=\frac{\sqrt{u^2+4C_{L,\sigma,\eta^*}}+u}{2C_{L,\sigma,\eta^*}}-1,
\end{aligned}
\end{equation*}
where $u=C_{L,\sigma,\eta^*}+\gamma-1$.
\end{corollary} 
\begin{proof}
When $q(\cdot)=\qce(\cdot),$ it follows from Lemma \ref{lemma:qce} and  \eqref{eta} that  $\eta^*=\sqrt{P_T}$ and the objective function of problem \eqref{alpha} is a constant $C_{L,\sigma,\eta^*}$, which is  independent of $\bar{\alpha}$. Hence, the optimal solution set of problem \eqref{alpha} is $\{\bar{\alpha}^*\mid\bar{\alpha}^*> 0 \}$ and $\tau^*$ can be any positive number.  Without loss of generality, here we set $\tau^*=1$. This further gives the optimal $f^*(\cdot)$ and the optimal precoding matrix $\bP^*_{\text{QCE}}$ in \eqref{optqcef} and \eqref{optqce}, respectively. Finally,  $\overline{\text{\normalfont{SINR}}}^*_{\text{QCE}}$ can be obtained by plugging $\phi^*=C_{L,\sigma,\eta^*}$ into \eqref{Eqn:zeta_opt}  and using \cite[Eq. (2.42)]{tulino2004random}.
\end{proof}
 We note that when $\gamma$ is large, $\sqrt{u^2+4C_{L,\sigma, \eta^*}}\approx u$, which implies that
\begin{equation}\label{rzfapproxzf}
\begin{aligned}
\overline{\text{SINR}}^*_{\text{QCE}}&\approx \frac{u}{C_{L,\sigma,\eta^*}}-1=\frac{\gamma-1}{C_{L,\sigma,\eta^*}}=\overline{\text{SINR}}^{\text{ZF}}_{\text{QCE}},
\end{aligned}
\end{equation}
i.e., the performance of  quantized ZF precoding is nearly optimal when the antenna-user ratio is large.

\section{Simulation Results}\label{section6}

In this section, we present simulation results to demonstrate the theoretical results obtained in previous sections. We  consider both traditional quantized precoding (where the independent quantizer is used) and QCE precoding (where the CE quantizer is used), and assume that the precoding factor (namely, the linear scaling applied at the receiver side before detection) in \eqref{factor} is used. We also assume that the precoder adopted in this section has the optimal input power for the corresponding quantization function, i.e., $\bar{\alpha}=\bar{\alpha}^*$, where $\bar{\alpha}^*$ is an optimal solution to problem \eqref{alpha}. 
The transmit power is set as $P_T=1$. 
All results are averaged over $10^5$ channel realizations.

\subsection{Numerical Validation of SEP Formulas}
We first verify the accuracy of the SEP formulas given in \eqref{psk}--\eqref{e0}. 

In Fig. \ref{ser}, we plot the symbol error rate (SER) as a function of {\red the ratio of the number of antennas to users $\gamma$ for the} quantized MF and ZF precoding schemes. {\color{black}Both the cases of a large system with $K=100$ and a more practical system with $K=20$ are investigated, where the number of transmit antennas is set as $N=\gamma K$.}
 We consider three types of signal constellations: QPSK, 8-PSK, and 16-QAM, and two types of quantization: CE quantization and independent quantization. The channel noise is set as $\sigma=0$.  
As shown in the figure, there {\red is a} slight mismatch between simulations and asymptotic predictions when $K=20$, and the differences become indistinguishable when $K=100$. 

  \begin{figure} [t]
\centering    
  \hspace{-0.06\columnwidth}
\subfigure[Quantized MF.] { 
\includegraphics[width=0.5\columnwidth]{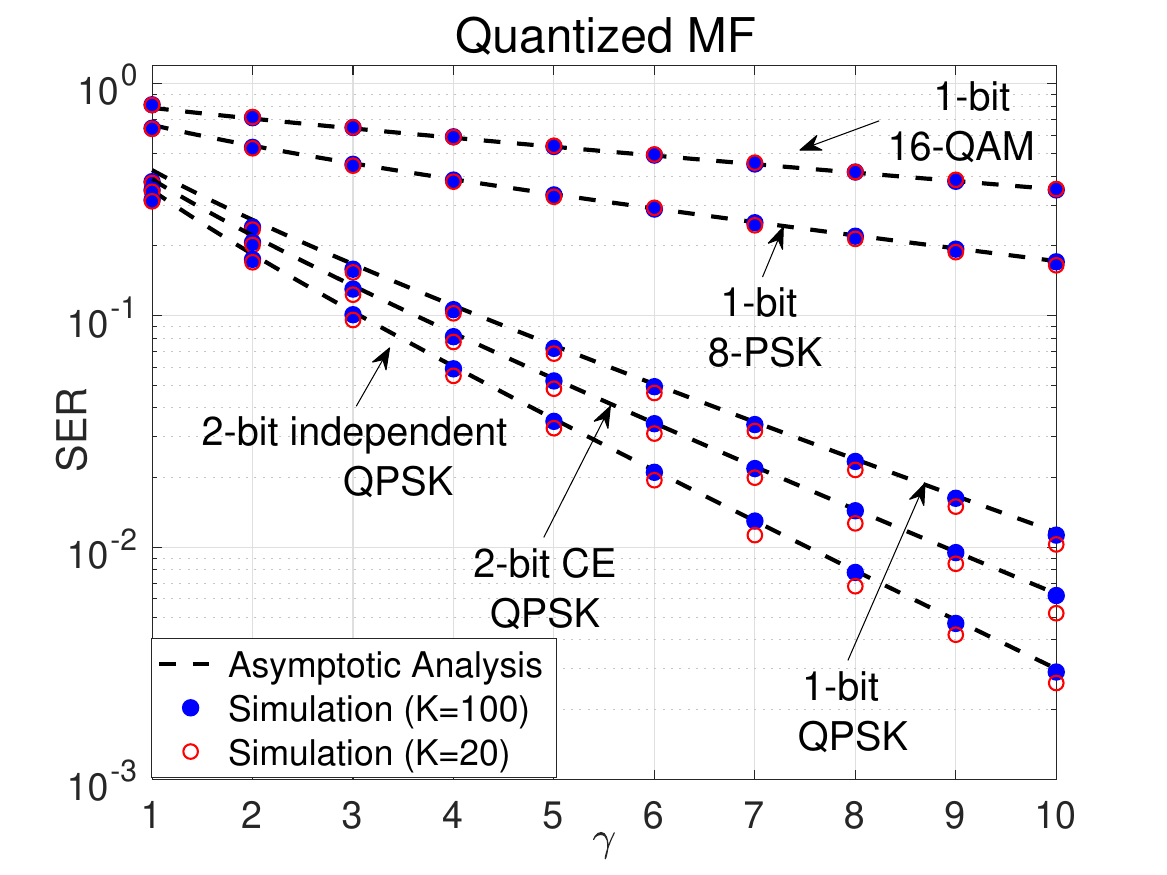}  
}     
\subfigure[Quantized ZF.] { 
\includegraphics[width=0.5\columnwidth]{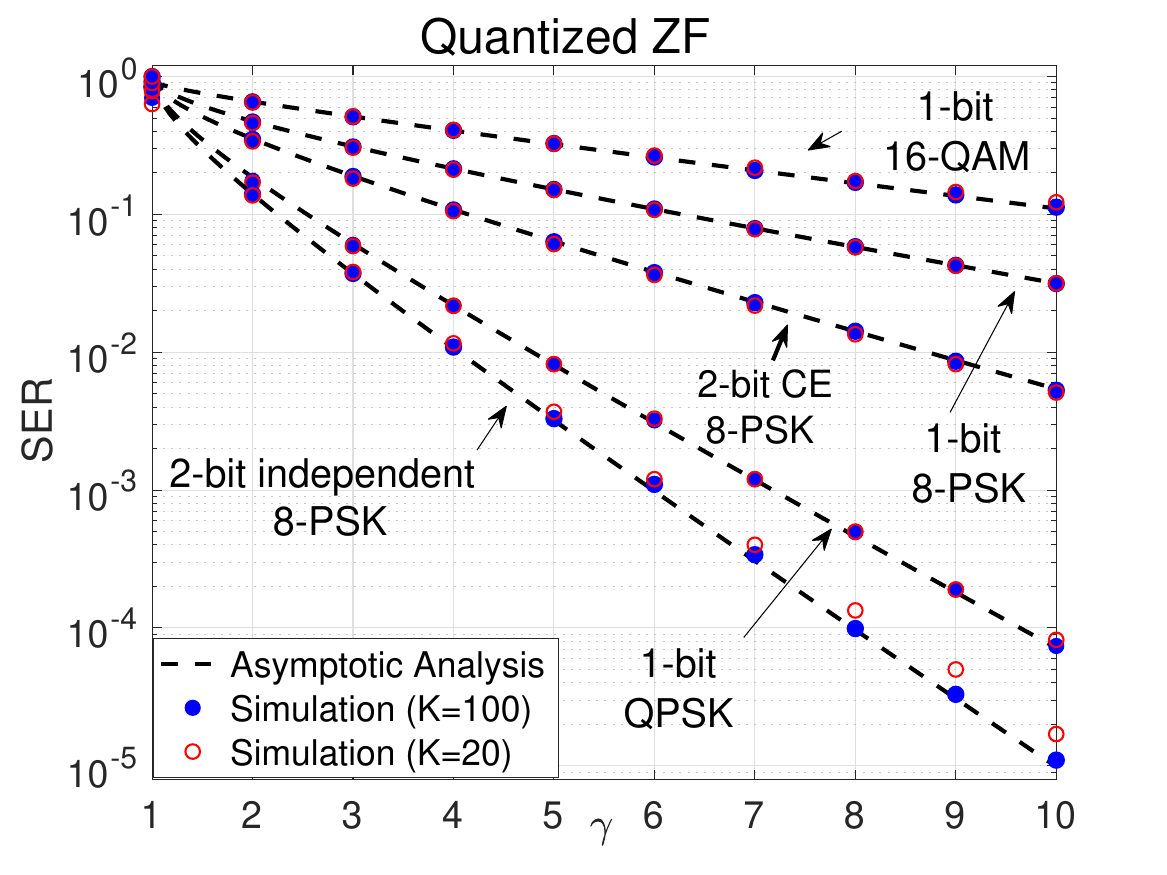}  
}    
\caption{SER performance versus $\gamma$ for quantized MF and ZF precoding with $\sigma=0$, where PSK and QAM represent the type of constellation and  ``independent'' and ``CE'' represent the type of quantization. The quantization interval for the independent quantizer in \eqref{eq:xl} is $\Delta=2$.}
 \label{ser}
\end{figure}

We can also draw some interesting observations from  Fig. \ref{ser}. First, the logarithm of SEP/SER decreases linearly with $\gamma$, i.e., the antenna-user ratio,  which reveals the gain of increasing the number of transmit antennas at the BS, and the slope of the decrease is determined by the precoder, the quantizer, and the constellation {\red that are employed}.
Second, the independent quantizer and the CE quantizer have different behaviors under different precoding schemes. For example, the superiority (in terms of SEP/SER) of {\red the} 2-bit independent quantizer over {\red the} 2-bit CE quantizer under ZF precoding is much more prominent than that under MF precoding.

\begin{table}[t]
\centering
\caption{$\phi^*$ in \eqref{alpha} for scalar and CE quantizers with $\sigma=0$.}
\begin{tabular}{ccccc}
\hline
&1 bit &2 bits &3 bits&4 bits \\
\hline
scalar&0.57&0.14&0.04&0.01\\
CE& 0.57&0.34&0.29&0.28\\
\hline
\end{tabular}
\label{table:phi}
\end{table}
In fact, the above observations are clear from our asymptotic analysis.  To be specific, applying the approximation $Q(x)\approx \frac{1}{2}e^{-\frac{1}{2}x^2}$\cite{digitalcommunication} to the SEP predictions for $M$-PSK and $M$-QAM {\red constellations} in \eqref{psk} and \eqref{qam}, we get
\begin{equation}\label{eqn:slope1}
\ln\overline{\text{SEP}}\approx\left\{
\begin{aligned}
-2\sin^2\frac{\pi}{M}\overline{\text{SINR}},\qquad\qquad\qquad\qquad~~&\text{for PSK;}\\
-\frac{3}{2(M-1)}\overline{\text{SINR}}+\log(2-\frac{2}{\sqrt{M}}),~~~&\text{for QAM,}
\end{aligned}\right.
\end{equation}
i.e., the logarithm of $\overline{\text{SEP}}$ decreases linearly with $\overline{\text{SINR}}$. According to \eqref{snr}, $\overline{\text{SINR}}$ for MF and ZF precoding can be expressed as 
\begin{equation}\label{eqn:slope2}
\overline{\text{SINR}}=\left\{
\begin{aligned}
\frac{\gamma}{\phi^*+1},~~~&\text{for MF;}\\
\frac{\gamma-1}{\phi^*},~~~~&\text{for ZF},
\end{aligned}\right.
\end{equation}
where $\phi^*$ is given in \eqref{alpha} and is determined by the quantization type. The values of $\phi^*$ for scalar and CE quantizers with $\sigma=0$ are given in Table \ref{table:phi}.   
Eqs \eqref{eqn:slope1} and \eqref{eqn:slope2} demonstrate the linear decrease in the logarithm of {\red the} SEP with $\gamma${\red,} and quantitatively  characterize the effects of precoding, quantization, and constellation on the slope of  the decrease. In particular, we can see that $\phi^*$  has a stronger impact on ZF than on MF, since the slope of decrease is proportional to $\frac{1}{\phi^*}$ and  $\frac{1}{\phi^*+1}$ for ZF and MF, respectively, and $\phi^*$ is  {\red on} the order of $10^{-2}-10^{-1}$ as shown in Table \ref{table:phi}. This explains why {\red the} 2-bit independent quantizer has {\red a greater} performance gain {\red compared with the} 2-bit CE quantizer under ZF precoding than under MF precoding. 

In Fig. \ref{serl}, we plot the SER performance of quantized MF and ZF for both traditional quantized precoding and QCE precoding as a function of the {\red DAC} resolution.  For MF, the performance gain of both quantizers is small as the {\red DAC} resolution increases, especially when the resolution is larger than 2 bits. The same happens for ZF precoding with CE quantization. However, there is a remarkable gain for ZF precoding with independent quantization as the {\red DAC} resolution increases. These observations can also be interpreted by our previous discussions.


Finally, Fig. \ref{sersnr} shows the SER of quantized MF and ZF as a function of the channel SNR (i.e., ${1}/{\sigma^2}$) for a one-bit system with $\gamma=6$ and QPSK modulation.  We see that ZF has a noticeable performance gain compared with MF in the high SNR region.
\begin{figure}
\centering   \hspace{-0.06\columnwidth}
\subfigure[Quantized MF.]{
\includegraphics[width=0.5\columnwidth]{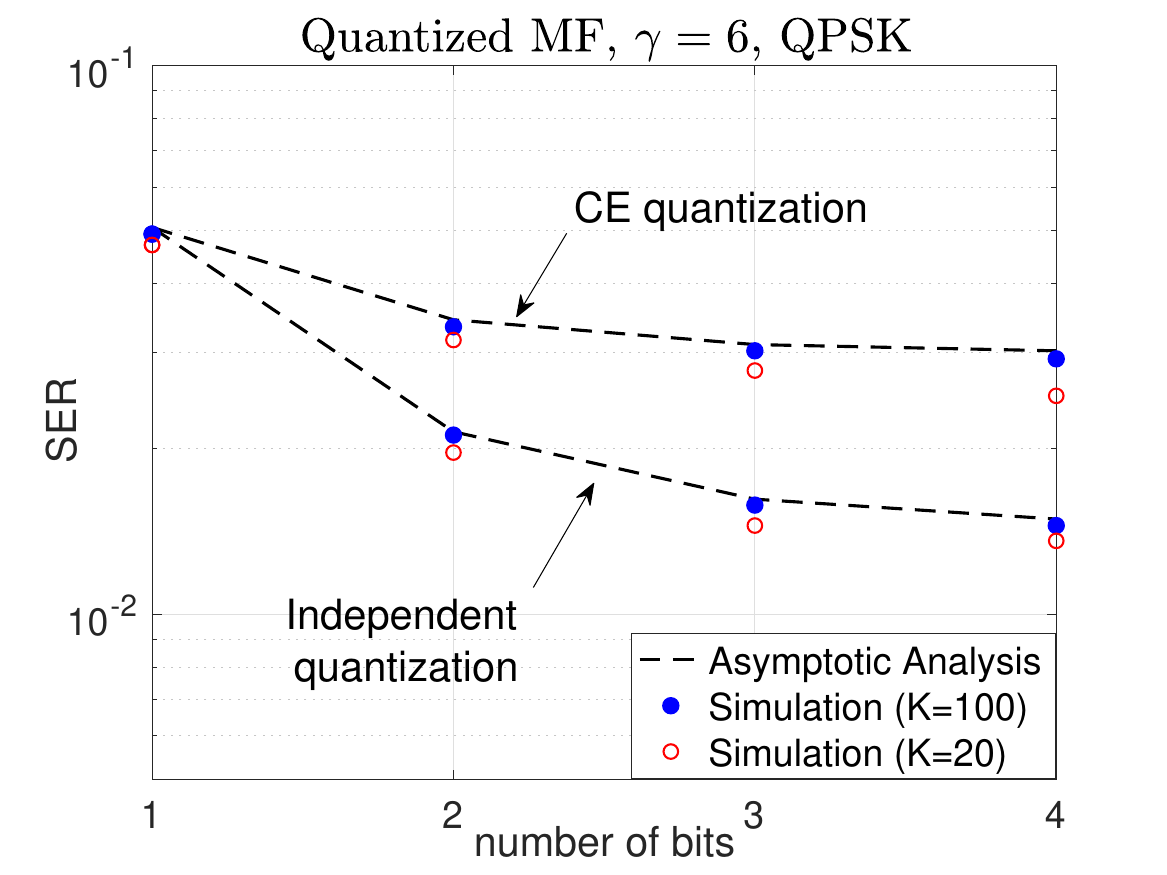}}
\subfigure[Quantized ZF.]{
\includegraphics[width=0.5\columnwidth]{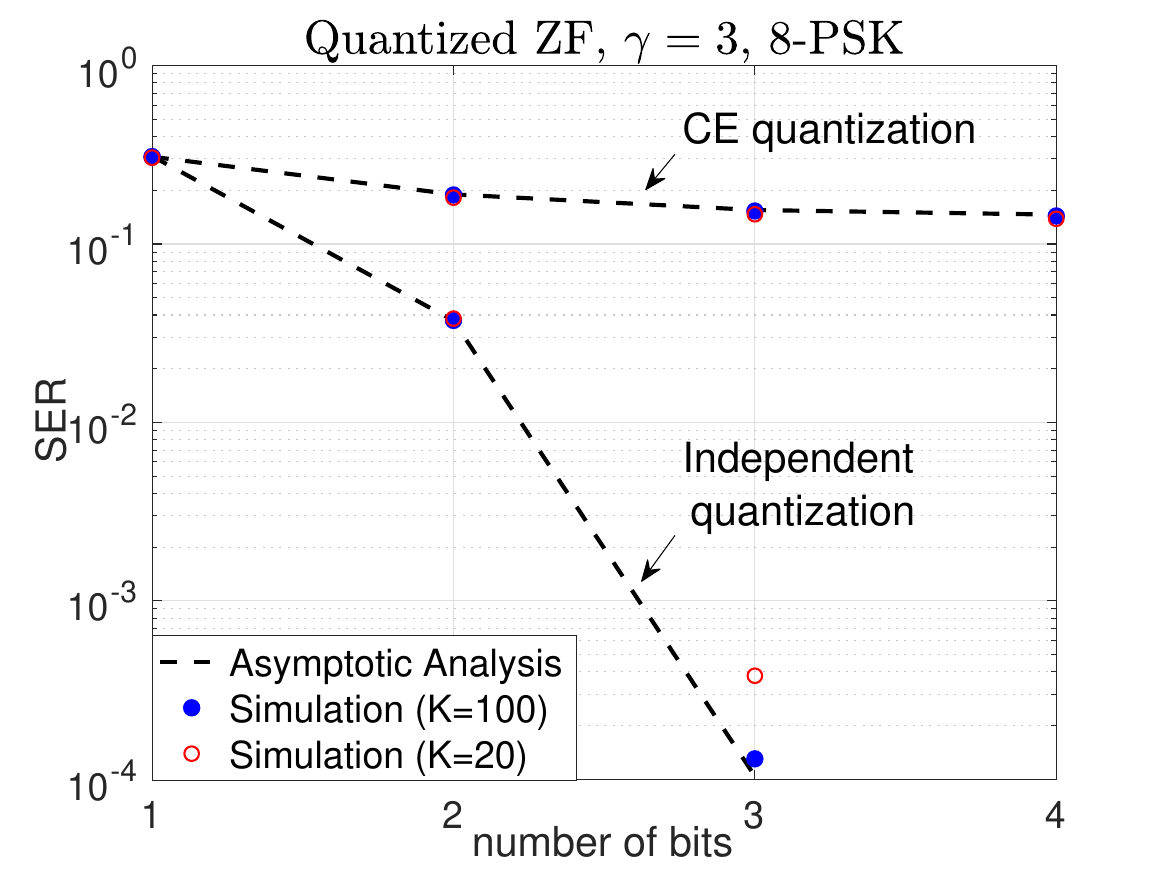}}
\caption{SER performance versus {\red DAC resolution for} quantized MF with $\gamma=6$ and QPSK modulation and quantized ZF  with $\gamma=3$ and $8$-PSK modulation; $\sigma=0$.}
\label{serl}
\end{figure}

\begin{figure}
\centering
\includegraphics[scale=0.45]{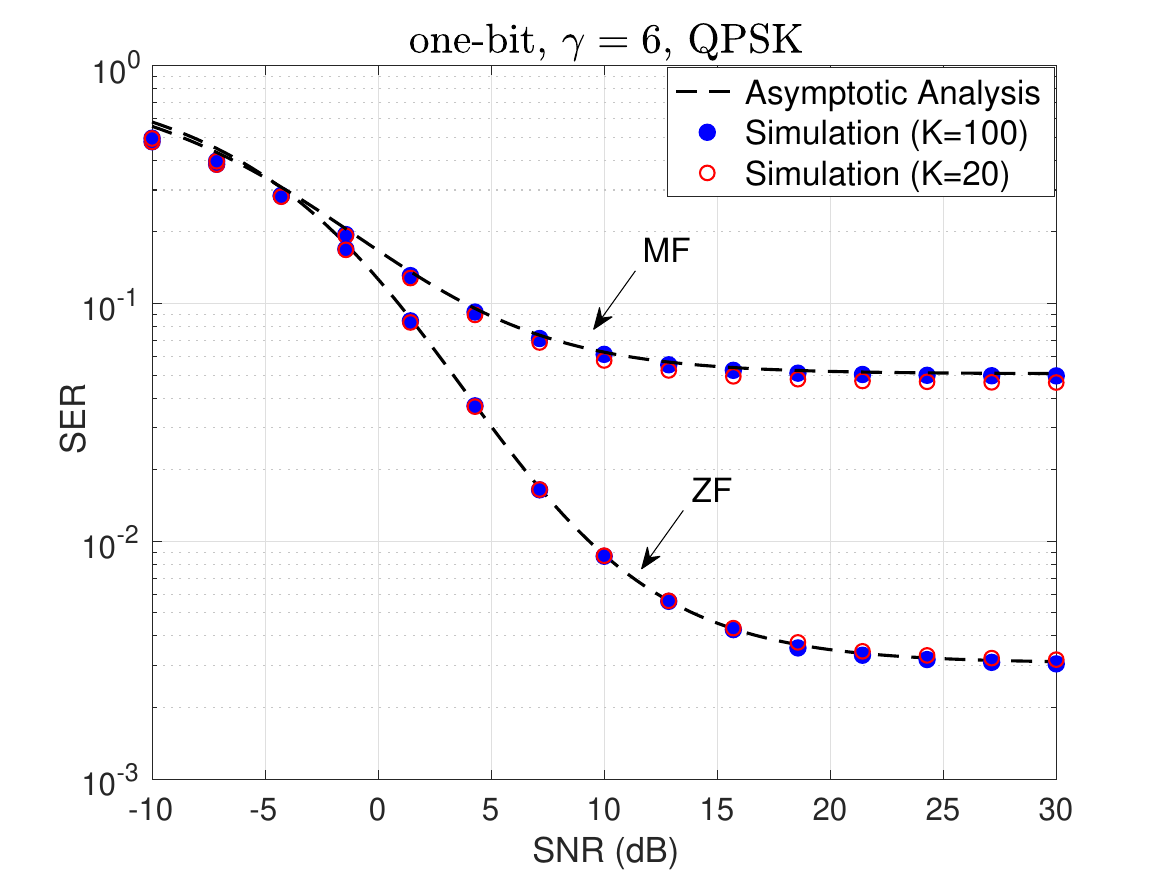}
\caption{SER performance versus {\red SNR for} one-bit quantized MF and ZF precoding with  $\gamma=6$ and QPSK modulation.}
\label{sersnr}
\end{figure}
\subsection{Optimality of Quantized RZF Precoding}
In this subsection, we present some simulation results to demonstrate the optimality of the quantized RZF precoding matrix given in  \eqref{eqn:optp}.


\begin{figure}[t]
\includegraphics[scale=0.45]{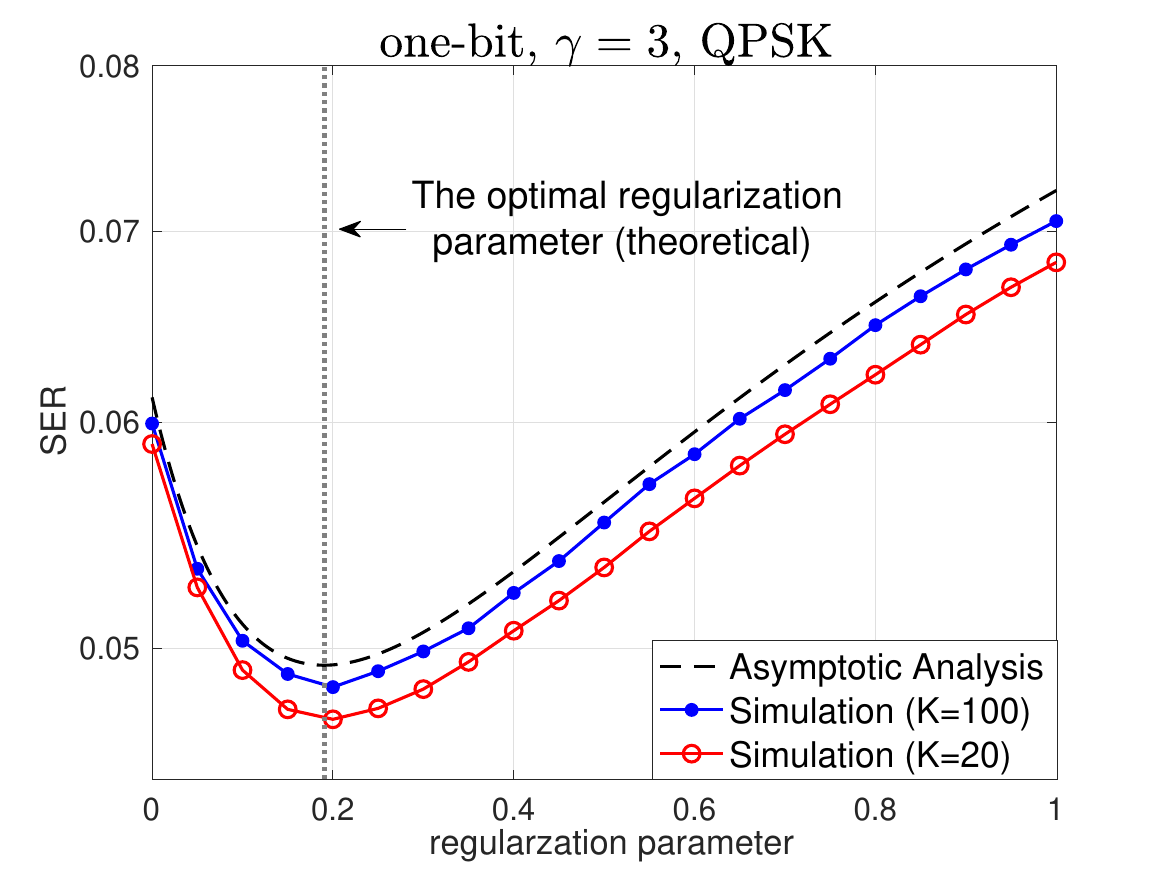}
\centering
\caption{SER performance versus regularization parameter for a one-bit system with $\gamma=3$ and QPSK modulation.}
\label{serrho}
\end{figure}

 \begin{figure} [t]
\centering   \hspace{-0.06\columnwidth}
\subfigure[$\gamma=1.5$.] { 
\includegraphics[width=0.5\columnwidth]{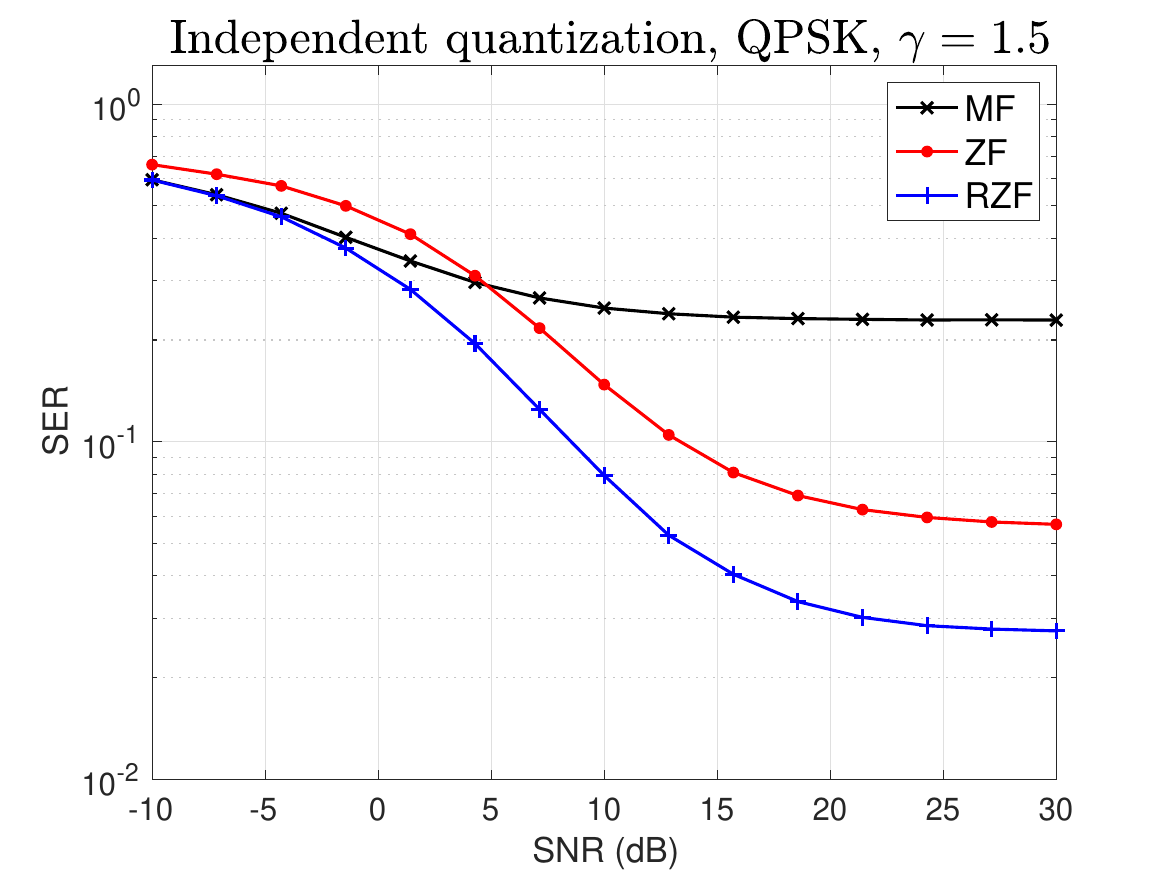}  
}     
\subfigure[$\gamma=4$.] { 
\includegraphics[width=0.5\columnwidth]{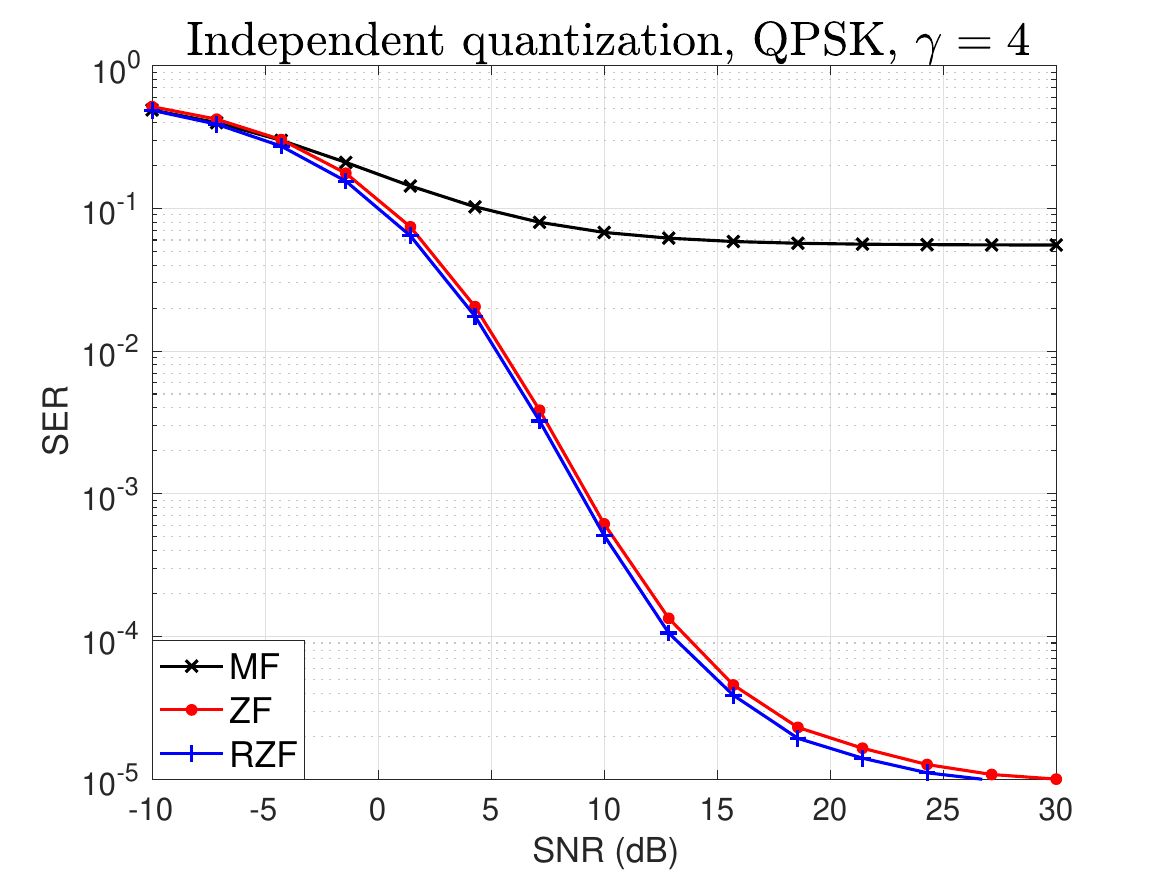}  
}    
\caption{Comparison of the SER performance between different linear-quantized precoders with $K=20$, for 2-bit independent quantization with QPSK.}
 \label{fourcom2}
\end{figure}
 \begin{figure} [t]
\centering   \hspace{-0.06\columnwidth}
\subfigure[$\gamma=1.5$.] { 
\includegraphics[width=0.5\columnwidth]{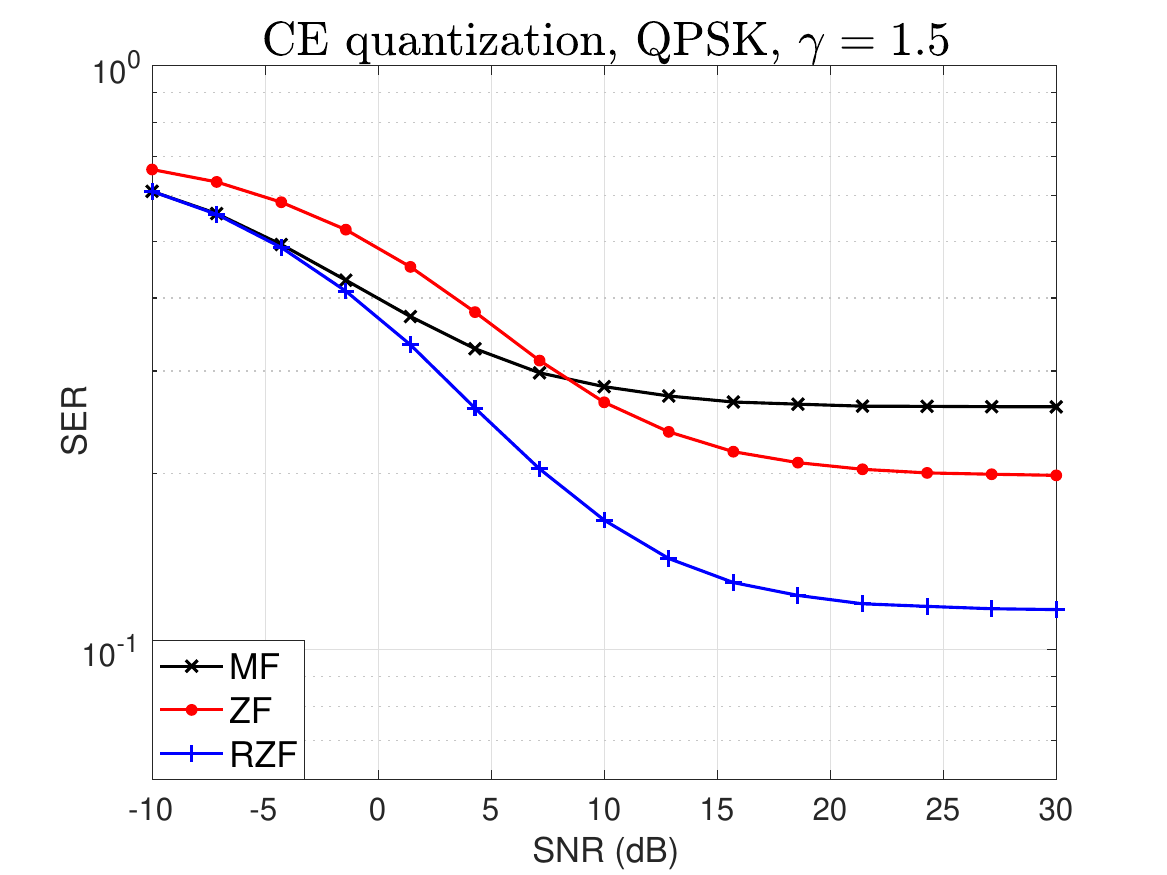}  
}     
\subfigure[$\gamma=4$.] { 
\includegraphics[width=0.5\columnwidth]{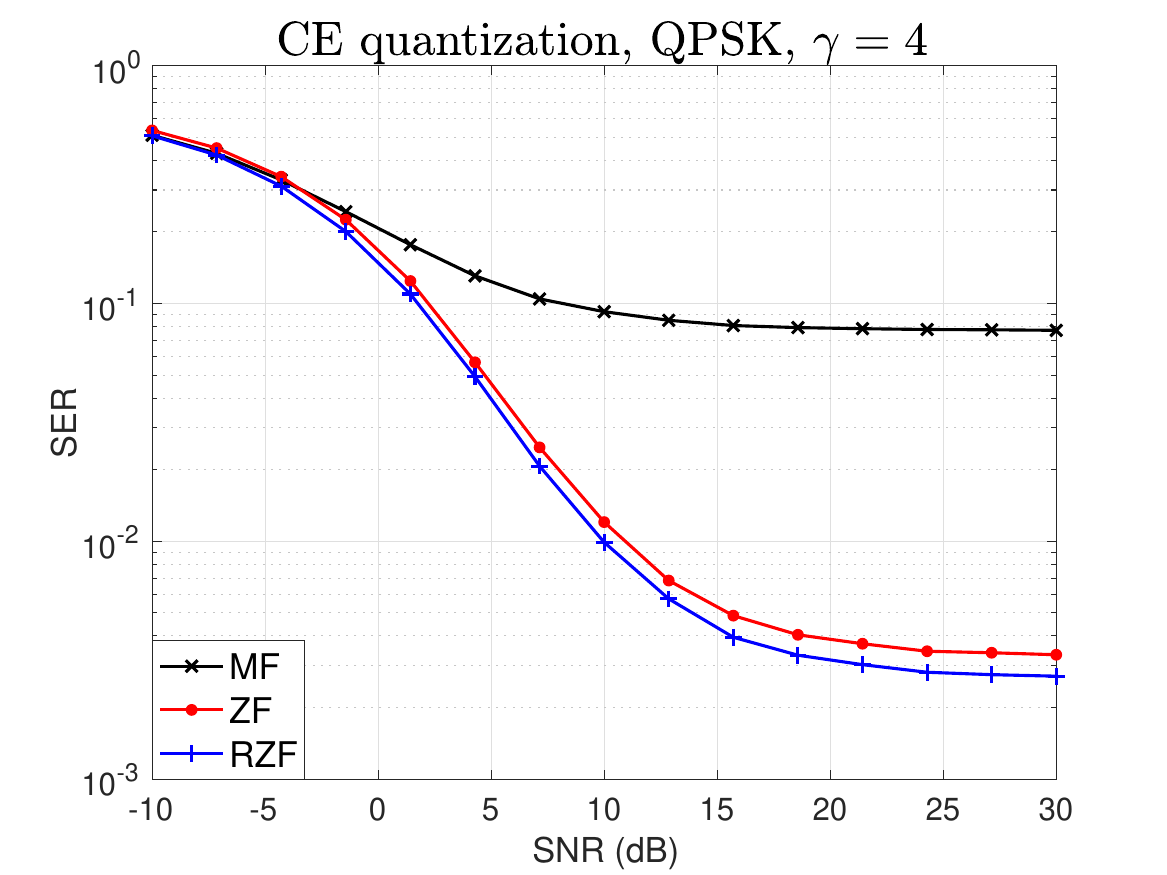}  
}    
\caption{Comparison of the SER performance between different linear-quantized precoders with $K=20$, for 2-bit CE quantization with QPSK.}
 \label{fourcom3}
\end{figure}
In Fig. \ref{serrho}, we consider quantized RZF precoding and depict the SER as a function of the regularization parameter for a one-bit system with $\gamma=3$, $\sigma=0$, and QPSK modulation. 
As shown in the figure, the errors between the asymptotic SER and the actual SER are within $0.002$ and $0.005$ for $K=100$ and $K=20$, respectively, which again validates the accuracy of our analytical results. In addition, it can be observed that the simulation curves and the (asymptotic) analysis curve show almost the same trends and all attain their minimum at {\red approximately 0.2, which agrees well with the predicted value of 0.19 for the optimal regularization parameter}.  This demonstrates the optimality of the RZF precoding matrix $\bP^*$ in \eqref{eqn:optp}.

In Figs. \ref{fourcom2} and \ref{fourcom3}, we further consider realistic systems with $K=20$. We plot the SER performance of quantized MF, ZF, and the proposed RZF precoding as a function of the SNR for both independent quantization and CE quantization. We investigate two different cases: $\gamma=1.5$ and $\gamma=4,$ which corresponds to small and large antenna-user ratio, respectively. Compared with {\red the} quantized MF and ZF precoder,  the proposed quantized RZF precoder enjoys a  {\red substantial} performance gain for small $\gamma$. When $\gamma$ is large, the proposed quantized RZF precoder {\red performs similarly to} the quantized ZF precoder (which is consistent with our discussion in \eqref{rzfapproxzf}), and they both {\red yield a much lower SER} than the quantized MF precoder. 
\section{Conclusion}\label{section7}
In this paper, we studied the performance of linear-quantized precoding in massive MIMO downlink systems. Assuming an i.i.d. Gaussian channel matrix, we showed that the linear-quantized precoding scheme is statistically equivalent to a simple scalar model in the asymptotic sense when the number of antennas $N$ and the number of users $K$ tend to infinity with a fixed ratio. We further derived the optimal precoding strategy within a class of  linear-quantized precoders, and found that it is precisely the RZF precoder where the optimal regularization parameter depends on the type and level of quantization and various system parameters.  {\color{black}For  future work, it would be interesting to extend the current analysis to encompass more general scenarios, such as general correlated channels and imperfect CSI. It is also interesting to give performance analysis of the more challenging nonlinear precoding schemes.}


\appendices
\section{Preliminaries of the  Householder Dice Technique \cite{HD}}\label{Appendix:HD}
{\color{black}In this appendix, we collect some useful results about the HD technique.}
{\color{black}We begin with the definition of the Householder transform.}
\begin{definition}[Householder Transform]
For a given vector $\v=(v_1,v_2,\dots,v_N)^\mathsf{T}\in\C^N\backslash\{\mathbf{0}\}$, denote 
$$p_\v=-e^{j\arg(v_s)},\quad \text{where }s=\min_{i}\left\{i\mid v_i\neq 0\right\}.$$
Let ${\H}(\v)$ be the Householder transform of $\v$, i.e., 
\[
{\H}(\v)=\mathbf{I}-{\red2\frac{\mathbf{u}\mathbf{u}^\mathsf{H}}{\|\mathbf{u}\|^2}}
\]
with $\mathbf{u}={\v-{\red p_\v\|\v\|\mathbf{e}_1}}$, ${\color{black}\mathbf{e}_1=(1,0,\dots,0)^\mathsf{T}}.$ Further, define 
\begin{equation}\label{reflector}
\bR(\v)=p_\v\,\H(\v).
\end{equation}
With {\red a slight abuse of notation}, $\bR(\v)$ will also be called the Householder transform matrix associated with $\v$. 
We further define a generalized Householder transform as \cite{HD}
\[
\bR_k(\v)=\left(\begin{matrix}\mathbf{I}_{k-1}&\mathbf{0}\\\mathbf{0}&\bR(\v[k:N])\end{matrix}\right).
\]
\end{definition}

The following lemma collects some useful properties of $\bR(\v)$ that will be used in the subsequent analysis. The proofs are straightforward and thus omitted.
\begin{lemma}[Facts and properties of $\bR(\v)$]\label{reflectorlemma}
For a given vector $\v\in\C^N\backslash\{\mathbf{0}\}$, the Householder transform $\bR(\v)$ defined in \eqref{reflector} satisfies:
\begin{itemize}
\item[(i)] $\bR(\v)\in\mathcal{U}(N)${\normalfont{;}}
\item[(ii)] ${\red \bR(\v)^\mathsf{H}\v=\|\v\| \mathbf{e}_1}${\normalfont{;}}
\item[(iii)] $\bR(\v)\mathbf{e}_1=\frac{\v}{\|\v\|}${, i.e., $\bR(\v)$ can be expressed as 
\begin{equation}\label{B}
\bR(\v)=\left(\begin{matrix}\frac{\v}{\|\v\|}&\mathbf{B}(\v)\end{matrix}\right),
\end{equation} where $\mathbf{B}(\v)\in\mathbb{C}^{N\times(N-1)}$ is a basis matrix for $\{\mathbf{v}\}^{\perp}$}{\normalfont{.}}
\end{itemize}
\end{lemma}


The generalized Householder transform $\mathbf{R}_k(\mathbf{v})$ has similar properties as $\mathbf{R}(\mathbf{v})$, except that it leaves the first $k-1$ entries of $\mathbf{v}$ unchanged and applies a Householder transform on $\mathbf{v}[k:N]$.

{\color{black}The following lemma is a recursive characterization of the Haar random matrix introduced in \cite[Lemma 1]{HD}, which serves as the theoretical basis of the HD technique}. The only difference is that we consider complex unitary matrices instead of real orthogonal matrices.
\begin{lemma}\label{Haar}
Let $\sg\sim\mathcal{C}\mathcal{N}(\mathbf{0},\mathbf{I}_{\color{black}N})$, $\bQ_{N-1}\sim\text{\normalfont{Haar}}(N-1)$, and $\v\in\C^N\backslash\{\mathbf{0}\}$, all of which are independent, {\color{black} where $N\geq 2$}. Then
\begin{equation}\label{construct}
\bQ_{N}:=\bR_1(\sg)\left(\begin{matrix}1&\mathbf{0}\\\mathbf{0}&\bQ_{N-1}\end{matrix}\right)\bR_1(\v)^\mathsf{H}\sim\text{\normalfont{Haar}}(N)
\end{equation}
and
\begin{equation}\label{construct2}
\widetilde{\bQ}_{N}:=\bR_1(\v)\left(\begin{matrix}1&\mathbf{0}\\\mathbf{0}&\bQ_{N-1}\end{matrix}\right)\bR_1(\sg)^\mathsf{H}\sim\text{\normalfont{Haar}}(N).
\end{equation}
Moreover, $\bQ_N$ and $\widetilde{\bQ}_N$ are independent of $\v$.
\end{lemma}
\begin{proof}
We first note that the first column of $\bR_1(\sg)$ is $\frac{\sg}{\|\sg\|}$, which is uniformly distributed on the unit sphere $\mathbb{S}^{N-1}\subseteq\C^N$. Then according to \cite[page 21]{Haarbook}, {\color{black}we have }
$$\bR_1(\sg)\left(\begin{matrix}1&\mathbf{0}\\\mathbf{0}&\bQ_{N-1}\end{matrix}\right)\sim\text{Haar}(N).$$ Moreover, since a Haar matrix is both left and right translation invariant, we are free to multiply unitary matrices (either deterministic or independent of the Haar matrix) from left or right, hence \eqref{construct} is correct. The above discussions imply that the conditional distribution of $\bQ_N$ given $\v$ is the same as the distribution of $\bQ_N$ (both are Haar distributed), i.e.,
\begin{equation}\label{independence}
\mu_{\bQ_N|\v}=\mu_{\bQ_N}=\mu,~~\forall\, \v\in\mathbb{C}^N\backslash\{\mathbf{0}\},
\end{equation} where $\mu$ denotes the Haar measure on $\mathcal{U}(N)$. Therefore, $\bQ_N$ and $\v$ are independent. 
Finally, \eqref{construct2} is also true since $\bQ\sim\text{Haar}(N)$  {\red implies}  $\bQ^\mathsf{H}\sim\text{Haar}(N)$, and the  independence between $\widetilde{\bQ}_N$ and $\v$ can be justified in a similar way as \eqref{independence}.
\end{proof}

{\color{black}\section{Proof of Theorem \ref{Equimodel}}\label{derivy}
In this section, we provide the detailed proof of Theorem \ref{Equimodel}. This section is long and is organized as follows:
\begin{itemize}
\item Section \ref{Sec:main_proof_th1} contains the main proof of Theorem \ref{Equimodel}, relying on two auxiliary results: Lemma \ref{the1:lemma1} and Lemma \ref{the1:lemma2};
\item Lemma \ref{the1:lemma1} is critical to the proof of Theorem \ref{Equimodel}. For better understanding, we present some high-level ideas about the proof of Lemma \ref{the1:lemma1} in Section \ref{discussion:the1};
\item Section \ref{proof:lemthe11} contains the complete proof of Lemma \ref{the1:lemma1};
\item Section \ref{proof:lemthe12} contains the proof of Lemma \ref{the1:lemma2}.
\end{itemize}

\subsection{Proof of Theorem \ref{Equimodel}}\label{Sec:main_proof_th1}
The proof of Theorem \ref{Equimodel} contains two major steps. The first step is to apply the HD technique \cite{HD} to our model in \eqref{Eqn:model_s1s3} to obtain a statistically equivalent model that is more amenable to analysis. The result of this step is summarized in  the following lemma and the proof is given in Appendix \ref{proof:lemthe11}.
\begin{lemma}\label{the1:lemma1}
The distribution of $(\s,\y)$ given in \eqref{Eqn:model_s1s3} is the same as {\color{black}that of} $({\s},\tilde{\y})$ specified by the following model: 
\begin{equation}\label{Eqn:model_s1s3_2}
\left\{
\begin{aligned}
\tilde{\s}_1&=f(\bD)^\mathsf{T}\bR_1(\mathsf{g}_1)\bR_1(\s)^\mathsf{H}\s,\\
\tilde{\s}_2&=q(\bR_1(\mathsf{z}_1)\bR_1(\tilde{\s}_1)^\mathsf{H}\tilde{\s}_1),\\
\tilde{\s}_3&=\bD\bR_1(\tilde{\s}_1)\bR_2(\mathsf{z}_2)\bR_1(\v_1)^\mathsf{H}\v_1,\\
\tilde{\y}&=\eta\bR_1(\s)\bR_2(\mathsf{g}_2)\bR_2(\v_2)^\mathsf{H}\v_2+\n,
\end{aligned}\right.
\end{equation}
where $\v_1=\bR_1(\mathsf{z}_1)^\mathsf{H}\tilde{\s}_2$ and $\v_2=\bR_1(\mathsf{g}_1)^\mathsf{H}\tilde{\s}_3${\normalfont{;}} $\mathsf{g}_1\sim\mathcal{CN}(\mathbf{0},\mathbf{I}_K),\mathsf{g}_2\sim\mathcal{CN}(\mathbf{0},\mathbf{I}_K), \mathsf{z}_1\sim\mathcal{CN}(\mathbf{0},\mathbf{I}_N),$ $ \mathsf{z}_2\sim\mathcal{CN}(\mathbf{0},\mathbf{I}_N)${\normalfont{;}} $\mathsf{g}_1,\mathsf{g}_2, \mathsf{z}_1, \mathsf{z}_2,\s,\n,\bD$ are mutually independent.
\end{lemma}
The second step is to simplify the above statistically equivalent model $(\s,\tilde{\y})$ using basic properties of the Householder transform in Lemma \ref{reflectorlemma} to obtain the explicit model \eqref{Equiy} in Theorem  \ref{Equimodel}. This step requires  careful calculation and we leave the details to Appendix \ref{proof:lemthe12}.
\begin{lemma}\label{the1:lemma2}
The distribution of $(\s,\tilde{\y})$ given in \eqref{Eqn:model_s1s3_2} is the same as that of $(\s,\bar{\y})$ given in \eqref{Equiy} in Theorem \ref{Equimodel}.
\end{lemma}

Combining  Lemmas  \ref{the1:lemma1} and \ref{the1:lemma2}, we get the desired result in Theorem \ref{Equimodel}.}


\subsection{Discussions on Lemma \ref{the1:lemma1}}\label{discussion:the1}

 {\color{black}
Since Lemma \ref{the1:lemma1} is critical to  the proof of Theorem \ref{Equimodel}, we would like to provide some high-level ideas and informal discussions before we present its full proof in Appendix \ref{proof:lemthe11}.

Recall that our original model in \eqref{Eqn:model_s1s3} reads
\begin{equation*}
\left\{
\begin{aligned}
\s_1&=f(\bD)^\mathsf{T}\bU^\mathsf{H}\s;\\
\s_2&=q(\bV \s_1);\\
\s_3&=\bD\bV^\mathsf{H}\s_2;\\
\y&=\bU\s_3+\n.
\end{aligned}\right.
\end{equation*} 
Our goal is to characterize the joint distribution of $(\mathbf{s},\mathbf{s}_1, \mathbf{s}_2, \mathbf{s}_3,\mathbf{y})$. Roughly speaking, for $N,K\ge3$, fixing the distribution of $(\mathbf{s},\mathbf{s}_1, \mathbf{s}_2, \mathbf{s}_3,\mathbf{y})$ does not fully fix the randomness of the Haar matrices $\bU$ and $\bV$, and we still have freedom to generate the remaining randomness in a convenient way. A systematic way of carrying out this process is via the HD technique in \cite{HD}. \textit{To be clear, we use $(\mathbf{s},\mathbf{s}_1, \mathbf{s}_2, \mathbf{s}_3,\mathbf{y})$ to denote the random vectors from the original model, and use $(\mathbf{s},\tilde{\mathbf{s}}_1, \tilde{\mathbf{s}}_2, \tilde{\mathbf{s}}_3,\tilde{\mathbf{y}})$ to denote the corresponding vectors generated via the HD method in the following discussion.} 

The main idea of HD is to  generate the Haar random matrices $\bU$ and $\bV$ involved in the above iterations in a recursive way by repeatedly applying Lemma \ref{Haar}. More specifically, the HD technique tells that at each iteration we only need to generate a single Gaussian vector to unfold the randomness of $\bU$ (or $\bV$) in one dimension,  and the resulting sequence will only depend on the initial condition $\{\s,\bD,\n\}$ and the exposed Gaussian vectors. 

Here we take the first iteration of \eqref{Eqn:model_s1s3} as an example to shed some light on this. To compute $\bU^{\mathsf{H}}\s$, we can construct a Haar random matrix $\bU^\mathsf{H} $ according to Lemma \ref{Haar} as 
\begin{equation}\label{U2}
\widetilde{\bU}^\mathsf{H}=\bR_1(\mathsf{g}_1)\left(\begin{matrix}1&\mathbf{0}\\\mathbf{0}&\bQ_{K-1}\end{matrix}\right)\bR_1(\s)^\mathsf{H},
\end{equation}
where $\sg_1\sim\mathcal{C}\mathcal{N}(\mathbf{0},\mathbf{I}_K)$ and $\bQ_{K-1}\sim\text{Haar}(K-1)$ are  independent  of each other and of $\s$, $\bD$, and $\n$. Then we have 
$$
\begin{aligned}
\tilde{\s}_1&:=f(\bD)^\mathsf{T}\widetilde{\bU}^\mathsf{H}\s\\
&=f(\bD)^\mathsf{T}\bR_1(\sg_1)\left(\begin{matrix}1&\mathbf{0}\\\mathbf{0}&\bQ_{K-1}\end{matrix}\right)\bR_1(\s)^\mathsf{H}\s\\
&=f(\bD)^\mathsf{T}\bR_1(\sg_1)\bR_1(\s)^\mathsf{H}\s,
\end{aligned}$$
where the last equality holds since $\bR_1(\s)^\mathsf{H}\s$ is only nonzero in its first element due to Lemma \ref{reflectorlemma}. \textit{An important observation is that $\tilde{\s}_1$ depends on the Haar matrix $\widetilde{\mathbf{U}}$ only through a Gaussian vector $\sg_1$, and is invariant to the remaining Haar matrix $\mathbf{Q}_{K-1}$.} Since $\bQ_{K-1}$ involved in \eqref{U2} is Haar distributed and independent of all the other random variables generated up to this point, we can apply the same technique to $\bQ_{K-1}$ when another multiplication involving $\bU$ is required, and a new Gaussian vector $\sg_2\sim\mathcal{CN}(\mathbf{0},\mathbf{I}_K)$ will be  exposed (see the expression of $\tilde{\y}$ in \eqref{Eqn:model_s1s3_2}).
The Haar matrix $\bV$ can be constructed similarly when we deal with the second and third iterations, and two Gaussian vectors $\sz_1\sim\mathcal{CN}(\mathbf{0},\mathbf{I}_N)$ and $\sz_2\sim\mathcal{CN}(\mathbf{0},\mathbf{I}_N)$ will be exposed in these two iterations after applying the HD technique (see the expression of $\tilde{\s}_2$ and $\tilde{\s}_3$ in \eqref{Eqn:model_s1s3_2}).  A detailed derivation of \eqref{Eqn:model_s1s3} using the HD technique is provided in Appendix \ref{EquiyA1}.

To gain some further insight,  we directly give the form of the two Haar matrices constructed using the HD technique without proof (see Appendix \ref{EquiyA1} for a detailed proof): 
\begin{equation}\label{expression:U}
\widetilde{\bU}=\bR_1(\s)\bR_2(\mathsf{g}_2)\left(\begin{matrix}\mathbf{I}_2&\mathbf{0}\\\mathbf{0}&\bQ_{K-2}\end{matrix}\right)\bR_2(\v_2)^\mathsf{H}\bR_1(\mathsf{g}_1)^\mathsf{H}
\end{equation}
 and
\begin{equation}\label{expression:V}
\widetilde{\bV}=\bR_1(\sz_1)\bR_2(\v_1)\left(\begin{matrix}\mathbf{I}_2&\mathbf{0}\\\mathbf{0}&\bP_{N-2}\end{matrix}\right)\bR_2(\sz_2)^\mathsf{H}\bR_1(\tilde{\s}_1)^\mathsf{H}.
\end{equation}
In the above expressions, $\bQ_{K-2}$ and $\bP_{N-2}$ are Haar matrices independent of all the other random variables, which are the unexposed random matrices that are absent in the final result; $\mathsf{g}_1,\mathsf{g}_2,\mathsf{z}_1,\mathsf{z}_2$ are the exposed Gaussian vectors; $\tilde{\s}_1,\v_1,$ and $\v_2$ are  some intermediate random vectors generated due to the recursive nature of the HD technique. 
Equation \eqref{Eqn:model_s1s3_2} can be interpreted as replacing the Haar random matrices $\mathbf{U}$ and $\mathbf{V}$ in the original model \eqref{Eqn:model_s1s3} by the two unitary matrices $\widetilde{\bU}$ and $\widetilde{\bV}$ given above. }
Here, we use the notation $\widetilde{\mathbf{U}}$ and $\widetilde{\mathbf{V}}$ to emphasize the fact that their distributional properties are yet to be proved. To show $(\s,\y)\overset{d}{=}(\s,\tilde{\y})$, it remains to check that $\widetilde{\bU}$ and $\widetilde{\bV}$ have the desired properties, i.e., $\widetilde{\bU}\sim\text{Haar}(K), \widetilde{\bV}\sim\text{Haar}(N)$, and $\widetilde{\bU},\widetilde{\bV},\bD, \s,\n$ are mutually independent. We relegate the details to Appendix \ref{EquiyA2}.

\subsection{Proof of Lemma \ref{the1:lemma1}}\label{proof:lemthe11}
{\color{black}In this subsection, we give the complete proof of Lemma \ref{the1:lemma1}, which consists of the following two steps:
\begin{itemize}
\item We first show that \eqref{Eqn:model_s1s3_2} is the sequence generated by applying the HD technique to \eqref{Eqn:model_s1s3};
\item We then prove \eqref{Eqn:model_s1s3_2} is statistically equivalent to \eqref{Eqn:model_s1s3}. 
\end{itemize}}
\subsubsection{Derivation of \eqref{Eqn:model_s1s3_2}}\label{EquiyA1}
We first give a detailed derivation of how \eqref{Eqn:model_s1s3_2} is obtained  via the HD technique. Recall that our original model \eqref{sysmodel1} can be written as the following  iterative process (see \eqref{Eqn:model_s1s3}):
\begin{equation*}
\left\{
\begin{aligned}
\s_1&=f(\bD)^\mathsf{T}\bU^\mathsf{H}\s;\\
\s_2&=q(\bV \s_1);\\
\s_3&=\bD\bV^\mathsf{H}\s_2;\\
\y&=\bU\s_3+\n.
\end{aligned}\right.
\end{equation*} 
Next, we apply the HD technique to deal with the above model.

{\color{black}For the first iteration, 
we construct $\bU^\mathsf{H}$ according to Lemma \ref{Haar} as}
\begin{equation}\label{U}
\left(\bU^{(1)}\right)^\mathsf{H}=\bR_1(\mathsf{g}_1)\left(\begin{matrix}1&\mathbf{0}\\\mathbf{0}&\bQ_{K-1}\end{matrix}\right)\bR_1(\s)^\mathsf{H},
\end{equation}
where $\sg_1\sim\mathcal{C}\mathcal{N}(\mathbf{0},\mathbf{I}_K)$ and $\bQ_{K-1}\sim\text{Haar}(K-1)$ are  independent {\red of each other and of} $\s$, $\bD$, and $\n$. Then we have 
$$
\begin{aligned}
\tilde{\s}_1&=f(\bD)^\mathsf{T}\left(\bU^{(1)}\right)^\mathsf{H}\s\\
&=f(\bD)^\mathsf{T}\bR_1(\sg_1)\left(\begin{matrix}1&\mathbf{0}\\\mathbf{0}&\bQ_{K-1}\end{matrix}\right)\bR_1(\s)^\mathsf{H}\s\\
&=f(\bD)^\mathsf{T}\bR_1(\sg_1)\bR_1(\s)^\mathsf{H}\s,
\end{aligned}$$
where the last equality holds since $\bR_1(\s)^\mathsf{H}\s$ is only nonzero in its first element.

{\color{black}For the second iteration,} we use a similar technique to construct  the Haar matrix $\bV$ according to Lemma \ref{Haar} as
\begin{equation}\label{eq:V1}
\bV^{(1)}=\bR_1(\sz_1)\left(\begin{matrix}1&\mathbf{0}\\\mathbf{0}&\bP_{N-1}\end{matrix}\right)\bR_1(\tilde{\s}_1)^\mathsf{H},
\end{equation}
where $\sz_1\sim\mathcal{C}\mathcal{N}(\mathbf{0},\mathbf{I}_N)$ and $\bP_{N-1}\sim\text{Haar}(N-1)$ are {\red independent of each other and of} all the existing random variables. It follows immediately that 
$$
\begin{aligned}
\bV^{(1)}\tilde{\s}_1&=\bR_1(\sz_1)\left(\begin{matrix}1&\mathbf{0}\\\mathbf{0}&\bP_{N-1}\end{matrix}\right)\bR_1(\tilde{\s}_1)^\mathsf{H}\tilde{\s}_1\\
&=\bR_1(\sz_1)\bR_1(\tilde{\s}_1)^\mathsf{H}\tilde{\s}_1,
\end{aligned}$$
and thus
\begin{equation*}\label{tildes2}
\tilde{\s}_2=q(\bV^{(1)}\tilde{\s}_1){=}q(\bR_1(\sz_1)\bR_1(\tilde{\s}_1)^\mathsf{H}\tilde{\s}_1).
\end{equation*}

 {\color{black}  For the third iteration, we need to  calculate $\bD\bV^\mathsf{H}\tilde{\s}_2$}. From \eqref{eq:V1}, we get 
$$\left(\bV^{(1)}\right)^\mathsf{H}=\bR_1(\tilde{\s}_1)\left(\begin{matrix}1&\mathbf{0}\\\mathbf{0}&\bP_{N-1}^\mathsf{H}\end{matrix}\right)\bR_1(\sz_1)^\mathsf{H}.$$ 
Let $\v_1=\bR_1(\sz_1)^\mathsf{H}\tilde{\s}_2$. Since  $\bP_{N-1}^\mathsf{H}$ is Haar distributed and independent of $\v_1$, we can still apply the above technique to construct $\bP_{N-1}^\mathsf{H}$ as 
$$\bP_{N-1}^\mathsf{H}=\bR_1(\sz_2[2:N])\left(\begin{matrix}1&\mathbf{0}\\\mathbf{0}&\bP_{N-2}^\mathsf{H}\end{matrix}\right)\bR_1(\v_1[2:N])^\mathsf{H},$$ where 
 $\sz_2\sim\mathcal{C}\mathcal{N}(\mathbf{0},\mathbf{I}_N)$ and $\bP_{N-2}\sim\text{Haar}(N-2)$ are {\red independent of each other and of} all the existing random variables. Then we have 
$$\left(\bV^{(2)}\right)^\mathsf{H}=\bR_1(\tilde{\s}_1)\bR_2(\sz_2)\left(\begin{matrix}\mathbf{I}_2&\mathbf{0}\\\mathbf{0}&\bP_{N-2}^\mathsf{H}\end{matrix}\right)\bR_2(\v_1)^\mathsf{H}\bR_1(\sz_1)^\mathsf{H}$$
 and 
 $$
 \begin{aligned}
 \tilde{\s}_3&=\bD \left(\bV^{(2)}\right)^\mathsf{H} \tilde{\s}_2\\
 &=\bD \bR_1(\tilde{\s}_1)
 \bR_2(\sz_2)\left(\begin{matrix}\mathbf{I}_2&\mathbf{0}\\\mathbf{0}&\bP_{N-2}^\mathsf{H}\end{matrix}\right)\bR_2(\v_1)^\mathsf{H}\v_1\\
 &=\bD \bR_1(\tilde{\s}_1)\bR_2(\sz_2)\bR_2(\v_1)^\mathsf{H}\v_1,
 \end{aligned}
 $$
  where $\bP_{N-2}^\mathsf{H}$ disappears in the second equality since $\bR_2(\v_1)^\mathsf{H}\v_1$ is only nonzero in its first two elements.
  
 Finally, we calculate $\tilde{\y}=\bU\tilde{\s}_3+\n$. According to \eqref{U},
 $$\bU^{(1)}=\bR_1(\s)\left(\begin{matrix}1&\mathbf{0}\\\mathbf{0}&\bQ_{K-1}^\mathsf{H}\end{matrix}\right)\bR_1(\sg_1)^\mathsf{H}.$$ Similarly, let $\v_2=\bR_1(\sg_1)^\mathsf{H}\tilde{\s}_3$ and {\color{black} construct $\bQ_{K-1}^\mathsf{H}$ as}
 $$\bQ_{K-1}^\mathsf{H}=\bR_1(\sg_2[2:K])\left(\begin{matrix}1&\mathbf{0}\\\mathbf{0}&\bQ_{K-2}\end{matrix}\right)\bR_1(\v_2[2:K])^\mathsf{H},$$ 
where $\sg_2\sim\mathcal{C}\mathcal{N}(\mathbf{0},\mathbf{I}_K)$ and  $\bQ_{K-2}\sim\text{Haar}(K-2)$ are independent {\red of each other and of} all the existing random variables.  Then we have 
$$\bU^{(2)}=\bR_1(\s)\bR_2(\sg_2)\left(\begin{matrix}\mathbf{I}_2&\mathbf{0}\\\mathbf{0}&\bQ_{K-2}\end{matrix}\right)\bR_2(\v_2)^\mathsf{H}\bR_1(\sg_1)^\mathsf{H}$$ and 
$$
\begin{aligned}
\tilde{\y}&=\bU^{(2)}\tilde{\s}_3+\n\\
&=\bR_1(\s)\bR_2(\sg_2)\left(\begin{matrix}\mathbf{I}_2&\mathbf{0}\\\mathbf{0}&\bQ_{K-2}\end{matrix}\right)\bR_2(\v_2)^\mathsf{H}\bR_1(\sg_1)^\mathsf{H}\tilde{\s}_3+\n\\
&=\bR_1(\s)\bR_2(\sg_2)\bR_2(\v_2)^\mathsf{H}\v_2+\n.
\end{aligned}
$$ 
This gives the sequence $(\tilde{\s}_1,\tilde{\s}_2,\tilde{\s}_3,\tilde{\y})$ in \eqref{Eqn:model_s1s3_2}, and the two constructed Haar matrices  are $\widetilde{\bU}=\bU^{(2)}$ and $\widetilde{\bV}=\bV^{(2)}$, which are exactly those given in \eqref{expression:U} and \eqref{expression:V}.

\subsubsection{Statistical Equivalence of \eqref{Eqn:model_s1s3_2} and \eqref{Eqn:model_s1s3}}\label{EquiyA2}

{\color{black}The proof follows the general principle proposed in \cite[Theorem 2]{HD}.  Here we provide a complete proof to make the paper self-contained.}

First, it is easy to check that 
$(\s,\tilde{\y})$ given in \eqref{Eqn:model_s1s3_2} can be obtained by substituting $\widetilde{\bU}$ in \eqref{expression:U} and $\widetilde{\bV}$ in \eqref{expression:V} into \eqref{Eqn:model_s1s3}. To show the statistical equivalence, we still need to prove that $\widetilde{\bU}$ and $\widetilde{\bV}$ in \eqref{expression:U} and \eqref{expression:V} have the following desired properties:  
 \begin{subequations}\label{eqn:propertyuv}
 \begin{align}
  &\hspace{-7.3cm}\bullet\,\widetilde{\bU}\sim\text{Haar}(K), \widetilde{\bV}\sim \text{Haar}(N);\\
 &\hspace{-7.3cm}\bullet\,\widetilde{\bU}, \widetilde{\bV}, \s, \bD, \n \text{are mutually independent}. 
 \end{align}
  \end{subequations}
Next, we prove the above properties for $\widetilde{\bU}$, and those for $\widetilde{\bV}$ can be proved by similar arguments.
We first analyze the {\color{black} inner} term $\bR_2(\sg_2)\left(\begin{smallmatrix}\mathbf{I}_2&\mathbf{0}\\\mathbf{0}&\bQ_{K-2}\end{smallmatrix}\right)\bR_2(\v_2)^\mathsf{H}$ {\color{black} in $\widetilde{\bU}$}:
$$
\begin{aligned}
\bR_2(\sg_2)\left(\begin{matrix}\mathbf{I}_2&\mathbf{0}\\\mathbf{0}&\bQ_{K-2}\end{matrix}\right)\bR_2(\v_2)^\mathsf{H}
=&\left(\begin{matrix}1&\mathbf{0}\\\mathbf{0}&\bR_1(\sg_2[2:K])\left(\begin{smallmatrix}1&\mathbf{0}\\\mathbf{0}&\bQ_{K-2}\end{smallmatrix}\right)\bR_1(\v_2[2:K])^\mathsf{H}\end{matrix}\right)\\:=&\left(\begin{matrix}1&\mathbf{0}\\\mathbf{0}&\bQ_{K-1}\end{matrix}\right).
\end{aligned}
$$
 By the definition of $\v_2$ and from the {\red method} of generating $\bQ_{K-2}$ and $\sg_2$, it is clear that $\sg_2, \bQ_{K-2}$, and $\v_2$ are mutually independent.  It follows immediately from Lemma \ref{Haar} that $\bQ_{K-1}\sim\text{Haar}(K-1)$ and is independent of $\v_2$, and hence independent of all other random  variables except $\bQ_{K-2}$ and $\sg_2$ that construct $\bQ_{K-1}$. We then investigate $$\tilde{\bU}=\bR_1(\s)\left(\begin{matrix}1&\mathbf{0}\\\mathbf{0}&\bQ_{K-1}\end{matrix}\right)\bR_1(\sg_1)^\mathsf{H}.$$
 Again from Lemma \ref{Haar}, we know that $\widetilde{\bU}\sim\text{Haar}(K)$ and is independent of all other random variables (except $\bQ_{K-2}$, $\sg_2$, and $\sg_1$), {\color{black}i.e., $\widetilde{\bU}$ has the desired properties in \eqref{eqn:propertyuv}}.   This completes our proof of the statistical equivalence between $(\s,\y)$ in \eqref{Eqn:model_s1s3} and $(\s,\tilde{\y})$ in \eqref{Eqn:model_s1s3_2}.
 
\subsection{Proof of Lemma \ref{the1:lemma2}}\label{proof:lemthe12}
In this subsection, we give the proof of Lemma \ref{the1:lemma2}, i.e., we derive \eqref{Equiy} from  \eqref{Eqn:model_s1s3_2}.
For clarity, we copy the expressions of $\tilde{\s}_1,\tilde{\s}_2,\tilde{\s}_3,\tilde{\y}$ in \eqref{Eqn:model_s1s3_2} here.
\begin{equation*}
\left\{
\begin{aligned}
\tilde{\s}_1&=f(\bD)^\mathsf{T}\bR_1(\mathsf{g}_1)\bR_1(\s)^\mathsf{H}\s;\\
\tilde{\s}_2&=q(\bR_1(\mathsf{z}_1)\bR_1(\tilde{\s}_1)^\mathsf{H}\tilde{\s}_1);\\
\tilde{\s}_3&=\bD\bR_1(\tilde{\s}_1)\bR_2(\mathsf{z}_2)\bR_1(\v_1)^\mathsf{H}\v_1;\\
\tilde{\y}&=\eta\bR_1(\s)\bR_2(\mathsf{g}_2)\bR_2(\v_2)^\mathsf{H}\v_2+\n,
\end{aligned}\right.
\end{equation*}
where $\v_1=\bR_1(\mathsf{z}_1)^\mathsf{H}\tilde{\s}_2$ and $\v_2=\bR_1(\mathsf{g}_1)^\mathsf{H}\tilde{\s}_3.$
In the subsequent analysis, we will frequently encounter Householder transform-vector multiplications of the following form:
$$\bR_1(\sg)\bR_1(\v)^\mathsf{H}\v,$$ where $\sg\sim\mathcal{C}\mathcal{N}(\mathbf{0},\mathbf{I})$ and  $\v$ is   a random vector independent of $\sg$. According to the properties of the Householder transform, i.e., (ii) and (iii) of Lemma \ref{reflectorlemma}, we have 
 \begin{equation}\label{eq8}
 \bR_1(\sg)\bR_1(\v)^\mathsf{H}\v=\|\v\|\,\bR_1(\sg)\mathbf{e}_1=\frac{\|\v\|}{\|\sg\|}\,  \sg.
 \end{equation}
It follows immediately that 
\begin{equation}\label{hats1}
\begin{aligned}\tilde{\s}_1&=f(\bD)^{\mathsf{T}}\bR_1(\sg_1)\bR_1(\s)^\mathsf{H}\s=\frac{\|\s\|}{\|\sg_1\|}\, f(\bD)^\mathsf{T}\sg_1=\hat{\s}_1,
\end{aligned}
\end{equation}
where the last equality is due to the definition of $\hat {\s}_1$ in \eqref{Equiy_2}. 

Next we begin our derivation of \eqref{Equiy}. {\red First}, we compute $\bR_2(\sg_2)\bR_2(\v_2)^\mathsf{H}\v_2$ in  $\tilde{\y}$ using \eqref{eq8}:
\begin{equation*}\label{iny}
\begin{aligned}
\bR_2(\sg_2)\bR_2(\v_2)^\mathsf{H}\v_2=&\left(\begin{matrix}1&\mathbf{0}\\\mathbf{0}&\bR_1(\sg_2[2:K])\bR_1(\v_2[2:K])^\mathsf{H}\end{matrix}\right)\left(\begin{matrix}\v_2[1]\\\v_2[2:K]\\ \end{matrix}\right)\\
=&\left(\begin{matrix} \v_2[1]\\ \bR_1(\sg_2[2:K])\bR_1(\v_2[2:K])^\mathsf{H}\v_2[2:K]\end{matrix}\right)\\
=&\left(\begin{matrix}\v_2[1]\\
\frac{\|\v_2[2:K]\|}{\|\sg_2[2:K]\|}\, \sg_2[2:K]
\end{matrix}\right).
\end{aligned}
\end{equation*}
\hspace{-0.1cm}{\color{black}Combining the above equation with the definition $\bR_1(\bs)=\left(\frac{\s}{\|\s\|}~~\bB(\s)\right)$, we can express $\tilde{\y}$ as  }
\begin{equation}\label{tildey2}
\begin{aligned} 
\tilde{\y}&=\eta\bR_1(\s)\bR_2(\sg_2)\bR_2(\v_2)^\mathsf{H}\v_2+\n\\
&=\eta\left(\begin{matrix}\frac{\s}{\|\s\|}&\bB(\s)\end{matrix}\right)\left(\begin{matrix}\v_2[1]\\\frac{\|\v_2[2:K]\|}{\|\sg_2[2:K]\|}\, \sg_2[2:K]\end{matrix}\right)+\n\\
&=\eta\frac{\v_2[1]}{\|\s\|} \s+\eta\frac{\|\v_2[2:K]\|}{\|\sg_2[2:K]\|} \bB(\s)\sg_2[2:K]+\n.
\end{aligned}
\end{equation}
{\color{black}By  further noting that} $$\mathbf{B}(\bs)\sg_2[2:K]=\mathbf{R}(\bs)\sg_2-\frac{\sg_2[1]}{\|\bs\|}\,\bs$$ {\color{black}and substituting it} into \eqref{tildey2}, we get
\begin{equation}\label{tildey3}
\begin{aligned} 
\tilde{\y}=&\eta\left(\frac{\v_2[1]}{\|\s\|}-\frac{\sg_2[1]}{\|\bs\|}\,\frac{\|\v_2[2:K]\|}{\|\sg_2[2:K]\|} \right)\s+\eta\frac{\|\v_2[2:K]\|}{\|\sg_2[2:K]\|} 
\bR(\bs)\sg_2+\n.\\
\end{aligned}
\end{equation}
Note that if $\sg\sim\mathcal{C}\mathcal{N}(\mathbf{0},\mathbf{I})$, then for any unitary matrix $\mathbf{U}$ independent of $\sg$, $\mathbf{U}\sg\sim\mathcal{C}\mathcal{N}(\mathbf{0},\mathbf{I})$, i.e., $\mathbf{U}\sg\overset{d}{=}\sg$, and $\mathbf{U}\sg$ is independent of $\mathbf{U}$. 
Therefore, since $\sg_2$ is independent of all other random variables in \eqref{tildey3}, i.e., $\s,\v_2,$ and $\n$, 
  $\bR(\s)^{-1}\sg_2\overset{d}{=}\sg_2$ and $\bR(\s)^{-1}\sg_2$ is also independent of  $\bs$, $\v_2$ and $\n$, and hence we can replace $\sg_2$ with $\bR(\s)^{-1}\sg_2$ in \eqref{tildey3}, which yields 
 \begin{equation}\label{tildey}
\begin{aligned} 
\tilde{\y}\overset{d}{=}\,&\eta\left(\frac{\v_2[1]}{\|\s\|}-\frac{(\bR(\bs)^{-1}\sg_2)[1]}{\|\bs\|}\, \frac{\|\v_2[2:K]\|}{\|(\bR(\bs)^{-1}\sg_2)[2:K]\|}\right) \s+\eta\frac{\|\v_2[2:K]\|}{\|(\bR(\bs)^{-1}\sg_2)[2:K]\|} \sg_2+\n.
\end{aligned}
\end{equation}
  We next analyze 
  \begin{equation}\label{v2_2}
  \begin{aligned}\v_2&=\bR_1(\sg_1)^\mathsf{H}\tilde{\s}_3=\bR_1(\sg_1)^\mathsf{H}\bD\bR_1(\tilde{\s}_1)\bR_2(\sz_2)\bR_2(\v_1)^\mathsf{H}\v_1
  \end{aligned}
  \end{equation} step by step.
{\red First},  from the definitions of $\v_1$ and $\tilde \s_2,$ we have
\begin{equation}\label{v1}
\begin{aligned}\v_1&=\bR_1(\sz_1)^\mathsf{H}q(\bR_1(\sz_1)\bR_1(\tilde{\s}_1)^\mathsf{H}\tilde{\s}_1)\\
&\overset{}{=}\bR_1(\sz_1)^\mathsf{H} q\left(\frac{\|\tilde{\s}_1\|}{\|\sz_1\|}\, \sz_1\right)~~(\text{from \eqref{eq8})}\\
&\overset{}{=}\bR_1(\sz_1)^\mathsf{H} q\left(\frac{\|\hat{\s}_1\|}{\|\sz_1\|} \,\sz_1\right)~~(\text{from \eqref{hats1}})\\
&\overset{}{=}\left(\begin{matrix}\frac{\sz_1^\mathsf{H}}{\|\sz_1\|}\\\bB(\sz_1)^\mathsf{H}\end{matrix}\right)q\left(\frac{\|\hat{\s}_1\|}{\|\sz_1\|}\, \sz_1\right)~~(\text{from \eqref{B})}\\
&=\left(\begin{matrix}\frac{\sz_1^\mathsf{H}q\left(\frac{\|\hat{\s}_1\|}{\|\sz_1\|} \sz_1\right)}{\|\sz_1\|}\\\bB(\sz_1)^\mathsf{H}q\left(\frac{\|\hat{\s}_1\|}{\|\sz_1\|} \sz_1\right)\end{matrix}\right).
\end{aligned}
\end{equation}
Then we have
$$\begin{aligned}
\bR_2(\sz_2)\bR_2(\v_1)^\mathsf{H}\v_1
=&\left(\begin{matrix}1&\mathbf{0}\\\mathbf{0}&\bR_1(\sz_2[2:N])\bR_1(\v_1[2:N])^\mathsf{H}\end{matrix}\right)\left(\begin{matrix}\v_1[1]\\\v_1[2:N]\end{matrix}\right)\\
=&\left(\begin{matrix}\v_1[1]\\\bR_1(\sz_2[2:N])\bR_1(\v_1[2:N])^\mathsf{H}\v_1[2:N]\end{matrix}\right)\\
=&\left(\begin{matrix}
\frac{\sz_1^\mathsf{H}q\left(\frac{\|\hat{\s}_1\|}{\|\sz_1\|} \sz_1\right)}{\|\sz_1\|}\\\frac{\left\|\bB(\sz_1)^\mathsf{H}q\left(\frac{\|\hat{\s}_1\|}{\|\sz_1\|} \sz_1\right)\right\|}{\|\sz_2[2:N]\|}\,\sz_2[2:N]\end{matrix}\right),
\end{aligned}
$$ where the last equality comes from \eqref{eq8} and \eqref{v1}.
It follows from the above equality, \eqref{B}, and \eqref{hats1} that 
\begin{equation*}\label{eqv2_1}
\begin{aligned}
&\bR_1(\tilde{\bs}_1)\bR_2(\sz_2)\bR_2(\v_1)^\mathsf{H}\v_1\\
=\,&\bR_1(\hat{\bs}_1)\bR_2(\sz_2)\bR_2(\v_1)^\mathsf{H}\v_1\\
=\,&\left(\begin{matrix}\frac{\hat{\bs}_1}{\|\hat{\bs}_1\|}&\bB(\hat{\bs}_1)\end{matrix}\right)\left(\begin{matrix}
\frac{\sz_1^\mathsf{H}q\left(\frac{\|\hat{\s}_1\|}{\|\sz_1\|}\sz_1\right)}{\|\sz_1\|}\\\frac{\left\|\bB(\sz_1)^\mathsf{H}q\left(\frac{\|\hat{\s}_1\|}{\|\sz_1\|} \sz_1\right)\right\|}{\|\sz_2[2:N]\|}\,\sz_2[2:N]\end{matrix}\right)\\
=\,&\frac{\sz_1^\mathsf{H}q\left(\frac{\|\hat{\s}_1\|}{\|\sz_1\|}\, \sz_1\right)}{\|\sz_1\|\|\hat{\bs}_1\|}\, \hat{\bs}_1
+\frac{\left\|\bB(\sz_1)^\mathsf{H}q\left(\frac{\|\hat{\s}_1\|}{\|\sz_1\|}\, \sz_1\right)\right\|}{\|\sz_2[2:N]\|}\, \bB(\hat{\s}_1)\sz_2[2:N]\\
:=\,&{C_1}\hat{\s}_1+{C_2} \bB(\hat{\s}_1)\sz_2[2:N].
\end{aligned}
\end{equation*}
Finally, from the above equality, \eqref{B}, and \eqref{v2_2}, we have 
\begin{equation}\label{v2}
\begin{aligned}
\v_2&=\left(\begin{matrix}
\frac{\sg_1^\mathsf{H}}{\|\sg_1\|}\\\bB(\sg_1)^\mathsf{H}\end{matrix}\right)( {C_1}\bD \hat{\s}_1+{C_2} \bD \bB(\hat{\s}_1)\sz_2[2:N])\\
&=\left(\begin{matrix}\frac{\sg_1^\mathsf{H}\{{C_1}\bD \hat{\s}_1+{C_2} \bD \bB(\hat{\s}_1)\sz_2[2:N])\}}{\|\sg_1\|}\\ \bB(\sg_1)^\mathsf{H}\{ {C_1}\bD\hat{\s}_1+{C_2} \bD \bB(\hat{\s}_1)\sz_2[2:N]\}\end{matrix}\right).
\end{aligned}
\end{equation}
Plugging \eqref{v2} into \eqref{tildey}, we can get the desired model \eqref{Equiy}.

\section{Proof of Theorem \ref{asy}}\label{asymptotic}

In this section, we provide a detailed proof of Theorem \ref{asy}. {\color{black}In Section \ref{Sec:pre}, we first give some definitions and preliminary results of the asymptotic analysis. Section \ref{Sec:aux} gives two useful auxiliary results that are important  for the proof, and 
 Section \ref{Sec:main_proof} contains the main proof of Theorem \ref{asy}.}

\subsection{Preliminaries}\label{Sec:pre}


\begin{lemma}\label{as}
Let $\{X_N\}$ and  $\{Y_N\}$ be sequences of random variables. If 
$$X_N\xrightarrow{a.s.}X,\quad Y_N\xrightarrow{a.s.} Y,$$
then  
$$X_N+Y_N\xrightarrow{a.s.}X+Y,\quad X_N Y_N\xrightarrow{a.s.}XY, \quad$$ and  $$ \frac{X_N}{Y_N}\xrightarrow{a.s.}\frac{X}{Y}~(\text{if } Y_N, Y\neq 0).$$
\end{lemma}
\begin{lemma}[Kolmogorov’s strong law of large numbers \cite{probability}]\label{largenumber}
Assume that $X_1, X_2, \dots$ are independent with means $\mu_1, \mu_2, \dots$ and variances $\sigma_1^2, \sigma_2^2, \dots$ such that $\sum_N\frac{\sigma_N^2}{N^2} < \infty.$ Then
\begin{equation*}
\begin{aligned} \frac{X_1 + X_2 + \dots + X_N - (\mu_1 + \mu_2 + \dots + \mu_N)}{N} \xrightarrow{a.s.} 0. \end{aligned}
\end{equation*}
As a corollary, if $X_1,X_2,\dots$ are i.i.d. with mean $\mu$, then 
$$\frac{X_1+X_2+\cdots+X_N}{N}-\mu\xrightarrow{a.s.} 0.$$
\end{lemma}

{\color{black}
\begin{lemma}[{\cite[Theorem 3]{Ferguson}}]\label{HB}
Let $\{X_N\}$ be a sequence of random variables. If $\{X_N\}$ converges in distribution to $X$, then 
$$\mathbb{E}[g(X_N)]\to\mathbb{E}[g(X)]$$ for all bounded measurable functions $g$ such that $\mathbb{P}\{X\in C(g)\}=1$, where $C(g)=\{x\mid g \text{ is continuous at } x\}$ denotes the continuity set of $g$.
\end{lemma}
}

\begin{definition}[Empirical Spectral Distribution\cite{RMTbook}]\label{def:esd}
Consider an $N\times N$ Hermitian matrix $\bT_N$. Define its empirical spectral distribution (e.s.d.) $F^{\bT_N}$ to be the distribution function of the eigenvalues of $\bT_N$, i.e., for $x\in\R$, 
$$F^{\bT_N}(x)=\frac{1}{N}\sum_{j=1}^N1_{\{\lambda_j\leq x\}}(x),$$
where $\lambda_1,\dots,\lambda_N$ are the eigenvalues of $\bT_N$.
\end{definition}
\begin{lemma}[
{\cite[Section 3.2 and Section 7.1]{RMTbook}}]\label{MP} Consider a matrix $\bX\in\C^{K\times N}$ with i.i.d. entries following $\mathcal{C}\mathcal{N}\left(0,\frac{1}{N}\right)$. 
As $K,N\rightarrow \infty$ with $\frac{K}{N}\rightarrow c\in (0,1)$, the following results hold:
\begin{itemize}
\item[(i)] the  e.s.d. of $\bX\bX^\mathsf{H}$ converges weakly and almost surely to a distribution function $F$ with density given by:
$$p(x)=\frac{1}{2\pi cx}\sqrt{(x-a)_+(b-x)_+},$$
where $a=(1-\sqrt{c})^2, b=(1+\sqrt{c})^2$, and $(x)_+=\max\{x,0\}$;
\item[(ii)] the largest eigenvalue of $\bX\bX^\mathsf{H}$, denoted by $\lambda_{\max}$, satisfies 
$$\lambda_{\max}\xrightarrow{a.s.}(1+\sqrt{c})^2.$$
\end{itemize}
\end{lemma}


\vspace{5pt}
\subsection{Auxiliary Lemma}\label{Sec:aux}
This subsection introduces two auxiliary lemmas used in the main proof in Section \ref{Sec:main_proof}.

\begin{lemma}\label{calculatelemma}
Define $\widetilde{\bD}=\text{diag}(d_1,d_2,\dots,d_K)$, where $d_1,d_2,\dots,d_K$ are the nonzero  singular values of $\H$ (satisfying Assumption \ref{Ass:distribution}). Assume that $\sigma(\cdot)$ is a function {\color{black}that is continuous a.e. and bounded on any compact set of $(0,\infty)$}; $\sg_1\sim \mathcal{C}\mathcal{N}\left(\mathbf{0},\mathbf{I}_K\right)$, $\sg_2\sim \mathcal{C}\mathcal{N}\left(\mathbf{0},\mathbf{I}_K\right)$, and $\widetilde{\bD}$ are mutually independent.  Then as $N,K\to\infty$ and $\frac{N}{K}\rightarrow\gamma>1$,  it holds that 
\begin{itemize}
\item[(i)]$\dfrac{\sg_1^\mathsf{H}\sigma(\widetilde{\bD})\sg_1}{K}\xrightarrow{a.s.}\mathbb{E}[\sigma({d})];$
\vspace{0.3cm}
\item[(ii)]$\dfrac{\sg_1^\mathsf{H}\sigma(\widetilde{\bD})\sg_2}{K}\xrightarrow{a.s.}0,$
\end{itemize}
where $d=\sqrt{\lambda}$ and $\lambda$ follows the Marchenko-Pastur distribution, whose density is given in \eqref{MPdistribution}.
\end{lemma}
\begin{proof}
{\color{black}First, from the definition of $\sigma(\cdot)$ and (ii) of Lemma \ref{MP}, we know that for sufficiently large $N$ and $K$, there exists a constant $M>0$ such that
$$\sup_{1\leq i\leq K}|\sigma(d_i)|\leq M$$ with probability one.} 
Then according to {\color{black} the strong law of large numbers in} Lemma \ref{largenumber},  we have
\begin{equation*}\label{eq1}
\begin{aligned}
 &\frac{\sg_1^\mathsf{H}\sigma(\widetilde{\bD})\sg_1}{K}-\frac{1}{K}\sum_{i=1}^K \sigma(d_i)
=\frac{1}{K}\sum_{i=1}^K \sigma(d_i)|\sg_{1}[i]|^2-\frac{1}{K}\sum_{i=1}^K \sigma(d_i)
\xrightarrow{a.s.}0,
\end{aligned}
\end{equation*}
and
 \begin{equation*}
 \frac{\sg_1^\mathsf{H}\sigma(\widetilde{\bD})\sg_2}{K}=\frac{1}{K}\sum_{i=1}^K \sigma(d_i)\sg_{1}[i]^{\dagger}\sg_{2}[i]\xrightarrow{a.s.}0,\end{equation*}
 which immediately gives (ii) of Lemma \ref{calculatelemma}. Next we continue to prove (i) of Lemma \ref{calculatelemma}. Note that
 \begin{equation*}\label{eqd}
\frac{1}{K}\sum_{i=1}^K \sigma(d_i)=\frac{1}{K}\sum_{i=1}^K\sigma(\sqrt{\lambda_i}),
\end{equation*}
where $\lambda_1,\lambda_2,\dots,\lambda_K$ are the eigenvalues of $\H\H^\mathsf{H}$. Let {\color{black}$X_K$ be a random variable following the e.s.d. of $\H\H^\mathsf{H}$, then
 \begin{equation*}\label{eq2}
 \frac{1}{K}\sum_{i=1}^K \sigma(\sqrt{\lambda_i})=\mathbb{E}\left[\sigma(\sqrt{X_K})\right].
 \end{equation*}}
Define
$$
g(x)=\left\{
\begin{aligned} \sigma&(\sqrt{x}),~~~~\text{if }x\in[0,(1+\sqrt{c})^2+1];\\
&0,~~\qquad\qquad\qquad\text{otherwise},
\end{aligned}\right.$$where $c=\frac{1}{\gamma}$. 
{\color{black}Then $g(\cdot)$ is continuous a.e. and bounded}. It follows from (ii) of Lemma  \ref{MP} that with probability one,  the largest eigenvalue of $\H\H^\mathsf{H}$ is bounded by $(1+\sqrt{c})^2+1$ for sufficiently large $K$, which implies that
{\color{black}$$\mathbb{E}\left[\sigma(\sqrt{X_K})\right]\rightarrow\mathbb{E}\left[g(X_K)\right].$$
Let $\lambda$ be a random variable following the Marchenko-Pastur distribution. Then we have  $\mathbb{P}\{\lambda\in C(g)\}=1$ since $\lambda$ is a continuous random variable and $g$ is continuous a.e., where the definition of $C(g)$ is given in Lemma \ref{HB}. 
Therefore, according to (i) of Lemma \ref{MP}, we can apply Lemma \ref{HB} to $\{X_K\}$ and $\lambda$ to obtain 
\begin{equation*}\label{eq3}
\mathbb{E}[g(X_K)]\rightarrow\mathbb{E}[g(\lambda)]= \mathbb{E}\left[\sigma(\sqrt{\lambda})\right].
\end{equation*} } Combining the above discussions, we get the desired result
$$\frac{\sg_1^\mathsf{H}\sigma(\widetilde{\bD})\sg_1}{K}\xlongrightarrow{a.s.} \mathbb{E}\left[\sigma(\sqrt{\lambda})\right]=\mathbb{E}\left[\sigma(d)\right].$$ The proof is completed.
\end{proof}

\vspace{5pt}


\begin{lemma}\label{Lemma10}
Assume that $\sz\sim\mathcal{C}\mathcal{N}(\mathbf{0},\mathbf{I}_N)$, $\alpha\xrightarrow{a.s.} \bar{\alpha}$, and $q(\cdot)$ satisfies Assumption \ref{Ass:q}. Then, 
$$\frac{\sz^\mathsf{H}q(\alpha\sz)}{N}\xrightarrow{a.s.}\mathbb{E}\left[\sZ^\dagger q(\bar{\alpha} \sZ)\right],~~\text{where }\sZ\sim\mathcal{C}\mathcal{N}\left({0},{1}\right).$$
\end{lemma}
\begin{proof}
Note that 
$$
\begin{aligned}
\frac{\sz^\mathsf{H}q(\alpha \sz)}{N}
=&\frac{1}{N}\sum_{i=1}^Nz_i^\dagger q(\alpha z_i)\\
=&\frac{1}{N}\sum_{i=1}^N\RR(z_i)\RR(q(\alpha z_i))+\frac{1}{N}\sum_{i=1}^N\I(z_i)\I(q(\alpha z_i))\\
~&+j\frac{1}{N}\sum_{i=1}^N\RR(z_i)\I(q(\alpha z_i))-j\frac{1}{N}\sum_{i=1}^N\I(z_i)\RR(q(\alpha z_i)).
\end{aligned}
$$
To show the almost sure convergence, it suffices to show that each of the above four terms converges almost surely to a constant. Next, we will  prove that 
$$\frac{1}{N}\sum_{i=1}^N\RR(z_i)\RR(q(\alpha z_i))\xrightarrow{a.s.}\mathbb{E}\left[\RR(\sZ)\RR(q(\bar{\alpha} \sZ))\right],$$where $\sZ\sim\mathcal{C}\mathcal{N}(0,1),$ and the convergence of the other three terms can be proved using similar arguments.
Let $g(x)=\RR(q(x))$. For any given $\epsilon>0$, we construct the following lower and upper bounds of $g$:
$$l_\epsilon(x)=\inf_y\left\{g(y)+\frac{1}{\epsilon}|y-x|\right\},$$
$$u_\epsilon(x)=\sup_y\left\{g(y)-\frac{1}{\epsilon}|y-x|\right\}.$$
According to \cite[page 15]{Ferguson}, $l_\epsilon$ and $u_\epsilon$ satisfy
\begin{enumerate}
\item[(i)] $l_\epsilon(x)\leq g(x)\leq u_\epsilon(x), \forall\, x\in\C$.
\item[(ii)] $l_\epsilon$ and $u_\epsilon$ are Lipschitz continuous with Lipschitz constant $\frac{1}{\epsilon}$, i.e., $|l_\epsilon(x_1)-l_\epsilon(x_2)|\leq \frac{1}{\epsilon}|x_1-x_2|$ and $|u_\epsilon(x_1)-u_\epsilon(x_2)|\leq \frac{1}{\epsilon}|x_1-x_2|$.
\item[(iii)] $l_\epsilon(x)$ and  $u_\epsilon(x)$ are bounded, i.e., $\exists\, B>0$ not depending on $\epsilon$ such that $|l_\epsilon(x)|\leq B$ and $|u_\epsilon(x)|\leq B$, $\forall\, x\in\mathbb{C}$.
\item[(iv)] If $x$ is a continuous point of $g(\cdot)$, then $\displaystyle\lim_{\epsilon\downarrow 0}l_\epsilon(x)=g(x)=\lim_{\epsilon\downarrow 0} u_\epsilon(x)$.   
\end{enumerate}
Denote $G(x,\alpha)=\RR(x)\RR(q(\alpha x))=\RR(x)g(\alpha x)$. Based on the above upper and lower bounds of $g(x)$, we can upper and lower bound $G(x,\alpha)$ as 
$$G(x,\alpha)\leq \RR(x)_+u_\epsilon(\alpha x)+\RR(x)_-l_\epsilon(\alpha x)\triangleq U_\epsilon(x,\alpha)$$
and 
$$G(x,\alpha)\geq \RR(x)_+l_\epsilon(\alpha x)+\RR(x)_-u_\epsilon(\alpha x)\triangleq L_\epsilon(x,\alpha),$$
where $x_+=\max\{x,0\}$ and $x_-=\min\{x,0\}$.
It follows that for any $\epsilon>0$,
\begin{equation}\label{boundg}
\frac{1}{N}\sum_{i=1}^N L_\epsilon(z_i,\alpha)\leq \frac{1}{N}\sum_{i=1}^N G(z_i,\alpha)\leq \frac{1}{N}\sum_{i=1}^N U_\epsilon(z_i,\alpha).
\end{equation}
To show that $\frac{1}{N}\sum_{i=1}^N G(z_i,\alpha)\xrightarrow{a.s.}\mathbb{E}[G(\sZ,\bar{\alpha})]$, we will first prove that for a given $\epsilon>0$, 
\begin{subequations}\label{condition1}
\begin{align}
\frac{1}{N}\sum_{i=1}^N L_\epsilon(z_i,\alpha)&\xrightarrow{a.s.}\mathbb{E}[L_\epsilon(\sZ,\bar{\alpha})],\\
\frac{1}{N}\sum_{i=1}^N U_\epsilon(z_i,\alpha)&\xrightarrow{a.s.}\mathbb{E}[U_\epsilon(\sZ,\bar{\alpha})],
\end{align}
\end{subequations}
which, together with \eqref{boundg}, implies 
\begin{subequations}\label{condition2}
\begin{equation}\label{condition1_2}
\liminf_{N\to\infty} \frac{1}{N}\sum_{i=1}^N G(z_i,\alpha)\geq \mathbb{E}[L_\epsilon(\sZ,\bar{\alpha})]~~ a.s.
\end{equation} 
\begin{equation}\label{condition1_3}
\limsup_{N\to\infty}\frac{1}{N}\sum_{i=1}^N G(z_i,\alpha)\leq \mathbb{E}[U_\epsilon(\sZ,\bar{\alpha})]~~ a.s.\end{equation}
\end{subequations}
Then we will show that
\begin{equation}\label{epsilon0}
\lim_{\epsilon\downarrow 0}\mathbb{E}[L_\epsilon(\sZ,\bar{\alpha})]=\lim_{\epsilon\downarrow 0}\mathbb{E}[U_\epsilon(\sZ,\bar{\alpha})]=\mathbb{E}[G(\sZ,\bar{\alpha})].
\end{equation}
Letting $\epsilon\downarrow0$ in \eqref{condition2} and using \eqref{epsilon0} give the desired result:
\[
\frac{1}{N}\sum_{i=1}^N G(z_i,\alpha)\xrightarrow{a.s.}\mathbb{E}[G(\sZ,\bar{\alpha})].
\]

It remains to prove \eqref{condition1} and \eqref{epsilon0}. We will only provide the proof of \eqref{condition1} and \eqref{epsilon0} for the lower bound $L_\epsilon$ {\red since} the upper bound can be proved in exactly the same way. {\red First,}
$$
\begin{aligned}
&\left|\frac{1}{N}\sum_{i=1}^N L_\epsilon(z_i,\alpha)-\mathbb{E}[L_\epsilon(\sZ,\bar{\alpha})]\right|\\
\leq\,& \left|\frac{1}{N}\sum_{i=1}^N L_\epsilon(z_i,\alpha)-\frac{1}{N}\sum_{i=1}^N L_\epsilon(z_i,\bar{\alpha})\right|+\left|\frac{1}{N}\sum_{i=1}^N L_\epsilon(z_i,\bar{\alpha})-\mathbb{E}[L_\epsilon(\sZ,\bar{\alpha})]\right|.
\end{aligned}
$$
It follows immediately from the strong law of large numbers (i.e., Lemma \ref{largenumber}) that the second term tends to zero almost surely. For the first term, we have
$$
\begin{aligned}
\left|\frac{1}{N}\sum_{i=1}^N L_\epsilon(z_i,\alpha)-\frac{1}{N}\sum_{i=1}^N L_\epsilon(z_i,\bar{\alpha})\right|
\leq&\frac{1}{N}\sum_{i=1}^N\left|L_\epsilon(z_i,\alpha)-L_\epsilon(z_i,\bar{\alpha})\right|\\
\leq &\frac{1}{N}\sum_{i=1}^N |\RR(z_i)|\,\frac{|z_i(\alpha-\bar{\alpha})|}{\epsilon}\\
\leq &\frac{|\alpha-\bar{\alpha}|}{\epsilon}\,\frac{1}{N}\sum_{i=1}^N|z_i|^2\xrightarrow{a.s.} 0,
\end{aligned}$$
where the second inequality holds due to Lipschitz continuity (i.e., property (ii))) of $l_\epsilon(x)$ and $u_\epsilon(x)$, and the almost convergence is due to the assumption $\alpha \xrightarrow{a.s.} \bar \alpha,$ Lemma \ref{largenumber}, and Lemma \ref{as}.
This proves \eqref{condition1}. Since $l_\epsilon$ and $u_\epsilon$ are bounded (i.e., property (iii))), we have
$$|L_\epsilon(x,\bar{\alpha})|\leq B|\RR(x)|.$$
According to the dominated convergence theorem \cite{probability}, 
$$\lim_{\epsilon\downarrow 0}\mathbb{E}[ L_\epsilon(\sZ,\bar{\alpha})]=\mathbb{E}\left[\lim_{\epsilon\downarrow 0} L_\epsilon(\sZ,\bar{\alpha})\right].$$
Property (iv) of $l_\epsilon$ implies that
$\lim_{\epsilon\downarrow 0} L_\epsilon(x,\bar{\alpha})=G(x,\bar{\alpha}),$~if $\bar{\alpha} x$  is a continuous point of $g(x)$, i.e., $\lim_{\epsilon\downarrow 0} L_\epsilon(x,\bar{\alpha})=G(x,\bar{\alpha})$ almost everywhere. Therefore, 
$$\lim_{\epsilon\downarrow 0}\mathbb{E}[ L_\epsilon(\sZ,\bar{\alpha})]=\mathbb{E}\left[\lim_{\epsilon\downarrow 0} L_\epsilon(\sZ,\bar{\alpha})\right]=\mathbb{E}[G(\sZ,\bar{\alpha})],$$ which completes the proof.

\end{proof}

\subsection{Main proof of Theorem \ref{asy}}\label{Sec:main_proof}
To prove Theorem \ref{asy}, it suffices to prove (see Lemma \ref{as})
\[
T_\bs\xrightarrow{a.s.}\overline{T}_\bs\quad\text{and}\quad T_\sg\xrightarrow{a.s.}\overline{T}_\sg,
\]
where
$$\begin{aligned}T_\s=&\frac{\sg_1^\mathsf{H}\{C_1\bD\hat{\bs}_1+C_2\bD \bB(\hat{\bs}_1)\sz_2[2:N]\}}{\|\sg_1\|\|\s\|}-T_\sg\,\frac{(\bR(\bs)^{-1}\sg_2)[1]}{\|\bs\|},\\
T_\sg=&\frac{\|\bB(\sg_1)^\mathsf{H}\{C_1\bD \hat{\bs}_1+C_2\bD \bB(\hat{\bs}_1)\sz_2[2:N]\}\|}{\|(\bR(\bs)^{-1}\sg_2)[2:K]\|},
\end{aligned}$$
and
\begin{equation*}
\begin{split}
\overline{T}_\bs&=\overline{C}_1\,\mathbb{E}[{d}\, f({d})],\\
\overline{T}_\sg&=\sqrt{\sigma_s^2\,|\overline{C}_1|^2\, \text{\normalfont{var}}[{d}\, f({d})]+\overline{C}_2^2}.
\end{split} 
\end{equation*}
In the above {\red equations}, 
\begin{equation*}\label{Eqn:C1C2_repeat}
\begin{aligned}
C_1&=\frac{\sz_1^\mathsf{H}q\left(\frac{\|\hat{\s}_1\|}{\|\sz_1\|}\,\sz_1\right)}{\|\hat{\bs}_1\|\|\sz_1\|},~C_2=\frac{\left\|\bB(\sz_1)^\mathsf{H}q\left(\frac{\|\hat{\s}_1\|}{\|\sz_1\|}\,\sz_1\right)\right\|}{\|\sz_2[2:N]\|},\\
\overline{C}_1&= \frac{\mathbb{E}[\sZ^\dagger q(\bar{\alpha} \sZ)]}{\bar{\alpha}},\,~~~\overline{C}_2=\sqrt{\mathbb{E}[|q(\bar{\alpha} \sZ)|^2]-|\mathbb{E}[\sZ^\dagger q(\bar{\alpha} \sZ)]|^2}.
\end{aligned}
\end{equation*}
Further, $\hat{\bs}_1=\frac{\|\s\|}{\|\sg_1\|} f(\bD)^\mathsf{T}\sg_1$ and 
$\bar{\alpha}=\sqrt{\frac{\sigma_s^2\,\mathbb{E}[f^2(d)]}{\gamma}}$.



We start with analyzing the limiting values of ${C}_1$ and ${C}_2$. From {\color{black} the strong law of large numbers in } Lemma \ref{largenumber}, we have 
\begin{equation}\label{eq4}
\frac{\|\sz_1\|^2}{N}\xrightarrow{a.s.} 1,\quad\frac{\|\sg_1\|^2}{K}\xrightarrow{a.s.} 1,~\frac{\|\s\|^2}{K}\xrightarrow{a.s.} \sigma_s^2.
\end{equation}
Furthermore, applying Lemma \ref{calculatelemma} in Section \ref{Sec:aux} with $\sigma(x)=f^2(x)$ yields
\begin{equation}\label{eq5}
\frac{\|f(\bD)^\mathsf{T}\sg_1\|^2}{K}=\frac{\sg_1^\mathsf{H}f(\bD)f(\bD)^\mathsf{T}\sg_1}{K}\xrightarrow{a.s.} \mathbb{E}[f^2({d})].
\end{equation}
It follows immediately from  \eqref{eq4}, \eqref{eq5}, and Lemma \ref{as} that 
\begin{equation}\label{alpha_1}
\begin{aligned}
\alpha:=\frac{\|\hat{\s}_1\|}{\|\sz_1\|}=\frac{\|\s\|}{\|\sg_1\|}\frac{\|f(\bD)^\mathsf{T}\sg_1\|}{\|\sz_1\|}\xrightarrow{a.s.}\sqrt{\frac{\mathbb{E}[f^2(d)]\, \sigma_s^2}{\gamma}}=\bar{\alpha}.
\end{aligned}
\end{equation}
Note that $C_1$ can be expressed as 
$$C_1=\frac{\sz_1^\mathsf{H}q\left(\frac{\|\hat{\s}_1\|}{\|\sz_1\|}\sz_1\right)}{\|\hat{\bs}_1\|\|\sz_1\|}=\frac{\sz_1^\mathsf{H}q(\alpha \sz_1)}{\alpha \|\sz_1\|^2}.$$
{\color{black} According to} Lemma \ref{Lemma10} in Section \ref{Sec:aux} {\color{black}and noting that $\alpha\xrightarrow{a.s.} \bar{\alpha}$,} we have 
\begin{equation}\label{eq6}
\frac{\sz_1^\mathsf{H}q(\alpha \sz_1)}{N}\xrightarrow{a.s.}\mathbb{E}[\sZ^\dagger q(\bar{\alpha} \sZ)],
\end{equation}
where $\sZ\sim\mathcal{C}\mathcal{N}(0,1)$. 
This, together with \eqref{eq4}, \eqref{alpha_1}, and Lemma \ref{as}, proves the convergence of $C_1$:
$$C_1=\frac{\sz_1^\mathsf{H}q(\alpha \sz_1)}{\alpha \|\sz_1\|^2}\xrightarrow{a.s.}\frac{\mathbb{E}[\sZ^\dagger q(\bar{\alpha} \sZ)]}{\bar{\alpha}}=\overline{C}_1.$$
To show the convergence of 
$C_2=\frac{\|\bB(\sz_1)^\mathsf{H}q(\alpha \sz_1)\|}{\|\sz_2[2:N]\|}$, we {\color{black} first} note that
\begin{equation}\label{Eqn:C2_temp}
\begin{split}
&\frac{\left\|\bB(\sz_1)^\mathsf{H}q\left(\alpha\sz_1\right)\right\|^2}{N}\\
=\,&\frac{\left\|\left(
\begin{smallmatrix} \frac{\sz_1^\mathsf{H}}{\|\sz_1\|}\\\bB(\sz_1)^\mathsf{H}\end{smallmatrix}\right)q\left(\alpha\sz_1\right)\right\|^2-\left|\frac{\sz_1^\mathsf{H}q\left(\alpha\sz_1\right)}{\|\sz_1\|}\right|^2}{N}\\
=\,&\frac{\|q(\alpha\sz_1)\|^2}{N}-\left|\frac{\sz_1^\mathsf{H}q(\alpha\sz_1)}{N}\right|^2\frac{N}{\|\sz_1\|^2},
\end{split}
\end{equation}
where the second equality holds because  $\bR(\sz_1)^\mathsf{H}=\left(\begin{smallmatrix}\frac{\sz_1^\mathsf{H}}{\|\sz_1\|}\\\bB(\sz_1)^\mathsf{H}\\\end{smallmatrix}\right)$ and $\|\cdot\|$ is rotationally invariant.
Similar to Lemma \ref{Lemma10}, we can show that
\begin{equation*}\label{eq7}
\frac{\|q(\alpha\sz_1)\|^2}{N} \xrightarrow{a.s.} \mathbb{E}[|q(\bar{\alpha} \sZ)|^2].
\end{equation*}
Combining this with \eqref{eq4}, \eqref{eq6}, and \eqref{Eqn:C2_temp}, and noticing that $\frac{\|\sz_2[2:N]\|^2}{N}\xrightarrow{a.s.}1$, we have
\[
\begin{split}
C_2\xrightarrow{a.s.}\sqrt{\mathbb{E}[|q(\bar{\alpha} \sZ)|^2]-|\mathbb{E}[\sZ^\dagger q(\bar{\alpha} \sZ)]|^2}=\overline{C}_2.
\end{split}
\]

Next, we analyze the convergence of $T_\s$ and $T_\sg$.  We begin with $T_\s$: 
$$
\begin{aligned}
T_\bs=\,&\frac{\sg_1^\mathsf{H}\{C_1\bD\hat{\bs}_1+C_2\bD \bB(\hat{\bs}_1)\sz_2[2:N]\}}{\|\sg_1\|\|\s\|}-T_\sg\,\frac{(\bR(\s)^{-1}\sg_2)[1]}{\|\s\|}\\
=\,&C_1\,\frac{\sg_1^\mathsf{H}\bD f(\bD)^\mathsf{T}\sg_1}{\|\sg_1\|^2}+C_2\,\frac{\sg_1^\mathsf{H}\bD \bB(\hat{\bs}_1)\sz_2[2:N]}{\|\sg_1\|\|\s\|}-T_\sg\,\frac{(\bR(\s)^{-1}\sg_2)[1]}{\|\s\|}\\
:=\,&T_{11}+T_{12}+T_{13},
\end{aligned}$$
where the second equality uses the definition of $\hat{\s}_1$ in \eqref{Equiy_2}.
In what follows, we will show
\begin{subequations}
\begin{align}
T_{11} & \xrightarrow{a.s.} \overline{C}_1\,\mathbb{E}\left[{d}\, f({d})\right] \label{t11},\\
T_{12}& \xrightarrow{a.s.}0,\label{t12}\\
T_{13}& \xrightarrow{a.s.}0.\label{t13}
\end{align}
\end{subequations}
Substituting $\sigma(x)=xf(x)$ into Lemma \ref{calculatelemma}, we have   $$\dfrac{\sg_1^\mathsf{H}\bD f(\bD)^\mathsf{T}\sg_1}{K}\xrightarrow{a.s.} \mathbb{E}\left[{d}\, f({d})\right],$$ which, together with $\frac{\|\sg_1\|^2}{K}\xrightarrow{a.s.} 1$ and $C_1\xrightarrow{a.s.}\overline{C}_1$, implies \eqref{t11}. 
We next show that $T_{12}$ vanishes asymptotically. Recalling the definition $\bR(\hat{\bs}_1)=\left(\frac{\hat{\bs}_1}{\|\hat{\bs}_1\|} ~~\bB(\hat{\bs}_1)\right)$, we then have 
\begin{equation}\label{Rs1z2}
\bB(\hat{\s}_1)\sz_2[2:N]=\bR(\hat{\s}_1)\sz_2-\sz_2[1]\,\frac{\hat{\s}_1}{\|\hat{\s}_1\|}.
\end{equation}
Using \eqref{Rs1z2} and the definition of $\hat{\s}_1$ in \eqref{Equiy_2}, we can upper bound  $T_{12}$ by
\begin{equation*}
\begin{aligned}
|T_{12}|&={C}_2\,\left|\frac{\sg_1^\mathsf{H}\bD \bB(\hat{\bs}_1)\sz_2[2:N]}{\|\sg_1\|\|\s\|}\right|\\ &\leq {C}_2\,\left|\frac{\sg_1^\mathsf{H}\bD\bR(\hat{\s}_1)\sz_2}{\|\sg_1\|\|\s\|}\right|+{C}_2\,\left|\frac{ \sz_2[1]\,\sg_1^\mathsf{H}\bD\frac{\hat{\bs}_1}{\|\hat{\bs}_1\|}}{\|\sg_1\|\|\s\|}\right|\\
&={C}_2\,\left|\frac{\sg_1^\mathsf{H}\bD\bR(\hat{\s}_1)\sz_2}{\|\sg_1\|\|\s\|}\right|+{C}_2\,\left|\sz_2[1]\,\frac{\sg_1^\mathsf{H}\bD f(\bD)^\mathsf{T}\sg_1}{\|f(\bD)^\mathsf{T}\sg_1\|\|\sg_1\|\|\s\|}\right|\\
&\leq{C}_2 \left|\frac{\sg_1^\mathsf{H}\bD\bR(\hat{\s}_1)\sz_2}{\|\sg_1\|\|\s\|}\right|+{C}_2\left|\frac{\sz_2[1]}{\|\s\|}\,\|\bD\|\right|.
\end{aligned}
\end{equation*}
 Since both $\hat{\s}_1$ and  $\sg_1$ are independent of $\sz_2$, $\bR(\hat{\s}_1)\sz_2\overset{d}{=}\sz_2$ and $\bR(\hat{\s}_1)\sz_2$ is independent of $\sg_1$. It follows immediately from \eqref{eq4} and  Lemmas \ref{as} and \ref{calculatelemma} that
$$\frac{\sg_1^\mathsf{H}\bD\bR(\hat{\s}_1)\sz_2}{\|\sg_1\|\|\s\|}\xrightarrow{a.s.} 0.$$
In addition, since $\frac{\sz_2[1]}{\|\s\|}\xrightarrow{a.s.} 0$ and $\|\bD\|\xrightarrow{a.s.} 1+\sqrt{1/\gamma}$ {\color{black}(see (ii) of Lemma \ref{MP})}, we have $\left|\frac{\sz_2[1]}{\|\s\|}\,\|\bD\|\right|\xrightarrow{a.s.} 0.$ Noting that $C_2\xrightarrow{a.s.}\overline{C}_2$, we further have \eqref{t12}.
Finally, since $\frac{\bR(\s)^{-1}\sg_2[1]}{\|\s\|}\xrightarrow{a.s.}0$ and $T_\sg$ converges almost surely to a constant (which will be shown in the sequel), we have \eqref{t13}.
Combining the above discussions, we get $$T_\bs=T_{11}+T_{12}+T_{13}\xrightarrow{a.s.}\overline{C}_1\,\mathbb{E}\left[{d}\, f({d})\right].$$

We next show the almost sure convergence of $T_\sg$ to a constant.  Note that $\bR(\bs)^{-1}\sg_2\overset{d}{=}\sg_2$ due to the independence between  $\s$ and $\sg_2$. Hence the following holds for the denominator of $T_\sg$:
\begin{equation*}\label{eqn:Tg1}
\frac{\|\left(\bR(\bs)^{-1}\sg_2\right)[2:K]\|}{\sqrt{K}}\xrightarrow{a.s.}1.
\end{equation*}
For the numerator, using \eqref{Rs1z2} and the definition of $\hat{\s}_1$ in \eqref{Equiy_2},  we  have 
\begin{equation*}\label{eqn:Tg2}
\begin{aligned}
&\left|\frac{\|\bB(\sg_1)^\mathsf{H}\{{C}_1\bD\hat{\bs}_1+{C}_2\bD \bB(\hat{\bs}_1)\sz_2[2:N]\}\|}{\sqrt{K}}\right.
-\left.\frac{\|\bB(\sg_1)^\mathsf{H}\{{C}_1\bD\hat{\bs}_1+{C}_2\bD\bR(\hat{\s}_1)\sz_2\}\|}{\sqrt{K}}\right|\\
\leq \,&{C}_2\,\frac{\| \sz_2[1]\,\bB(\sg_1)^\mathsf{H}\bD\frac{\hat{\bs}_1}{\|\hat{\bs}_1\|}\|}{\sqrt{K}}\\
=\,&{C}_2 \,\frac{|\sz_2[1]|\, \|\bB(\sg_1)^\mathsf{H}\bD f(\bD)^\mathsf{T}\sg_1\|}{\sqrt{K}\,\|f(\bD)^\mathsf{T}\sg_1\|}\\
\leq\,& {C}_2\, \frac{|\sz_2[1]|\|\bB(\sg_1)\|\|\bD\|}{\sqrt{K}}\xrightarrow{a.s.} 0,
\end{aligned}
\end{equation*}
where the almost sure convergence uses the facts that $\|\bD\|\xrightarrow{a.s.}1+\sqrt{1/\gamma}$, $\|\bB(\sg_1)\|\leq 1$, and $C_2\xrightarrow{a.s.} \overline{C}_2$.
Furthermore, 
 \begin{equation}\label{T2}
\begin{aligned}
&\dfrac{\|\bB(\sg_1)^\mathsf{H}\{{C}_1\bD\hat{\bs}_1+{C}_2\bD\bR(\hat{\s}_1)\sz_2\}\|^2}{K}\\
\overset{d}{=}\,&\dfrac{\|\bB(\sg_1)^\mathsf{H}\{{C}_1\bD\hat{\bs}_1+{C}_2\bD\sz_2\}\|^2}{K}\\
=\,&\dfrac{\left\|\bR(\sg_1)^\mathsf{H}\{{C}_1\bD\hat{\bs}_1+{C}_2\bD\sz_2\}\right\|^2}{K}-\dfrac{\left\|\frac{\sg_1^\mathsf{H}}{\|\sg_1\|}\{{C}_1\bD\hat{\bs}_1+{C}_2\bD\sz_2\}\right\|^2}{K}\\
=\,&\frac{\|{C}_1\frac{\|\s\|}{\|\sg_1\|}\bD f(\bD)^\mathsf{T}\sg_1+{C}_2 \bD\sz_2\|^2}{K}-\frac{1}{K}\,\left|\frac{\sg_1^\mathsf{H}\{{C}_1\frac{\|\s\|}{\|\sg_1\|}\bD f(\bD)^\mathsf{T}\sg_1+{C}_2 \bD\sz_2\}}{\|\sg_1\|}\right|^2\\
:=\,&T_{21}-T_{22},
\end{aligned}
\end{equation}
where the first equality holds since $\sz_2$ is independent of $\sg_1,\bD,$ and $\hat{\s}_1$; the second equality is due to {\color{black} the definition of $\bB(\cdot)$ in} \eqref{B}; the third equality uses the definition of $\hat{\s}_1$ in \eqref{Equiy_2} and the rotational invariance of $\|\cdot\|$. The first term $T_{21}$ in \eqref{T2} can be expressed as  
$$
\begin{aligned}
T_{21}=&|{C}_1|^2\frac{\|\s\|^2}{\|\sg_1\|^2}\frac{\sg_1^\mathsf{H}f(\bD)\bD^\mathsf{T}\bD f(\bD)^\mathsf{T}\sg_1}{K}
+{C}_2^2\, \frac{\sz_2^\mathsf{H}\bD^\mathsf{T}\bD\sz_2}{K}+2\frac{\|\s\|}{\|\sg_1\|}\frac{\mathcal{R}\left(C_1C_2\sz_2^\mathsf{H}\bD^\mathsf{T}\bD f(\bD)^\mathsf{T}\sg_1\right)}{K}.\\
\end{aligned}
$$
Applying Lemma \ref{calculatelemma} with $\sigma(x)=x^2 f^2(x),$ $\sigma(x)=x^2$, and $\sigma(x)=x^2 f(x)$, and using the almost sure convergence of $C_1$ and $C_2$, we have
\begin{equation*}\label{T21}
T_{21}\xrightarrow{a.s.}\sigma_s^2\, |\overline{C}_1|^2\,\mathbb{E}[d^2 f^2({d})]+\overline{C}_2^2\, \mathbb{E}[d^2].\end{equation*}
Similarly, for the second term $T_{22}$ in \eqref{T2}, we have 
\begin{equation*}\label{T22}
\begin{aligned}
T_{22}&=\left|{C}_1\frac{\sqrt{K}\|\s\|}{\|\sg_1\|^2}\,\frac{\sg_1^\mathsf{H}\bD f(\bD)^\mathsf{T}\sg_1}{K}\right.+\left.{C}_2\frac{\sqrt{K}}{\|\sg_1\|}\,\frac{\sg_1^\mathsf{H}\bD \sz_2}{K}\right|^2\\
&\xrightarrow{a.s.} \sigma_s^2\,|\overline{C}_1|^2\,\mathbb{E}^2[{d}\, f({d})].
\end{aligned}
\end{equation*}
Combining the above discussions and noting that $\mathbb{E}[d^2]=1$ \cite{tulino2004random},  we can conclude that
$$T_\sg\xrightarrow{a.s.}\sqrt{\sigma_s^2\,|\overline{C}_1|^2\,\text{var}[d\, f({d})]+\overline{C}_2^2},$$
which completes our proof.

\section{Proof of Theorem \ref{corollary:sinrsep}}\label{proof:corollary1}

In this section, we give the proof of Theorem  \ref{corollary:sinrsep}. We first prove the convergence of $\widehat{\text{SINR}}_k$ and then show the convergence of $\widehat{\text{SEP}}_k(\beta)$.
\subsection{Proof of Convergence of $\widehat{\text{SINR}}_k$}
Note that $\mathbb{E}[|s_k|^2]=\sigma_s^2$ is a constant (see Assumption \ref{Ass:distribution}).  If we can show that $\mathbb{E}[|\hat{y}_k|^2]\rightarrow\mathbb{E}[|\bar{y}_k|^2]$ and $\mathbb{E}[s_k^\dagger\hat{y}_k]\rightarrow\mathbb{E}[s_k^\dagger\bar{y}_k]$, then according to the definition of the SINR in \eqref{def:sinr}, we have 
$$\widehat{\text{SINR}}_k\rightarrow\overline{\text{SINR}}.$$
In the following, we will only prove the convergence of $\mathbb{E}[|\hat{y}_k|^2]$, and the convergence of $\mathbb{E}[s_k^\dagger\hat{y}_k]$ can be shown similarly. 

Firstly, it follows from Theorem \ref{asy} that $\hat{y}_k\xrightarrow{a.s.}\bar{y}_k$, and thus   $$|\hat{y}_k|^2\xrightarrow{a.s.}|\bar{y}_k|^2.$$ To show the convergence of expectation, we only need to prove that 
$\{|\hat{y}_k|^2\}$ is uniformly integrable \cite{probability}. {\color{black}Note that for a sequence of random variables $\{X_N, N\geq 1\}$,  the boundedness of $\mathbb{E}[|X_N|^2]$  can imply uniform integrability of $\{X_N\}$ (see \cite[Section 9.5]{probability} for details about uniform integrability of random variables).  Therefore, it suffices to prove that the sequence 
$\{\mathbb{E}[|y_k|^4\}$ is bounded. According to \eqref{Equiy}, we have }
\begin{equation}\label{convergesnr1}
\begin{aligned}
|\hat{y}_k|^4&=\left|\eta T_\bs  s_k+\eta T_\sg \sg_2[k]+n_k\right|^4\\
&\leq 27|\eta T_\bs s_k|^4+27|\eta T_\sg\sg_2[k]|^4+27|n_k|^4\\
&\leq \frac{27}{2}(\eta^8|T_\bs|^8+|s_k|^8+\eta^8|T_\sg|^8+|\sg_2[k]|^8+2|n_k|^4),\\
\end{aligned}
\end{equation}
where the first inequality holds since $|a+b+c|^2\leq 3|a|^2+3|b|^2+3|c|^2$.
From Theorem \ref{asy} and Lemma \ref{as}, we have 
$$|T_\s|^8\xrightarrow{a.s.} |\overline{T}_\s|^8\quad\text{and}\quad |{T}_\sg|^8\xrightarrow{a.s.}|\overline{T}_\sg|^8,$$
where $\overline{T}_\s$ and $\overline{T}_\sg$ are both constants, and thus
\begin{equation}\label{convergesnr2}
\mathbb{E}[|T_\s|^8]\rightarrow|\overline{T}_\s|^8\quad\text{and}\quad \mathbb{E}[|\overline{T}_\sg|^8]\rightarrow|\overline{T}_\sg|^8.
\end{equation}
Taking expectations over both the left-hand and the right-hand sides of \eqref{convergesnr1} and using \eqref{convergesnr2}, we can conclude that 
\begin{equation*}\label{convergesnr3}
\sup_k\mathbb{E}[|\hat{y}_k|^4]<+\infty,
\end{equation*}
which completes the proof.
 
\subsection{Proof of Convergence of $\widehat{\text{SEP}}_k(\beta)$}
Given a constellation symbol $s$, we use $\mathcal{D}_s$ to denote its decision region, i.e., the region within which $s$ will be recovered. With this notation, $\widehat{\text{\normalfont{SEP}}}_k(\beta)$ can be expressed as
\begin{equation}\label{Eqn:SEP_expression}
\begin{aligned}
\widehat{\text{\normalfont{SEP}}}_k(\beta)=&\,1-\mathbb{P}\left(\beta\hat{y}_k\in \mathcal{D}_{s_k}\right)\\
=&\,1-\sum_{m=1}^M\mathbb{P}\left(s_k=s^{(m)}\right)\mathbb{P}\left(\beta\hat{y}_k\in \mathcal{D}_{s_k}\mid s_k=s^{(m)}\right)\\
=&\,1-\frac{1}{M}\sum_{m=1}^M\mathbb{P}\left(\beta\hat{y}_k \in \mathcal{D}_{s_{k}}\mid s_k=s^{(m)}\right),\\
\end{aligned}
\end{equation}
where the third equality holds since $s_k$ is uniformly drawn from $\mathcal{S}_M$.
Similarly, 
$$\overline{\text{\normalfont{SEP}}}_k(\beta)=1-\frac{1}{M}\sum_{m=1}^M\mathbb{P}\left(\beta\bar{y}_k\in \mathcal{D}_{s_k}\mid s_k=s^{(m)}\right).
$$ 
From Theorem \ref{asy}, we have  $(s_k,\hat{y}_k)\xrightarrow{a.s.}(s_k,\bar{y}_k)$. Hence, the joint distribution of $(s_k,\hat{y}_k)$ converges weakly to that of $(s_k,\bar{y}_k)$.
 If we can further show 
\begin{equation}\label{sepprove1}\mathbb{P}\left(\beta\bar{y}_k\in\delta \mathcal{D}_{s_k}\mid s_k=s^{(m)}\right)=0,
\end{equation} where $\delta \mathcal{D}_{s_k}$ denotes the boundary of $\mathcal{D}_{s_k}$, then according to  \cite[Lemma 2.2]{statistics},
$$ \mathbb{P}\left(\beta\hat{y}_k \in \mathcal{D}_{s_{k}}\mid s_k=s^{(m)}\right)\rightarrow\mathbb{P}\left(\beta\bar{y}_k \in \mathcal{D}_{s_k}\mid s_k=s^{(m)}\right),$$ and hence $\widehat{\text{\normalfont{SEP}}}_k(\beta)\rightarrow \overline{\text{\normalfont{SEP}}}_k(\beta)=\overline{\text{SEP}}(\beta)$. 

We next show \eqref{sepprove1}. When nearest-neighbor decoding is employed as assumed in this paper, the decision region of $s$ can be expressed as\footnote{A minor problem is that the decision regions given here are all open sets and cannot cover the whole complex space. This is not an essential problem: if $r$ lies on the boundary of two decision regions, we can just randomly choose one of the corresponding constellation point as the recovered symbol.}
$$\mathcal{D}_s=\left\{r\mid |r-s|^2< \left|r-s^{(i)}\right|^2,~ \forall \, s^{(i)}\in\mathcal{S}_M,~s^{(i)}\neq s\right\},$$
or equivalently,
\begin{equation}\label{eqn:Ds}
\begin{aligned}
\mathcal{D}_s=\left\{r\mid2\RR((s^{(i)}-s)^\dagger r)+|s|^2-|s^{(i)}|^2<0, ~\forall\, s^{(i)}\in\mathcal{S}_M,~s^{(i)}\neq s\right\}.
\end{aligned}\end{equation}
It follows that the boundary of $\mathcal{D}_s$ can be expressed as  
\begin{equation*}\label{boundary}\delta \mathcal{D}_s=\left\{r\mid\max_{i, s^{(i)}\neq s}\left\{2\RR((s^{(i)}-s)^\dagger r)+|s|^2-|s^{(i)}|^2\right\}= 0\right\}.
\end{equation*}
Recall that $\beta\bar{y}_k=\beta(\eta\overline{T}_\s\, s+\eta\overline{T}_\sg\, \sg_2[k]+n_k)$. Then $$\begin{aligned}
&\mathbb{P}\left(\beta\bar{y}_k\in\delta \mathcal{D}_{s_k}\mid s_k=s^{(m)}\right)\\
=\,&\mathbb{P}\left(\beta\left(\eta\overline{T}_\s\, s^{(m)}+\eta\overline{T}_\sg\, \sg_2[k]+n_k\right)\in\delta \mathcal{D}_{s^{(m)}}\right)\\
=\,&\mathbb{P}\left(\max_{i\neq m}\left\{\RR\left(a_{m,i}^\dagger(\eta\overline{T}_\sg\, \sg_2[k]+n_k)\right)+C_{m,i}\right\}=0\right),
\end{aligned}
$$
where $a_{m,i}=2\beta(s^{(i)}-s^{(m)})$ and $C_{m,i}=2\RR(\eta\overline{T}_\s a_{m,i}^\dagger  s^{(m)})+|s^{(m)}|^2-|s^{(i)}|^2$ are both constants. 
The last probability can further be upper bounded as $$
\begin{aligned}
&\mathbb{P}\left(\max_{i\neq m}\left\{\RR\left(a_{m,i}^\dagger(\eta\overline{T}_\sg\, \sg_2[k]+n_k)\right)+C_{m,i}\right\}=0\right)
\leq \sum_{i\neq m}\mathbb{P}\left(\RR\left(a_{m,i}^\dagger(\eta\overline{T}_\sg\, \sg_2[k]+n_k)\right)+C_{m,i}=0\right).
\end{aligned}$$
Since $\eta\overline{T}_\sg\, \sg_2[k]+n_k$ is Gaussian,   each term of the right-hand side of the above inequality is $0$,  which further implies that $\mathbb{P}\left(\beta\bar{y}_k\in\delta D_{s_k}\mid s_k=s^{(m)}\right)=0$. This completes our proof. 


\section{Equivalence of maximizing the Asymptotic SINR and {\color{black}minimizing} the Asymptotic SEP}\label{App:SINR_SEP}
For the scalar asymptotic model \eqref{scalarmodel} with precoding factor $\beta$ in \eqref{factor}, the received signal is 
\begin{equation}\label{eqn:r=s+n}
\beta\bar{y}=s+\frac{\overline{T}_\s^\dagger}{\eta|\overline{T}_\s|^2}(\eta\overline{T}_\sg\, \sg+n)\triangleq s+\bar{n},
\end{equation}
where $\bar{n}\sim\mathcal{C}\mathcal{N}\left(0,\frac{\eta^2\overline{T}_\sg^2+\sigma^2}{\eta^2|\overline{T}_\s|^2}\right)$ since $\sg\sim\mathcal{C}\mathcal{N}(0,1)$ and $n\sim\mathcal{C}\mathcal{N}(0,\sigma^2)$ are independent.
Following \eqref{Eqn:SEP_expression} in Appendix \ref{proof:corollary1}, we can express the asymptotic SEP as
$$\begin{aligned}
\overline{\text{\normalfont{SEP}}}=1-\frac{1}{M}\sum_{m=1}^M\mathbb{P}\left(\beta\bar{y}\in D_{s}\mid s=s^{(m)}\right).
\end{aligned}
$$
Each probability in the above summation can further be expressed as 
$$
\begin{aligned}
\mathbb{P}\left(\beta\bar{y}\in D_{s}\mid s=s^{(m)}\right)&=\mathbb{P}\left(\bar{n}\in D_{s^{(m)}}-s^{(m)}\right)\\
&=\mathbb{P}\left(\sZ\in \frac{\overline{\text{SINR}}}{\sigma_s^2} \left(D_{s^{(m)}}-s^{(m)}\right)\right),
\end{aligned}$$
where we have rewritten $\bar{n}$ as $\bar{n}= \frac{{{\sigma_s^2 }}}{\overline{\text{SINR}}}\, \sZ$ with $\sZ\sim\mathcal{C}\mathcal{N}(0,1)$.
From \eqref{eqn:Ds}, it is easy to check that $D_{s^{(m)}}-s^{(m)}$ is a polyhedron containing $0$, and thus the above probability is increasing  in  $\overline{\text{SINR}}$, i.e., $\overline{\text{SEP}}$ is a decreasing function of $\overline{\text{SINR}}$.
\section{Proof of Theorem \ref{optimalprecoder}}\label{proof:Th4}
{\color{black}In this section, we give the proof of Theorem \ref{optimalprecoder}, i.e., we prove that $(\bar{\alpha}^*, \eta^*, f^*)$ given in \eqref{Eqn:theorem_all} is  the optimal solution to the following problem:
\begin{equation}\label{prob:maxsinr2}
\begin{aligned}
\zeta^*:=\sup_{f,\eta>0,\bar{\alpha}>0}~&\frac{\mathbb{E}^2[d\, f(d)]}{\text{var}[d\, f(d)]+\phi(\bar{\alpha},\eta)\, \frac{\mathbb{E}[f^2(d)]}{\gamma} }\\
\text{s.t. }~~~~
&\eta^2\,\mathbb{E}[|q(\bar{\alpha} \sZ)|^2]\leq P_T,\\
&\mathbb{E}[f^2(d)]=\frac{\gamma}{\sigma_s^2}\, \bar{\alpha}^2.\\
\end{aligned}
\end{equation} }

We first note that $\phi^\ast=\infty$ (see the definition in \eqref{alpha}) is a pathological case where the SINR in \eqref{prob:maxsinr2} is identically zero\footnote{This is not the case for quantization functions used in practice.} and any $(\bar{\alpha},\eta,f)$ is optimal. This happens, e.g., when $q(\cdot)$ is a constant function. In the rest of the proof, we assume $\phi^\ast<\infty$. 

It is convenient to work with the inverse of the objective function in \eqref{prob:maxsinr2} and consider the following equivalent problem:
\begin{equation}\label{prob3}
\begin{aligned}
\Phi^*:=\inf_{\eta>0,\bar{\alpha}>0,f}\ &\frac{\mathbb{E}[d^2 f^2(d)]+\frac{\phi(\bar{\alpha},\eta)}{\gamma}\,\mathbb{E}[f^2(d)]}{\mathbb{E}^2[{d}\, f({d})]}-1\\[3pt]
\text{s.t. }\quad
&\eta^2\,\mathbb{E}[|q(\bar{\alpha} \sZ)|^2]\leq P_T,\\[3pt]
&\mathbb{E}[f^2(d)]=\frac{\gamma}{\sigma_s^2}\, \bar{\alpha}^2.
\end{aligned}
\end{equation}
We first ignore the constraints. By the Cauchy-Swarchz inequality, we can upper bound the denominator of the objective function of \eqref{prob3} by
$$\mathbb{E}^2[d\, f(d)]\leq \mathbb{E}\left[\left(d^2+\frac{\phi(\bar{\alpha},\eta)}{\gamma}\right) f^2(d)\right]\mathbb{E}\left[\frac{d^2}{d^2+\frac{\phi(\bar{\alpha},\eta)}{\gamma}}\right],$$
and thus 
\begin{equation*}\label{Eqn:Cauchy}
\begin{split}
\frac{\mathbb{E}[d^2  f^2(d)]+\frac{\phi(\bar{\alpha},\eta)}{\gamma}\,\mathbb{E}[f^2(d)]}{\mathbb{E}^2[{d}\, f({d})]} &=\frac{\mathbb{E}\left[\left(d^2+\frac{\phi(\bar{\alpha},\eta)}{\gamma}\right) f^2(d)\right]}{\mathbb{E}^2[d\, f(d)]}\\
&\ge \frac{1}{\mathbb{E}\left[\frac{d^2}{d^2+\phi(\bar{\alpha},\eta)/\gamma}\right]},
\end{split}
\end{equation*}
where the inequality holds with {\red equality}  when
\[
f(x):=\frac{x}{\tau \left(x^2+\frac{\phi(\bar{\alpha},\eta)}{\gamma}\right)},\quad\forall\,\tau\neq0.
\]
By the above inequality, the constrained infimum in \eqref{prob3} is further lower bounded by
\begin{equation}\label{Eqn:lower_bound}
\begin{split}
\Phi^*\ge\inf_{\eta>0, \bar{\alpha}> 0}\quad & \frac{1}{\mathbb{E}\left[\frac{d^2}{d^2+\phi(\bar{\alpha},\eta)/\gamma}\right]}-1\\[3pt]
\text{s.t. }\quad~~ & \eta^2\,\mathbb{E}[|q(\bar{\alpha} \sZ)|^2]\leq P_T,\\[3pt]
&\mathbb{E}[f^2(d)]=\frac{\gamma}{\sigma_s^2}\, \bar{\alpha}^2.
\end{split}
\end{equation}
Note that the objective function of \eqref{Eqn:lower_bound} is independent of $f(\cdot)$, and thus the second constraint can be removed. In addition, $\mathbb{E}[d^2/(d^2+x)]$ is decreasing in $x\in[0,\infty)$ and it is straightforward to check that $\phi(\bar{\alpha},\eta)\ge 0$ (see \eqref{Eqn:phi_def}).
Hence, the optimization problem in the right-hand side of \eqref{Eqn:lower_bound} is equivalent  to
\begin{equation}\label{Eqn:lower_bound2}
\begin{split}
\inf_{\eta>0, \bar{\alpha}> 0}\quad &\phi(\bar{\alpha},\eta)\\[3pt]
\text{s.t. }\quad~~ & \eta^2\,\mathbb{E}[|q(\bar{\alpha} \sZ)|^2]\leq P_T.\\[3pt]
\end{split}
\end{equation}
Since $\phi(\bar{\alpha},\eta)$ is decreasing in $\eta$ (see \eqref{Eqn:phi_def}), the power constraint in \eqref{Eqn:lower_bound2} must be satisfied with equality at the infimum, namely, 
 \begin{equation}\label{equality: eta}
 \eta=\sqrt{\frac{P_T}{\mathbb{E}[|q(\bar{\alpha} \sZ)|^2]}}.
 \end{equation}
Substituting \eqref{equality: eta} into \eqref{Eqn:lower_bound2} yields the optimization problem in \eqref{alpha}, whose optimal value is given by $\phi^*$. Hence, 
\[
\zeta^*=\frac{1}{\Phi^\ast}\le\frac{1}{1-\mathbb{E}\left[\frac{d^2}{d^2+\frac{\phi^{*}}{\gamma}}\right]}-1.
\]
The proof is complete by further verifying that  $(\bar{\alpha}^\ast,\eta^\ast,f^\ast)$ is feasible {\red for} \eqref{prob3} and $\Phi^*$ is attained at $(\bar{\alpha}^\ast,\eta^\ast,f^\ast)$.
%
%
%
\section*{Acknowledgment}
The authors would like to thank Professor Yue M. Lu from Harvard University {\color{black}and the anonymous reviewers for their} helpful comments on the paper. 
\bibliographystyle{IEEEtran}
\bibliography{IEEEabrv,reference_dce}
\end{document}